\documentclass[a4paper,12pt,reqno,english]{amsart}
\usepackage[utf8]{inputenc}
\usepackage[T1]{fontenc}
\usepackage{mathrsfs}
\usepackage{dsfont}
\usepackage{amsmath}
\usepackage{amssymb}
\usepackage{amsthm}
\usepackage{amsfonts}
\usepackage{amstext}
\usepackage{amsopn}
\usepackage{amsxtra}
\usepackage[foot]{amsaddr}
\usepackage{dsfont}
\usepackage{color}
\usepackage{esint}
\usepackage{graphicx}
\usepackage{thmtools}
\usepackage[margin=3cm]{geometry}
\usepackage[iso]{isodate}
\usepackage{xcolor}
\usepackage{todonotes}
\setuptodonotes{inline}
\usepackage[shortlabels]{enumitem}
\usepackage{bbm}
\usepackage{braket}
\usepackage{hyperref}
\usepackage[capitalise,sort]{cleveref}
\usepackage{lipsum}

\usepackage{chngcntr}
\counterwithin{figure}{section}
\counterwithin{equation}{section}
\counterwithin{table}{section}

\usepackage[style=alphabetic,
  sorting = nyt,
  sortcites=true,
  giveninits=true,
  date=year,
  isbn=false,
  maxbibnames=99,
  maxalphanames=6,
  backend=biber]{biblatex}
\DeclareFieldFormat*{titlecase}{\MakeSentenceCase{#1}}

\DeclareSourcemap{
  \maps[datatype=bibtex]{
    \map[overwrite]{
      \step[fieldsource=doi, final]
      \step[fieldset=url, null]
      \step[fieldset=eprint, null]
    }  
  }
}

\DeclareSourcemap{
  \maps[datatype=bibtex, overwrite]{
    \map{
      \step[fieldset=language, null]
      \step[fieldset=month, null]
      \step[fieldset=urldate, null]
    }
  }
}

\AtEveryBibitem{%
  \clearlist{language}%
  \clearlist{month}%
  \clearlist{urldate}%
  \clearfield{urldate}%
  \ifentrytype{inproceedings}
    {\clearlist{booktitle}%
     \clearlist{publisher}%
     \clearlist{location}}
}

\renewbibmacro{in:}{}

\DeclareFieldFormat[misc]{title}{\mkbibquote{#1}}
\DeclareFieldFormat[article,inproceedings]{volume}{\mkbibbold{#1}}
\DeclareFieldFormat[inproceedings]{series}{\mkbibitalic{#1}}
\DeclareFieldFormat{titlecase}{#1}
\renewcommand*{\bibfont}{\footnotesize}
\appto\bibfont{\setlength{\emergencystretch}{2em}}
\usepackage{seqsplit}

\theoremstyle{plain}

\newtheorem{theorem}{Theorem}[section]
\newtheorem{corollary}[theorem]{Corollary}
\newtheorem{lemma}[theorem]{Lemma}
\newtheorem{proposition}[theorem]{Proposition}

\theoremstyle{definition}

\newtheorem{remark}[theorem]{Remark}
\newtheorem{definition}[theorem]{Definition}
\newtheorem{assumption}[theorem]{Assumption}

\newtheorem{'definition}[theorem]{``Definition''}
\newtheorem{'proposition}[theorem]{``Proposition''}

\crefname{assumption}{Assumption}{Assumptions}

\newcommand{\R}{\mathbb{R}}
\newcommand{\Z}{\mathbb{Z}}
\newcommand{\C}{\mathbb{C}}
\newcommand{\N}{\mathbb{N}}

\newcommand{\bS}{\mathbb{S}}

\newcommand{\dps}{\displaystyle}
\newcommand{\ii}{\infty}

\newcommand{\cP}{\mathcal{P}}
\newcommand{\cT}{\mathcal{T}}
\newcommand{\cA}{\mathcal{A}}

\newcommand{\cV}{\mathcal{V}}

\newcommand{\cB}{\mathcal{B}}

\newcommand{\cE}{\mathcal{E}}

\newcommand{\cG}{\mathcal{G}}

\newcommand{\cH}{\mathcal{H}}

\newcommand\pscal[1]{{\ensuremath{\left\langle #1 \right\rangle}}}
\newcommand{\norm}[1]{ \left\| #1 \right\|}
\newcommand{\abs}[1]{{\left | #1 \right |}}
\newcommand{\nnorm}[1]{{\left\vert\kern-0.25ex\left\vert\kern-0.25ex\left\vert #1 
    \right\vert\kern-0.25ex\right\vert\kern-0.25ex\right\vert}}
\newcommand{\snnorm}[1]{{\vert\kern-0.25ex\vert\kern-0.25ex\vert #1 \vert\kern-0.25ex\vert\kern-0.25ex\vert}}

\let\Re\relax
\let\Im\relax
\DeclareMathOperator{\Im}{Im}
\DeclareMathOperator{\Re}{Re}

\renewcommand{\geq}{\geqslant}
\renewcommand{\leq}{\leqslant}
\renewcommand{\le}{\leqslant}
\renewcommand{\hat}{\widehat}
\renewcommand{\tilde}{\widetilde}

\newcommand{\eps}{\varepsilon}

\newcommand{\nn}{\nonumber}
\newcommand{\rd}{\mathrm{d}}
\newcommand{\dx}{\rd x}

\newcommand{\dt}{\rd t}

\title[The inverse problem of TDDFT on the torus]{The inverse problem of time-dependent density functional theory on the torus}

\date{2026-09-17.}

\author[A.~B.~Lauritsen]{Asbjørn~Bækgaard~Lauritsen$^*$}
\email{$^*$lauritsen@ceremade.dauphine.fr}

\author[M.~Lewin]{Mathieu~Lewin$^\dagger$}
\email{$^\dagger$mathieu.lewin@math.cnrs.fr}

\author[J.~Oldenburg]{Jakob~Oldenburg$^\ddagger$}
\email{$^\ddagger$jakob.oldenburg@cnrs.fr}

\address{$^*$$^\dagger$$^\ddagger$CEREMADE, CNRS, Université Paris-Dauphine, PSL Research University, Place de Lattre de Tassigny, 75016 Paris, France}

\allowdisplaybreaks

\begin{document}

\begin{abstract}
    We prove that, for any given   sufficiently regular   one-particle density, there exists a unique external potential for which the solution to the many-body Schrödinger equation has this prescribed density. This implies in particular that one can exactly reproduce the time-dependent density of an interacting system using non-interacting electrons, which is at the heart of time-dependent density functional theory.
    Our result covers the Coulomb interaction and 
    densities arising from extended nuclei.
\end{abstract}
\maketitle

\setcounter{tocdepth}{1}
\tableofcontents

\section{Introduction}

Density functional theory (DFT) is the standard method used in industrial and academic applications to approximate the solution to the many-particle Schrödinger equation~\cite{ParYan-94,CanFri-DFT}. It has enormous success for ground states~\cite{HohKoh-64,KohSha-PR-65,Burke-12,Jones-15}. 
The applicability of DFT relies on the construction of efficient empirical functionals, 
whose design 
largely relies on the solid theoretical foundations of DFT, particularly on the variational results of Lieb~\cite{Lieb-83b,LewLieSei-23_DFT}. 
For these constructions of functionals, it is important that DFT is an exact theory to guide the derivation of good approximations. On the other hand, time-dependent DFT (TDDFT) \cite{RunGro-84}, which is one of the only viable methods for describing the dynamics of many electrons in molecules and solids~\cite{MarUllNogRubBurGro-06,Ullrichs-11,MarMaiNogGroRub-12}, has not yet reached the level of ground state DFT in terms of predictive power~\cite{Fuks-16,Mai-JCP-16,LacMai-NPJCM-23}. TDDFT is also very poorly understood mathematically~\cite{FouLamLewSor-16,Penz-PhD,Lampart-21,DFT-22,MAQUI_TDDFT-26} and has so far only been rigorously justified for discrete systems with finitely many points \cite{MAQUI_finite-26,FarTok-12,LiUll-08,RugPenLee-15}.

In this paper we provide the \textbf{first  mathematical justification of TDDFT}, solving thereby a famous open problem in quantum chemistry, condensed-matter physics, and mathematical physics~\cite{MaiTodWooBur-10}. Our main result is the solution to an inverse problem that can be quickly described as follows. Given a trajectory of one-particle densities $\rho(t,x)$ and a compatible initial $N$-body wavefunction $\Psi_0(x_1,...,x_N)$, we prove that there exists a time-dependent external potential $V(t,x)$, unique up to a time-dependent constant, so that the associated $N$-body wavefunction $\Psi(t,x_1,...,x_N)$ solving Schrödinger's equation starting at $\Psi_0$ has the desired density $\rho_\Psi(t,x)=\rho(t,x)$. This inverse problem is the key result for the foundations of TDDFT and Kohn--Sham theory. It implies in particular that TDDFT is exact in the sense that no information is lost by describing a system only through the density $\rho(t)$ and the initial state $\Psi_0$, rather than through the time-dependent wavefunction $\Psi(t)$. Indeed, once the external potential $V(t)$ is found from $\rho(t)$, one can in turn obtain $\Psi(t)$ through the Schrödinger equation.

The existence of a solution to the inverse problem also implies that the time-dependent density of an interacting system in a known external potential $V_{\rm ext}(t,x)$ can be exactly reproduced using a system of $N$ non-interacting electrons, at the expense of changing the external potential in the corresponding Kohn--Sham equations. As we will explain below, this is at the heart of Kohn--Sham theory. 

The construction of the inverse potential $V(t,x)$ is a notoriously difficult mathematical problem, due to a phenomenon of loss of derivatives~\cite{MaiTodWooBur-10,LuMar-15,Penz-PhD,FouLamLewSor-16,Lampart-21,DFT-22}. Namely, if we express the sought-after $V(t,x)$ in terms of the unknown wavefunction $\Psi$ and the given density $\rho(t,x)$, we obtain a relation in the general form
\begin{equation}
-\nabla\cdot(\rho\nabla V)=-\frac{\partial_t^2\rho}2+\text{terms involving up to 4 derivatives of $\Psi$}
\label{eq:V_Psi_vague}
\end{equation}
called the \emph{force-balance} or \emph{van Leeuwen equation} in the physics literature \cite{vanLeeuwen-99}.
Inserting this in Schrödinger's equation leads to a highly nonlinear equation for $\Psi$~\cite{MaiTodWooBur-10}. The four derivatives in the formula \eqref{eq:V_Psi_vague} for $V$ prevent one from using standard fixed-point techniques for finding a solution $\Psi$ \cite{RugLee-11,RugGiePenLee-12,RugPenLee-15,Penz-PhD,TarUll-21}.

Van Leeuwen \cite{vanLeeuwen-99} wrote the right side of~\eqref{eq:V_Psi_vague} as the double divergence of a momentum-stress tensor containing only second order derivatives of $\Psi$, but did not analyze this structure further. In this paper, we exploit the crucial observation that, using the (desired) relation $\rho_\Psi=\rho$, the potential $V(t,x)$ can in fact be expressed in the form
\begin{multline}
-\nabla\cdot(\rho\nabla V)=-\frac{\partial_t^2\rho+\Delta^2\rho}2\\*
+\sum_{|\alpha|\leq 2}\partial^\alpha\left(\text{terms involving up to 1 derivative of $\Psi$}\right).
\label{eq:V_Psi_vague_1derivative}
\end{multline}
The main idea is that the outside derivatives in the second term on the right will be canceled after inversion of the elliptic operator on the left. The inverse of this operator is easily controlled when the density $\rho(t,x)$ is bounded away from zero everywhere, and this is why we work on the torus to avoid handling its decay at infinity.
Inserting this in the Schrödinger equation then leads to a nonlinear equation for $\Psi$, which morally has a loss of only one derivative, but is highly nonlocal. 

The rewriting with a loss of one derivative in~\eqref{eq:V_Psi_vague_1derivative} suggests finding solutions in the analytic class, by means of the Cauchy--Kovalevskaya \cite{Cauchy-1842,Kowalevsky-1875,Ovsyannikov-65,Ovsyannikov-71,Nirenberg-72,Nishida-77,BaoGou-77,BaoGou-77b}
or Nash--Moser \cite{Nash-56,Moser-66,Moser-66a,Zehnder-75,Zehnder-76} schemes. 
However, a further difficulty is that the wavefunction $\Psi$ cannot be analytic, even in an analytic external potential, due to the \textbf{singular Coulomb repulsion between the electrons}. The latter induces well-known cusps at two-electron coalescence, namely when $x_j=x_k$ for some $j\neq k$. Our main idea to overcome this problem is to only require $\Psi$ to be \textbf{analytic with respect to the center of mass coordinate} $x_1+\cdots+x_N$ and not each $x_j$ separately, which turns out to be sufficient to make $\rho_\Psi$ (and then $V$) analytic in space. A similar idea was recently used in \cite{HeaSob-22,HeaSob-23} for ground states, following earlier works \cite{Hunziker-86,FouHofOstHofOstSor-02,FouHofOstHofOstSor-04,FouHofOstHofOstSor-07,Jecko-10}. A crucial step of the proof is then to show that all the first-order derivatives of $\Psi$ appearing on the right of~\eqref{eq:V_Psi_vague_1derivative} can be controlled purely in terms of the derivative in the direction of the center of mass. 

As we mentioned above, the inverse problem solved in this paper is a key result in TDDFT that \textbf{justifies the use of Kohn--Sham equations to replace the exact many-body Schrödinger equation}. It implies that one can reproduce the exact one-particle density of the many-body Schrödinger solution $\Psi$ in a known external potential $V_{\rm ext}(t,x)$ using a system of $N$ ``fictitious non-interacting electrons'', at the expense of changing $V_{\rm ext}(t,x)$ into an appropriate $V(t,x)$. The non-interacting Kohn--Sham equations are much easier to simulate numerically, if we assume we know $V$. In practice we of course do not know $V$ and cannot define it through the unknown Schrödinger density because the goal is precisely to avoid solving the many-body equation! As we will explain, our inverse result in fact allows one to define \textbf{self-contained nonlinear Kohn--Sham equations} that do not explicitly involve the unknown Schrödinger density and are nevertheless guaranteed to reproduce it for the solution. This construction relies on the inverse interacting problem; the non-interacting one is not sufficient.

The present paper opens many new questions and we hope that it will stimulate more works on TDDFT and DFT in general. Important open problems are the existence of $V$ for ground states in the same (analytic) setting, or the extension to point nuclei. Future work will be devoted to deriving numerical algorithms with guaranteed convergence that can compute an approximation of $V$. The results in this paper were announced in~\cite{LauLewOld-V-short-26_ppt}. Simultaneously with the first announcement of this work, the  classical mean-field problem (where the underlying $N$-particle Schrödinger equation is replaced by the Vlasov--Poisson equation \cite{ChaFin-05,Manfredi-20}) was solved independently in \cite{Bouedec-26_ppt} using also a Cauchy--Kovalevskaya method.

\medskip

\paragraph{\textbf{The paper is organized as follows.}} In the next section we describe our two main results in detail. In Section \ref{sec.TDDFT} we explain how to construct the associated nonlinear Kohn--Sham equations based on our inverse potentials. The rest of the paper is devoted to the proof of our results. 
In Appendix~\ref{sec.notation.appendix} we keep an overview of notation used in the paper.

\medskip

\paragraph{\textbf{Acknowledgment.}} This work has benefited from French State support managed by ANR under the France 2030 program through the MaQuI CNRS Risky and High-Impact Research program (RI)$^2$ (grant agreement ANR-24-RRII-0001). We thank \'Eric Cancès, Théo Duez, Jari van Gog and Julien Toulouse for stimulating discussions within this project.

%%%%%%%%%%%%%%%%%%%%%%%%%%%%%%%%%%%%%%%%%%%%%%%%%%%%%%%%%%%%%%%%%%%%%%%%%%%%%%%%%%%%%%%%%%%%%
%%%%%%%%%%%%%%%%%%%%%%%%%%%%%%%%%%%%%%%%%%%%%%%%%%%%%%%%%%%%%%%%%%%%%%%%%%%%%%%%%%%%%%%%%%%%%
\section{Main results}
%%%%%%%%%%%%%%%%%%%%%%%%%%%%%%%%%%%%%%%%%%%%%%%%%%%%%%%%%%%%%%%%%%%%%%%%%%%%%%%%%%%%%%%%%%%%%
%%%%%%%%%%%%%%%%%%%%%%%%%%%%%%%%%%%%%%%%%%%%%%%%%%%%%%%%%%%%%%%%%%%%%%%%%%%%%%%%%%%%%%%%%%%%%
In this first section we state all our main results concerning the inverse problem. We consider $N$ interacting spinless fermions on the $D$-dimensional cube  $\Lambda = [0,1]^D$, with $D\in\{1,2,3\}$, with periodic boundary conditions (that is, the $D$-dimensional torus). We neglect spin for simplicity of notation; all our results are valid with spin if we fix the total density. The evolution of this system in a given external potential $V(t,x):\Lambda\to\R$ is described by the time-dependent linear Schrödinger equation
\begin{equation}
\left\{
\begin{aligned}
i\partial_t\Psi(t) & = \cH_{V(t)}\Psi(t), \\
\Psi(0) & =\Psi_0,
\end{aligned}
    \label{eq:Schrodinger}
\right.
\end{equation}
with the Hamiltonian 
\begin{equation}
    \cH_{V} := \sum_{j=1}^N -\Delta_{x_j} + \sum_{j=1}^N V(x_j) + \sum_{1 \leq j < k \leq N} V_{ee}(x_j-x_k), 
    \label{eq:H_V}
\end{equation}
where $V_{ee}$ is the periodic interaction between the particles. For simplicity of notation we work in a system of units where the mass of the particles is set to $m=1/2$, instead of the standard atomic units that have $-\Delta/2$ instead of $-\Delta$. 

\begin{assumption}
Throughout the paper we will always assume that $V_{ee}\in L^2(\Lambda)$ is a real-valued and even ($V_{ee}(-x) = V_{ee}(x)$) function. 
\end{assumption}
This assumption covers the usual periodic Coulomb potential (the Green's function of the Laplacian) when $D=3$, that admits the Fourier coefficients $\widehat{V_{ee}}(k)=4\pi/|k|^2$, $k\in 2\pi\Z^3\setminus\{0\}$. In~\eqref{eq:Schrodinger}, the initial condition $\Psi_0\in L^2(\Lambda^N)$ is assumed to be normalized and antisymmetric with respect to permutations of its variables $x_1,...,x_N$, a property that is preserved under the flow. In the following we denote by $L^2_a(\Lambda^N) = \bigwedge^N L^2(\Lambda)$ the subspace of antisymmetric functions. All our results would apply the same to bosons
or to fermions with spin.  
Everywhere in the paper, the external potential $V$ will be at least in $C^0([0,T];C^2(\Lambda))$, in which case it is well-known that~\eqref{eq:Schrodinger} admits a unique solution $\Psi\in C^1([0,T];L^2_a(\Lambda^N))\cap C^0([0,T];H^{2}(\Lambda^N))$ for an initial condition $\Psi_0\in H^2(\Lambda^N)$, see~\cite{ReeSim2,Yajima-87}.

Recall that the one-particle density of a normalized wavefunction $\Psi\in L^2_a(\Lambda^N)$ is defined by
\begin{equation}\label{eqn.def.density}
    \rho_\Psi(x)=N\idotsint_{\Lambda^{N-1}}|\Psi(x,x_2,\ldots,x_N)|^2\,\dx_2\cdots\dx_N,
\end{equation}
see~\cite{LieSei-09}. Our goal is to solve the inverse problem of finding the potential $V(t,x)$ to be inserted in~\eqref{eq:Schrodinger} so that the solution $\Psi$ has a prescribed density $\rho_{\Psi(t)}(x)=\rho(t,x)$. The first difficulty is to identify a set of physically relevant target densities that is insensitive to the interaction $V_{ee}$~\cite{FouLamLewSor-16}. This is because the goal of TDDFT is to reproduce an interacting system ($V_{ee}\neq0$) with $N$ non-interacting electrons ($V_{ee}\equiv0$).

In this paper we will work with \textbf{analytic-in-space external potentials} $V(t,x)$ corresponding for instance to moving extended nuclei (point nuclei are thus not allowed) and applied electric fields. The wavefunctions $\Psi$, however, cannot be analytic because of the possible singularities of the interaction $V_{ee}$. Instead, we will work with wavefunctions that are only assumed analytic in the direction of the center of mass $x_1 + \ldots + x_N$. Concretely we introduce the following

\begin{definition}\label{def.norms.sigma}
For $N\geq1$, $\Psi\in L^2_a(\Lambda^N)$ and $\sigma \geq 0$, we define the norms by
\begin{equation}\label{eq:normsigma}
\norm{\Psi}_{\sigma}^2
=
\sum_{p_1,\ldots,p_N\in 2\pi\Z^D} e^{2\sigma\abs{\sum_{j=1}^N p_j}} \left(1+ \sum_{j=1}^N p_j^2\right)^2 \big|\hat\Psi(p_1,\dots,p_N)\big|^2
=\norm{e^{\sigma \abs{\cP}}\Psi}_{H^2}^2, 
\end{equation}
where $\cP=\sum_{j=1}^N -i \nabla_{x_j}$ is the total momentum operator. (Our convention for the Fourier coefficients in $L^2(\Lambda)$ is $\hat f(p) = \int_{\Lambda} f(x) e^{-ip\cdot x} \,\rd x$ for $p\in 2\pi \Z^D$ and similarly for $\hat\Psi$.) 
For simplicity we do not emphasize the $N$-dependence in the notation. We denote the associated Banach space by $\cB_\sigma$. 
Further, to separate out the special case $N=1$, we write $\cA_\sigma$ when $N=1$, that is, 
\begin{equation*}
    \cA_\sigma 
    = \left\{ f\in L^2(\Lambda) : \norm{f}_\sigma < \infty \right\},
    \qquad 
    \cB_\sigma 
    = \left\{ \Psi\in L^2_a(\Lambda^N) : \norm{\Psi}_\sigma < \infty \right\}. 
\end{equation*}
\end{definition}

For $N=1$, the norm on $\cA_\sigma$ can be expressed as 
\begin{equation}\label{eq:normsigma_N1}
\norm{f}_{\sigma}^2
=
\sum_{p\in 2\pi\Z^D} e^{2\sigma|p|} \left(1+ p^2\right)^2 \big|\hat{f}(p)\big|^2
=\norm{e^{\sigma \sqrt{-\Delta}}f}_{H^2}^2. 
\end{equation}
It is well-known that functions in $\cA_\sigma$ are all real-analytic. In fact, their Taylor series converges uniformly on any ball of radius $0<R<\sigma$, as follows from the proof of \cref{lemma.1/f.analytic} below. Conversely, any real-analytic function with uniform radius of convergence $R$ over the torus belongs to $\cA_\sigma$ for all $\sigma<R$. 
A simple argument showing uniform convergence on any ball of radius $R < \sigma/\sqrt{D}$ is
\begin{equation}    
    \norm{\partial^\alpha f}_{L^\ii(\Lambda)}\leq C\norm{\partial^\alpha f}_{H^2(\Lambda)}\leq C\norm{f}_\sigma\max_{p\in 2\pi\Z^D}\big(\abs{p^{\alpha}}e^{-\sigma|p|}\big)
    \leq \frac{C\alpha!}{(\sigma/\sqrt{D})^{|\alpha|}}\norm{f}_{\sigma}.
    \label{eq:simple_estim_Asigma}
\end{equation}
Note that we used in~\eqref{eq:simple_estim_Asigma}  that $H^2(\Lambda)\subset C^0(\Lambda)$ in dimensions $D\leq 3$ by the Sobolev embedding theorem. Adding an $H^2$ weight in the norms~\eqref{eq:normsigma} and~\eqref{eq:normsigma_N1} is natural because this is the operator domain of the Laplacian, hence of the corresponding Schrödinger operator $\cH_V$ under our assumptions on $V$ and $V_{ee}$.

We emphasize once again that, for $N\geq2$, the functions in $\cB_\sigma$ need not be analytic. Our space $\cB_\sigma$ is inspired by a famous work of Hunziker~\cite{Hunziker-86} and several more recent works dealing with the analyticity of the density for ground states~\cite{FouHofOstHofOstSor-02,FouHofOstHofOstSor-04,FouHofOstHofOstSor-07,Jecko-10,HeaSob-22,HeaSob-23}. The idea is to exploit the translation-invariance of the Hamiltonian $\cH_0= \cH_{V=0}$ with $V$ removed, which is equivalent to $\cH_0$ commuting with the total momentum operator $\cP$. This suggests analyticity of the wavefunction with respect to $\cP$, that is, in the center of mass direction $x_1+\cdots+x_N$, if we assume that $V$ is itself analytic. A central property of this construction is that the analyticity of the one-particle density follows from the analyticity of $\Psi$ in the center of mass direction only. More precisely, we have $\rho_\Psi\in\cA_\sigma$ whenever $\Psi\in \cB_\sigma$ with the estimate
    $
    \norm{\rho_\Psi}_\sigma \leq  C \norm{\Psi}_\sigma^2,
    $
see Lemma~\ref{lem:Psi_to_density_B_sigma} below.

\subsection{Regularity of \texorpdfstring{$\Psi$ and $\rho$}{Psi and rho} in an analytic external potential \texorpdfstring{$V$}{V}}
Our first result states that the eigenstates of $\cH_V$ are in $\cB_\sigma$ for any given external potential $V\in\cA_\sigma$ and that the time-evolution determined by the Schrödinger equation~\eqref{eq:Schrodinger} preserves this regularity, even with an interaction $V_{ee}$ that is only in $L^2(\Lambda)$, which covers the Coulomb interaction in 3D. 

\begin{theorem}[Eigenstates and dynamics in an analytic external potential $V$]\label{thm.density.is.analytic}
Let $V_{ee}\in L^2(\Lambda)$ be real-valued and even and let $\sigma>0$.

\begin{enumerate}[leftmargin=*,label=(\roman*)]
\item \emph{(Eigenstates)} Let $V\in \cA_\sigma$ be a real-valued function. Any eigenfunction $\Psi$ of $\cH_V$ belongs to the many-body space $\cB_\sigma$ and hence $\rho_\Psi\in\cA_\sigma$ is real-analytic.

\smallskip
\item \emph{(Dynamics)} Let $T>0$, $V \in C^0([0,T];\cA_\sigma)$ be real-valued, and $\Psi_0\in\cB_\sigma$. Then the corresponding solution to the time-dependent Schrödinger equation~\eqref{eq:Schrodinger} satisfies $\Psi\in C^0([0,T],\cB_\sigma)$, hence $\rho_\Psi\in C^0([0,T]; \cA_\sigma)$ is real-analytic in space for any time. Moreover, for any given $\sigma' < \sigma$ we have $\rho_\Psi\in C^2([0,T];\cA_{\sigma'})$. 
\end{enumerate}
\end{theorem}

This result justifies working in the space $\cA_\sigma$ for $V$ and $\rho$ and in $\cB_\sigma$ for $\Psi$, for the inverse problem that we study below. 

We now make some important remarks on Theorem~\ref{thm.density.is.analytic}.

\begin{remark}[Comparison with the literature]
For eigenstates, the proof of the analyticity of the density was given in~\cite{FouHofOstHofOstSor-02,FouHofOstHofOstSor-04,FouHofOstHofOstSor-07,Jecko-10,HeaSob-22,HeaSob-23} for potentials $V$ and interactions $V_{ee}$ that are analytic away from finitely many points. Under the assumption that $V$ is analytic everywhere, we can handle any $V_{ee}$ in $L^2(\Lambda)$. Our proof is also much simpler and does not use any localization technique. It only exploits the analyticity of the wavefunction in the direction of the center of mass, which follows easily from the translation-invariance of $\cH_0$ and the assumed analyticity of $V$. The reason this suffices for the density is that any derivative can be expressed as
$$\nabla_{x_1} \rho_\Psi(x_1)=N\idotsint_{\Lambda^{N-1}}\left(\sum_{j=1}^N\nabla_{x_j}\right)|\Psi(x_1,\ldots,x_N)|^2\,\dx_2\cdots\dx_N$$
because the added terms have vanishing integral. A similar argument holds for the current or the kinetic energy density, but not for the one-particle density matrix~\cite{Cioslowski-20,HeaSob-22,HeaSob-23,FouSob-26_ppt}.

In the time-dependent case, the proof that $\rho_\Psi\in C^0([0,T];\cA_\sigma)$ is a simple application of Grönwall's lemma, where only derivatives of the analytic $V$ enter, since the rest of the Hamiltonian is translation-invariant. The proof that $\rho_\Psi$ is $C^2$ in time is much more involved, however.
\end{remark}

\begin{remark}[Holographic theorem]
The analyticity-in-space of the density $\rho_{\Psi(t)}$ of the solution to the time-dependent Schrödinger equation was conjectured in \cite{Zheng-07,Zheng-11} and is for the first time proved in our paper. This property is sometimes called the \emph{time-dependent holographic electron density theorem} and it plays an important role in TDDFT for open systems. It implies that the knowledge of $\rho$ in a subdomain at any time $t$ implies the knowledge of the whole density.     
\end{remark}

\begin{remark}[Dependence on $V$ and $\Psi_0$]
\label{rmk.smooth.dependence.direct.problem}
The solution $\Psi$ (and hence $\rho_\Psi$) depends continuously on the external potential $V$ and the initial condition $\Psi_0$. More precisely, the map $(\Psi_0,V) \mapsto \Psi$ is Lipschitz. 
This is easily seen by writing the solution (for short times) using a Duhamel formulation and iterating this construction. 
\end{remark}

\begin{remark}[Less regularity in time]
We stated the result with $V\in C^0([0,T];\cA_\sigma)$ for convenience. We can also take $V\in L^\ii([0,T];\cA_\sigma)$, allowing kinks in time, in which case we have similarly that $\Psi\in C^0([0,T];\cB_\sigma)$, $\rho_\Psi \in C^0([0,T]; \cA_\sigma)$, and $\rho_\Psi \in C^{1,1}([0,T]; \cA_{\sigma'})$, where $C^{1,1}$ is the standard Hölder space of continuously differentiable functions whose derivative is Lipschitz.     
\end{remark}

\begin{remark}[More regularity in time]\label{rmk:time-regularity}
If $V_{ee}$ is $C^\ii$, then $\rho$ will inherit the time-regularity of $V$. Namely, we will have $\rho\in C^{k+2}([0,T];\cA_{\sigma'})$ whenever $V\in C^{k}([0,T];\cA_{\sigma})$ and $\Psi_0\in H^{2k}\cap \cB_\sigma$~\cite{FouLamLewSor-16}. Note, however, that one should not expect $\rho$ to be analytic in time in general, even if $V$ is. An easy example is $N=1$ and $V\equiv0$ in which case $\rho$ is analytic in time if and only if $\Psi_0$ is in the domain of the inverse heat operator $e^{\tau(-\Delta)}$ for some $\tau>0$.

Finally, we mention that for a singular $V_{ee}$, we do not expect $\rho$ to be $C^3$ in time, even for an analytic external potential $V(t,x)$. 
\end{remark}

\begin{remark}[Regularity in time needed for the inverse problem]
For the inverse problem studied below, it is a crucial fact that $\rho$ is twice differentiable in time. This is already apparent from \eqref{eq:V_Psi_vague} and \eqref{eq:V_Psi_vague_1derivative} above. 
\end{remark}

\begin{remark}[Point nuclei]\label{rmk.point.nuclei}
It would be interesting to deal with moving point nuclei, corresponding to $V(t,x)=-\sum_{m=1}^Mz_mG(x-R_m(t))$ where $G$ is the periodic Green's function of the Laplacian and $R_m(t)$ is assumed smooth. In~\cite{Zheng-07,Zheng-11}, it was conjectured that  $\rho(t,x)$ is analytic away from the $R_m(t)$, as it is for ground states~\cite{FouHofOstHofOstSor-04, Jecko-10, HeaSob-22}. The main difficulty is to analyze the precise structure of the wavefunction when some electrons hit a nucleus and how this cusp propagates in time. When several electrons meet at a nucleus, the singularity of $V$ conspires with that of $V_{ee}$, leading to rather delicate singularities for $\Psi$. For an analytic external potential $V$ we could avoid studying the electronic cusps using only the analyticity in the center of mass direction, but the latter will break for point nuclei. 
\end{remark}

\begin{remark}[Regularity in time for $\Psi$]
The solution $\Psi$ to the Schrödinger equation \eqref{eq:Schrodinger} is not just continuous in time. In fact, \eqref{eq:Schrodinger} and $\Psi\in C^0([0,T];\cB_\sigma)$ implies that $\Psi \in C^1([0,T]; D(e^{\sigma|\cP|}))$ with $D(e^{\sigma |\cP|}) = \{ \Psi \in L^2_a(\Lambda^N): \norm{e^{\sigma |\cP|} \Psi}_{L^2} < \infty\}$ the domain of the (unbounded) operator $e^{\sigma |\cP|}$. 
\end{remark}

\begin{remark}[Case of $\R^D$]
We stated Theorem~\ref{thm.density.is.analytic} on the torus $\Lambda$ because this is what we need for the inverse problem below, but the exact same result holds on the whole space $\R^D$, provided the sums are simply replaced by integrals in the definitions~\eqref{eq:normsigma} and~\eqref{eq:normsigma_N1} of $\cB_\sigma$ and $\cA_\sigma$, respectively. The space $\cA_\sigma$ does not contain smeared Coulomb potentials of the form $V=\mu\ast|x|^{-1}$ with analytic $\mu$, due to the slow decay in space at infinity. Such potentials are, however, integrable in Fourier, even with an analytic weight, which is enough to conclude analyticity of $\Psi$ in the direction of the center of mass and of its density in this case as well.
\end{remark}

\subsection{The inverse problem}
We can now turn to the main results of the paper concerning the inverse problem. We give ourselves an initial state $\Psi_0\in\cB_\sigma$ (for instance an eigenstate by Theorem~\ref{thm.density.is.analytic}) and a path of densities $\rho(t)$ in the analytic space $\cA_\sigma$, and we show that there exists a unique associated (time-dependent) potential $V(t)$ reproducing this density.

\begin{theorem}[The inverse problem]\label{thm.main.V}
Let $V_{ee}\in L^2(\Lambda)$ be real-valued and even. 
Let $T>0$, $\sigma > 0$ and $\Psi_0\in \cB_\sigma$ be a normalized wavefunction.
Consider a path of densities $\rho \in C^2([0,T];\cA_\sigma)$ satisfying for all $0\leq t\leq T$
\begin{itemize}
    \item $\int_\Lambda \rho(t) \, \rd x = N$,
    \item $\rho(t) \geq c > 0$,
    \item $\rho(0)=\rho_{\Psi_0}$,
    \item $\partial_t\rho(0)=-\nabla\cdot j_{\Psi_0}$.
\end{itemize}
Then there exists $0 < T' \leq T$, $0 < \sigma' \leq \sigma$ and a real-valued $V\in C^0([0,T'],\cA_{\sigma'})$
for which the solution $\Psi$ to the many-body Schrödinger equation~\eqref{eq:Schrodinger} satisfies $\rho_{\Psi(t)} = \rho(t)$ for all $0\leq t \leq T'$. 
This potential is unique in the following sense: If we have a real-valued $\tilde V\in C^0([0,T''],\cA_{\sigma''})$ for some $0< \sigma''\leq\sigma'$ and $0<T''\leq T'$ satisfying 
\begin{equation}
\int_\Lambda \tilde V(t,x) \,\rd x =\int_\Lambda V(t,x) \,\rd x    \label{eq:same_integral}
\end{equation}
for all $t\leq T''$ and such that the solution $\tilde\Psi$ to Schrödinger's equation~\eqref{eq:Schrodinger} satisfies $\rho_{\tilde\Psi}(t)=\rho(t)$ for $t\in[0,T'']$, then $\tilde V(t,x)=V(t,x)$ for all $t\in[0,T'']$ and $x\in\Lambda$.
\end{theorem}

\begin{remark}[Continuity equation]
In the statement we used the current of a wavefunction
\begin{equation*}
j_{\Psi}(x):= 2N\Im \idotsint_{\Lambda^{N-1}}\overline{\Psi(x,x_2,...,x_N)} \nabla_x\Psi(x,x_2,...,x_N)\,\dx_2\cdots\dx_N.
\end{equation*}
(The factor $2$ is because we are working in a system of units where the mass is $m=1/2$.)
Recall the \emph{continuity equation}
\begin{equation*}
    \partial_t\rho+\nabla\cdot j_{\Psi}=0
\end{equation*}
that is satisfied for any solution $\Psi(t)$ of Schrödinger's equation~\eqref{eq:Schrodinger}. In the statement we need to assume that this continuity equation holds at time $t=0$, which prescribes the initial value of the time-derivative $\partial_t\rho(0)$.
\end{remark}

\begin{remark}[Uniqueness]
Our somewhat long uniqueness statement in Theorem~\ref{thm.main.V} is to emphasize that, although we can find a solution with some given analytic regularity parameter $\sigma'$, this solution stays unique in the bigger space $\cA_{\sigma''}$ corresponding to smaller regularity $\sigma''\leq \sigma'$. This stronger uniqueness property will be used later. 
\end{remark}

\begin{remark}[Gauge-freedom]\label{rmk.gauge.inverse.pb}
Note that for any (real-valued) time-dependent constant $C(t)$, the two potentials $V$ and $V+C(t)$ give rise to the same density. Hence, uniqueness only holds up to addition of such a constant. This freedom is removed by adding the condition~\eqref{eq:same_integral}. 

Later, we will systematically resolve this trivial non-uniqueness of $V$ by adding the constraint $\int V  \,\rd x =0$ for all $t$. This is the most natural convention in our setting, since $V$ is obtained by inverting an operator on a space of functions orthogonal to constants (the elliptic operator on the left side of~\eqref{eq:V_Psi_vague}). Note that most physics papers work on the full space $\R^D$, where $V$ is not necessarily expected to be  integrable at infinity. On $\R^D$, it is more natural to require that $V\to0$ at infinity, or perhaps $\int_{\R^D} V\rho\,\dx=0$. 
\end{remark}

\begin{remark}[Ground states]\label{rmk.ground.states}
It is common to start with $\Psi_0$ a ground state of $\cH_{V_0}$ for some $V_0\in\cA_\sigma$, which is then driven out of equilibrium. By Theorem~\ref{thm.density.is.analytic}, we know that such ground states are always in $\cB_\sigma$. Only the positivity of the density $\rho_{\Psi_0}$ is missing to be able to apply Theorem~\ref{thm.main.V}. The strict positivity of the density is valid in the non-interacting case $V_{ee}\equiv0$ because the first orbital is itself positive, by Perron--Frobenius~\cite[Sec.~XIII.12]{ReeSim4}. By perturbation theory, the strict positivity remains for small-enough interactions. In dimension $D=1$, where a Perron--Frobenius property holds in the fermionic $N$-body space, the positivity of $\rho$ is valid with a general interaction $V_{ee}$, as was recently proved in~\cite{Corso-26,Corso-26b}. 
In dimension $D=3$, it is expected that densities of interacting ground states are always strictly positive for Coulomb systems~\cite{FouHofSor-07}, but no proof of this claim has yet been given. 
\end{remark}

\begin{remark}[Maximal time of existence]\label{rmk.max.T}
In the proof, we provide explicit lower bounds on the regularity parameter $\sigma'$ and the time of existence $T'$ for the solution $V(t,x)$ in terms of the data. For any such fixed $\sigma'$, we can iterate the argument and consider the largest time $0<T'(\sigma')\leq T$ for which the solution $V$ exists in $\cA_{\sigma'}$. It is unavoidable that this time depends on the fixed regularity parameter $\sigma'$ due to the loss of derivatives in the problem. Going to larger times requires decreasing the regularity. The time $T'(\sigma')$ is non-increasing with $\sigma'$ because the spaces $\cA_{\sigma'}$ are decreasing (ordered by inclusion) in $\sigma'$.
This allows us to define the maximal time of existence
\begin{equation*}
    T^{\rm max}[\Psi_0,\rho]:=\sup_{0<\sigma'\leq\sigma}T'(\sigma')\leq T
\end{equation*}
so that our solution $V$ is in fact in $C^0([0,T''];\cA_{\sigma''})$ for any given $0<T''<T^{\rm max}[\Psi_0,\rho]$ and some $\sigma''$ depending on $T''$. Note that in our convention the trajectory $\rho$ is given in $C^2([0,T],\cA_\sigma)$, hence $\sigma$ and $T$ are fixed once and for all. If $T^{\rm max}[\Psi_0,\rho]<T$ then this must be because $\norm{V(t_n)}_{\sigma''}$ and $\norm{\Psi(t_n)}_{{\sigma''}}$ blow up for an appropriate sequence $(t_n)$ with $\limsup_nt_n\leq T^{\rm max}[\Psi_0,\rho]$ and all $0<\sigma''\leq \sigma$. 
\end{remark}

\begin{remark}[Dependence on $\Psi_0$ and $\rho$]
The potential $V$ of \cref{thm.main.V} depends continuously on the initial state $\Psi_0$ and the given density $\rho$. More precisely, the map 
\begin{equation*}
\cB_\sigma \times C^2([0,T];\cA_\sigma) \ni (\Psi_0, \rho) \mapsto V\in C^0([0,T'];\cA_{\sigma'})
\end{equation*}
is locally Lipschitz for some $\sigma'\leq\sigma$ and $T'\leq T$. This is a consequence of the proof via the Banach fixed-point theorem (see \cref{sec.Banach.fixed.point} below), noting that from \cref{lem.Krho} we can deduce that the contraction operator is Lipschitz in $\rho$. 
\end{remark}

\begin{remark}[Less regularity in space]
The main message of Theorem~\ref{thm.density.is.analytic} was that working with analytic-in-space densities and potentials is well-justified physically, even for the Coulomb interaction, if we leave aside the question of point nuclei addressed in Remark~\ref{rmk.point.nuclei}. Identifying a lower regularity space for which the inverse problem of Theorem~\ref{thm.main.V} can be solved is a very interesting, though probably difficult, question. In this spirit we mention the work of Lampart~\cite{Lampart-21}, who proved that the set of densities is meagre in the sense of Baire in $C^0([0,T];H^{s})$ for potentials in $C^0([0,T];H^{s})$, this for all $s>0$ large enough.
\end{remark}

\begin{remark}[Other boundary conditions]
Periodic boundary conditions are physically important for solid state physics. 
It is natural to consider also other boundary conditions, e.g.~Neumann or Dirichlet. For both, we lose the translation invariance and thus cannot use the same analyticity-in-the-center-of-mass argument, but need instead to consider smoother $V_{ee}$'s. Dirichlet boundary conditions have the further difficulty that the density vanishes at the boundary.
\end{remark}

\begin{remark}[Dependence on the interaction $V_{ee}$]\label{rmk.adiabatic.connection}
The lower bounds on $T',\sigma'$ only depend on the interaction $V_{ee}$ through its norm $\|V_{ee}\|_{L^2(\Lambda)}$. In particular, we can interpolate between two interactions $V_{ee}, V_{ee}'\in L^2(\Lambda)$ in the manner $(1-\lambda) V_{ee}+\lambda V_{ee}'$ with a coupling constant $0\leq\lambda\leq1$ and we get a unique solution $V_\lambda$ over some time independent of $\lambda$. This is useful to give a meaning to the \emph{adiabatic connection} in TDDFT \cite[Chs.~13 and 14]{Ullrichs-11} (see also \cite{LanPer-75,LanPer-80}), where the interacting problem is connected to the non-interacting one ($V_{ee}'\equiv0$).
\end{remark}

\begin{remark}[Runge--Gross, van Leeuwen, and time-analyticity]
The famous Runge--Gross theorem \cite{RunGro-84,FouLamLewSor-16,RugPenBau-JPA-09} states that, if it exists, the potential $V$ is unique under the assumptions 
\begin{itemize}
    \item $\rho, V_{ee}, \Psi_0, V$ are infinitely often differentiable in both time and space, and 
    \item $V$ is real-analytic in time. 
\end{itemize}
Later, van Leeuwen~\cite{vanLeeuwen-99} identified an important equation satisfied by $V$ (that we alluded to in~\eqref{eq:V_Psi_vague}) and used it to determine the Taylor series of $V$ at $t=0$. Showing existence of $V$ along these lines, however, requires establishing time-analyticity and a good  control on the $x$-dependence. 

Unfortunately, the analyticity-in-time is not generally expected for the solution with a singular interaction such as Coulomb~\cite{FouLamLewSor-16}, and would in any case need more assumptions on the initial condition $\Psi_0$ (see Remark~\ref{rmk:time-regularity}). Our result, on the other hand, proves both existence and uniqueness with minimal regularity in time. More precisely, as we show in Theorem~\ref{thm.density.is.analytic}, the natural assumption is to require that the density is twice-differentiable in time, which provides an inverse potential that is {a priori} only \emph{continuous} in time. 
\end{remark}

\begin{remark}[Spin]
As we briefly mentioned above, our result is equally applicable to the setting of particles with spin, and to bosons, if we fix the total density (summing over spin). More precisely, we can replace $L^2_a(\Lambda^N)$ in the definition of $\cB_\sigma$ by $\bigwedge^N L^2(\Lambda, \C^{2s+1})$ (for spin-$s$-fermions) or $\bigotimes_{\mathrm{sym}}^N L^2(\Lambda;\C^{2s+1})$ (for spin-$s$-bosons). 

Our proof also applies to the more involved case where different densities $\rho_{\tau}$ are prescribed for each spin sector $\tau\in\{-s,-s+1,\ldots,s\}$. In this case one obtains $2s+1$  corresponding inverse potentials $V_\tau$, which amounts to taking for $V$ a $(2s+1)\times (2s+1)$ diagonal matrix with entries in $\cA_\sigma$.
\end{remark}

\begin{remark}[Non-interacting and interacting $v$-representability]
A famous problem in quantum chemistry is the \emph{$v$-representability problem}~\cite[Sec.~4.4.2]{MarMaiNogGroRub-12}
of classifying the set of densities $\rho(t)$ arising from a time-dependent Schrödinger evolution \eqref{eq:Schrodinger} for some external potential $V(t)$ with either $V_{ee}=0$ (non-interacting) or $V_{ee}\ne 0$ (interacting). 
\cref{thm.main.V} in particular shows that the set of analytic-in-space and $C^2$-in-time strictly positive densities is both non-interacting and interacting $v$-representable. 

The analogous static (i.e.~time-independent) $v$-representability problem of classifying densities arising from ground states of (non-)interacting Hamiltonians as in \eqref{eq:H_V} was recently solved in dimension $D=1$ in a series of works \cite{Corso-26,Corso-26b,SutSarPenRugLeeGie-24,Corso-2025}. The case of higher dimensions is largely open. 
\end{remark}

\subsection{Ideas of proof}

As we quickly explained in the introduction, the main difficulty of the inverse problem considered in this paper is the loss of derivatives that occurs when $V$ is expressed in terms of the wavefunction. That one should reduce the problem of finding $V$ to a nonlinear equation in $\Psi$ only was suggested first in \cite{MaiTodWooBur-10}, where the difficulties coming from the loss of derivatives were clearly identified.

We overcome this difficulty by rewriting the nonlinear equation in a suitable way, using the imposed constraints $\rho_\Psi=\rho$. In other words we replace many terms involving $\rho_\Psi$ by $\rho$ on the right of~\eqref{eq:V_Psi_vague_1derivative}. This does not change anything for the  final solution and can help to improve the structure of the nonlinear equation to be solved. In our case, we use $\Delta^2\rho=\Delta^2\rho_\Psi$ to absorb the highest-order derivatives of the wavefunction and obtain the nonlinear equation 
\begin{equation}
    i\partial_t \Psi (t) 
    = \left(
        \sum_j-\Delta_j + \sum_j V_{\Psi(t)} (t,x_j) + \sum_{j<k} V_{ee}(x_j-x_k)
    \right)
    \Psi(t), 
    \label{eq:nonlinear_Psi_V}
\end{equation}
where $V_{\Psi(t)}$ only involves first derivatives of $\Psi$ (and higher ones of $\rho$). The precise formula for $V_\Psi$ is provided  in \eqref{eq:VPsi} below.

Equation~\eqref{eq:nonlinear_Psi_V} involves a nonlocal nonlinearity that loses derivatives as well as terms involving the singular interaction $V_{ee}$. To show well-posedness we exploit that $V_{\Psi(t)}$ is a one-body object depending on $\Psi$ through terms in which $N-1$ coordinates are integrated out such as
\begin{align*}
 \int \rd X \, \overline{\partial^i\Psi(x,X)} \partial^j\Psi(x,X)
\end{align*}
where $\partial^i = \partial_{x^i}$ for $x=(x^1,\ldots,x^D)\in \Lambda$ and $X=(x_2,\ldots, x_N)\in \Lambda^{N-1}$. The main technical challenge is then to bound all these terms only using the total momentum $\cP$. The terms involving $V_{ee}$ in $V_{\Psi(t)}$ are particularly involved. Derivatives of $\Psi$ are needed to control the singularity of $V_{ee}$ by an analogue of a Sobolev inequality, but not too many derivatives can be involved. The $H^2$-weight in the norm~\eqref{eq:normsigma} corresponds to the domain of the $N$-particle Hamiltonians $\cH_0$ and $\cH_{V(t)}$. This weight plays a crucial role in controlling the derivatives of $\Psi$ that are not in the ``right'' direction of the center of mass.  

Finally, we also need to study the invertibility of the elliptic operator $V\mapsto -\nabla\cdot(\rho\nabla V)$ appearing on the left side of~\eqref{eq:V_Psi_vague}. For this it is essential that $\rho$ is itself analytic and does not vanish. 

With all the desired bounds at hand, the proof is essentially reduced to a classical Cauchy--Kovalevskaya setting~\cite{Ovsyannikov-65,Ovsyannikov-71,Nirenberg-72,Nishida-77}. No existing result seems to apply directly, however, due to the Schrödinger part $\cH_0$ that involves two derivatives, but commutes with the analytic weight in the total momentum. We thus provide the full argument for the convenience of the reader. Although we could have used a Moser iteration scheme with the norm index $\sigma$ decreased at each step~\cite{Nishida-77,Nirenberg-72}, we instead follow an elegant argument of Baouendi and Goulaouic~\cite{BaoGou-77,BaoGou-77b} which relies on a Banach fixed-point technique in a space where the regularity parameter $\sigma$ depends on time linearly, $\sigma(t)=\sigma(1-t/T)$, and the corresponding $\sigma(t)$-norm is allowed to blow up in a specific, integrable-in-time manner, when $t$ approaches the final time $T$. 

\begin{remark}[Analytic interaction $V_{ee}$]
The proof is much easier if the interaction $V_{ee}$ is assumed analytic. In this case one can show that $\Psi$ is analytic in all directions and the estimates are much easier. This result essentially fits into the standard Cauchy--Kovalevskaya setting.  
\end{remark}

\begin{remark}[Fixed-point]
Since our method is based on a Banach fixed-point technique, a consequence of our proof is that the sequence defined iteratively as
$$ i\partial_t\Psi_{n+1}=\cH_0\Psi_{n+1}+\sum_{j=1}^NV_{\Psi_{n}}(x_j)\Psi_n, 
\qquad 
\Psi_{0}(t) = e^{-it\cH_0}\Psi_0
$$
converges to the unique solution. 
The wavefunctions $\Psi_n$ constructed this way do not stay normalized along the iterations, even if the limit necessarily is. 
A slightly different iterative scheme was proposed in~\cite{RugLee-11,RugGiePenLee-12,RugPenLee-15,Penz-PhD,TarUll-21}. 
Studying the convergence of these and other possible iteration schemes is an interesting problem for future study.
\end{remark}

%%%%%%%%%%%%%%%%%%%%%%%%%%%%%%%%%%%%%%%%%%%%%%%%%%%%%%%%%%%%%%%%%%%%%%%%%%%%%%%%
%%%%%%%%%%%%%%%%%%%%%%%%%%%%%%%%%%%%%%%%%%%%%%%%%%%%%%%%%%%%%%%%%%%%%%%%%%%%%%%%
\section{Applications to time-dependent density functional theory}
\label{sec.TDDFT}
%%%%%%%%%%%%%%%%%%%%%%%%%%%%%%%%%%%%%%%%%%%%%%%%%%%%%%%%%%%%%%%%%%%%%%%%%%%%%%%%
%%%%%%%%%%%%%%%%%%%%%%%%%%%%%%%%%%%%%%%%%%%%%%%%%%%%%%%%%%%%%%%%%%%%%%%%%%%%%%%%
In this section we explain how \cref{thm.main.V} can be used to provide the \textbf{first rigorous justification of time-dependent density functional theory (TDDFT)}. We start by explaining how one can use the inverse result in Theorem~\ref{thm.main.V} to show that the many-body problem can be expressed solely in terms of the density. Then we turn to Kohn--Sham theory, where the high-dimensional \emph{linear} many-body Schrödinger equation~\eqref{eq:Schrodinger} is replaced by a \emph{nonlinear} equation posed in $L^2(\Lambda)^N$ for $N$ orbitals.

%%%%%%%%%%%%%%%%%%%%%%%%%%%%%%%%%%%%%%%%
\subsection{Orbital-free TDDFT}
In \emph{orbital-free TDDFT} we look for a nonlinear equation involving only the density $\rho(t,x)$ and the initial state $\Psi_0$, of which the exact Schrödinger density is the unique solution. The statement involves the interacting inverse problem that we solved in Theorem~\ref{thm.main.V}.

\begin{definition}[Universal interacting functional]\label{def:V_functional}
Let $\Psi_0$ be a normalized wavefunction in some $\cB_\sigma$ with $\sigma>0$ and let $\rho\in C^2([0,T];\cA_\sigma)$ be a trajectory of densities so that $\rho>0$, $\rho(0)=\rho_{\Psi_0}$, $\int_\Lambda \rho(t) \, \rd x = N$, and $\partial_t\rho(0)=-\nabla\cdot j_{\Psi_0}$. For any such $(\Psi_0,\rho)$, we denote by $V[\Psi_0,\rho](t,x)$ the potential constructed in Theorem~\ref{thm.main.V}, normalized in the manner
\begin{equation}
\int_\Lambda V[\Psi_0,\rho](t,x)\,\dx=0.   
\label{eq:normalization_integral_V}
\end{equation}
Its maximal time of existence is denoted by $T^{\rm max}[\Psi_0,\rho]\in(0,T]$ (cf.~Remark~\ref{rmk.max.T}). In particular, $V[\Psi_0,\rho]$ is uniquely defined in $C^0([0,T'];\cA_{\sigma'})$ for all $0<T'< T^{\rm max}[\Psi_0,\rho]$ and some $0<\sigma'\leq\sigma$ depending on $T'$.
\end{definition}

Next, based on our Theorem~\ref{thm.main.V}, we are able to provide an implicit equation characterizing the exact Schrödinger solution in terms of the functional $V[\Psi_0,\rho]$. This result is consistent with the Runge--Gross theorem \cite{RunGro-84}, but it is not typically presented as such in the chemistry literature (see \cite{MAQUI_TDDFT-26} for a general perspective on this topic).

\begin{corollary}[Orbital-free TDDFT]\label{cor:orbital_free}
Let $\Psi_0\in\cB_\sigma$ be so that $\rho_{\Psi_0}>0$ on $\Lambda$. Consider a real-valued external potential $V_{\rm ext}\in C^0([0,T],\cA_\sigma)$ normalized in the manner $\int_\Lambda V_{\rm ext}\,\dx=0$ and denote by $\Psi^{\rm S}$ the wavefunction solving the many-body Schrödinger equation~\eqref{eq:Schrodinger} and by $\rho^{\rm S}:=\rho_{\Psi^{\rm S}}$ its density. Upon decreasing $T$ and $\sigma$, we may assume that $\rho^{\mathrm{S}}\in C^2([0,T];\cA_\sigma)$ and $\rho^{\rm S}(t,x)>0$ for every $x\in\Lambda$ and $t\leq T$. Then this exact Schrödinger density $\rho^{\rm S}$ is the \emph{unique solution} to the implicit equation
\begin{equation}
    \boxed{V[\Psi_0,\rho]=V_{\rm ext}}
    \label{eq:orbital_free}
\end{equation}
among all trajectories $\rho\in C^2([0,T'],\cA_{\sigma'})$ for some $0<\sigma'\leq\sigma$ and $0<T'\leq T$ satisfying $\rho>0$, $\int \rho =N$, and the two constraints $\rho(0)=\rho_{\Psi_0}$ and $\partial_t\rho(0)=-\nabla\cdot j_{\Psi_0}$.
\end{corollary}

\begin{proof}
For any given trajectory $\rho$ as in the statement, recall that $V[\Psi_0,\rho]$ is well-defined over some maximal time of existence $T^{\rm max}[\Psi_0,\rho]\leq T'\leq T$. Uniqueness means that if~\eqref{eq:orbital_free} holds for all $0\leq t<T^{\rm max}[\Psi_0,\rho]$ then necessarily $T^{\rm max}[\Psi_0,\rho]=T'$ and $\rho(t)=\rho^{\rm S}(t)$ on $[0,T']$. To prove this claim, let us consider $T''<T^{\rm max}[\Psi_0,\rho]$ and $\sigma''$ so that $V[\Psi_0,\rho]\in C^0([0,T''];\cA_{\sigma''})$. The equation~\eqref{eq:orbital_free}  means that $V_{\rm ext}$ reproduces the density $\rho$ on $[0,T'']$, by definition of $V[\Psi_0,\rho]$. But the solution of Schrödinger's equation~\eqref{eq:Schrodinger} with $V_{\rm ext}$ is unique and its density is $\rho^{\rm S}$ by definition. It follows immediately that $\rho(t)=\rho^{\rm S}(t)$ for $t\in[0,T'']$ and thus by iteration until the final time $T'$.
Note here that the iteration reaches the final time as in fact no regularity is lost on the $[0,T'']$ due to \eqref{eq:orbital_free}, and hence, the time of existence in \cref{thm.main.V} is uniformly bounded from below.
\end{proof}

\begin{remark}[Gauge-freedom]\label{rmk.gauge}
We have the same gauge freedom as mentioned in Remark~\ref{rmk.gauge.inverse.pb}. The potential provided by Theorem~\ref{thm.main.V} is only defined up to addition of a time-dependent constant $C(t)$. We have chosen to break this gauge freedom to ensure uniqueness by requiring the normalization~\eqref{eq:normalization_integral_V}. 
\end{remark}

The implicit equation~\eqref{eq:orbital_free} is of little practical use because it is hard to derive efficient approximations to the functional $(\Psi_0,\rho)\mapsto V[\Psi_0,\rho]$. It is more a proof of concept, showing that the whole theory can be based on the density only. Note that this allows us to also view the Schrödinger solution $\Psi(t)$ as a functional of $\Psi_0$ and $\rho$ instead of the usual external potential $V_{\rm ext}$. 

As an illustration we now provide the simplest approximation to $V[\Psi_0,\rho]$, the famous time-dependent Thomas--Fermi functional $V^{\rm TF}[\rho]$ (it is independent of the initial state $\Psi_0$). 
It arises from a hydrodynamics equation, which we first discuss.

\begin{remark}[Hydrodynamics]\label{rmk:hydrodynamics}
For any given energy functional $\cE[\rho]$ of the density $\rho$ only, we can consider the associated hydrodynamics equations arising from this functional in an external potential $V_{\rm ext}$. This is the (classical/symplectic) Hamiltonian flow associated with
$\int_\Lambda (\rho |\nabla \theta|^2 + \cE[\rho]+V_{\rm ext}\rho) \, \rd x$ with conjugate variables $\theta$ and $\rho$. (One interprets $-\nabla \theta$ as the velocity field of the fluid having local density $\rho$, see \cite{Bloch-33,CheSie-18}.) The hydrodynamics equations  formally read 
\begin{equation}
\partial_t \theta(t,x) = |\nabla\theta(t,x)|^2 + \frac{\delta\cE[\rho(t)]}{\delta\rho(t,x)}+V_{\rm ext}(t,x), 
\qquad
\partial_t \rho(t,x) = 2 \nabla \cdot ( \rho(t,x) \nabla \theta(t,x)), 
\label{eqn.hydrodynamics}
\end{equation}
with 
 $\delta\cE/\delta\rho$
 the functional derivative.  
 We remark further that the second equation can also be viewed as the continuity equation for the (mass-$N$) one-body wavefunction $\varphi = \sqrt{\rho}e^{-i\theta}$, whose current is $j_\varphi = -2 \rho \nabla \theta$, and whose velocity field is thus $-\nabla\theta$. The equations~\eqref{eqn.hydrodynamics} are unchanged if $V_{\rm ext}$ is replaced by $V_{\rm ext}+C(t)$ and $\theta$ by $\theta+\int_0^t C(s)\,\rd s$. 

We like to see~\eqref{eqn.hydrodynamics} as an approximation to the many-particle Schrödinger equation~\eqref{eq:Schrodinger}. But we can also think of~\eqref{eqn.hydrodynamics} as replacing the true universal density-to-potential functional $V[\Psi_0,\rho]$ by an approximation $V^{\rm HD}[\rho]$ and then solving
\begin{equation}
    V^{\rm HD}[\rho]=V_{\rm ext}
    \label{eq:hydrodynamics_approx}
\end{equation}
in place of~\eqref{eq:orbital_free}. Indeed, for any trajectory $\rho(t)$, let us call $\theta[\rho](t,x)$ the unique solution of the elliptic equation $2\nabla\cdot(\rho\nabla\theta)=\partial_t\rho$ satisfying $\int_\Lambda\theta=0$. For a given trajectory $\rho\in C^2([0,T],\cA_\sigma)$ with $\rho>0$ and $\int \rho =N$, 
we can find the unique solution $\theta \in C^1([0,T];\cA_{\sigma})$ of zero mean by Lemma~\ref{lem.Krho} below. With such a functional $\theta[\rho]$ at hand, we can now introduce the potential functional
$$V^{\rm HD}[\rho]:=\partial_t\theta[\rho]-|\nabla\theta[\rho]|^2-\frac{\delta\cE[\rho(t)]}{\delta\rho(t,x)}+C(t)$$
with the constant $C(t)$ chosen so that $\int_\Lambda V^{\rm HD}[\rho]=0$ for all $t$. Then~\eqref{eq:hydrodynamics_approx} is nothing but~\eqref{eqn.hydrodynamics} in disguise, after adding $\int_0^t C(s)\,\rd s$ to $\theta$. 

This class of approximations uses only $\rho$ and not the initial state $\Psi_0$. The potential $V^{\rm HD}[\rho](t,x)$ depends on $\rho(t)$, $\partial_t\rho(t)$ and $\partial_t^2\rho(t)$ through $\partial_t\theta$. In this sense, it is not completely local in time. It is not local in space either, due to the inversion of the elliptic operator in the definition of $\theta$.

In principle we can use this scheme for any functional $\cE[\rho]$. If we choose the Levy--Lieb functional~\cite{Levy-79,Lieb-83b,LewLieSei-23_DFT} that provides the smallest $N$-particle quantum energy at given density $\rho$, then the functional derivative $\delta\cE/\delta\rho$ is the universal ground state potential (whose rigorous existence is still an open problem in most cases). This provides a kind of adiabatic theory, that is exact for ground states and only an approximation for time-dependent systems. In the next remark we mention a simpler approximation.
\end{remark}

\begin{remark}[Time-dependent Thomas--Fermi equation]
As a particular example of a hydrodynamics equation, we consider the \emph{time-dependent Thomas--Fermi equation}~\cite{Bloch-33,Gombas-49,CheSie-18}. It is the hydrodynamics equation \eqref{eqn.hydrodynamics} with the functional $\cE[\rho]$ chosen as the Thomas--Fermi functional \cite{Gombas-49,Lieb-81b,LieSim-77b}
\begin{equation*}
\cE^{\mathrm{TF}}[\rho] =  \int_\Lambda \left(\frac{D}{D+2}c^{\mathrm{TF}}\rho^{1+2/D} + \frac{1}{2}(V_{ee}*\rho)\rho \right) \rd x, 
\qquad 
c^{\rm TF} = 4\pi^2\left(\frac{D}{|\bS^{D-1}|}\right)^{2/D}.
\end{equation*}
The first term is the semi-classical kinetic energy functional and the second one the mean-field (i.e., Hartree) approximation to the interaction~\cite{LewLieSei-23_DFT}. The time-dependent Thomas--Fermi equation in an external potential $V_{\rm ext}$ is
\begin{equation*}
 \partial_t \theta = |\nabla\theta|^2 + c^{\mathrm{TF}}\rho^{2/D} + V_{ee}*\rho + V_{\rm ext}, 
 \qquad 
 \partial_t\rho=2\nabla\cdot(\rho\nabla\theta). 
\end{equation*}
As explained in the previous remark, this corresponds to approximating the exact $V[\Psi_0,\rho]$ by the \emph{time-dependent Thomas--Fermi  potential functional} 
\begin{equation*}
    V^{\rm TF}[\rho]:=\partial_t\theta-|\nabla\theta|^2-c^{\rm TF}\rho^{2/D}-V_{ee}\ast\rho+C(t),\qquad \partial_t\rho = 2 \nabla \cdot(\rho \nabla\theta),
\end{equation*}
where $\theta$ is given by the second equation and the constant $C(t)$ is chosen to ensure $\int_\Lambda V^{\mathrm{TF}}[\rho] \, \rd x = 0$. 
In the static case $\partial_t\rho=0$, the (unique up to a constant) solution is $\theta\equiv0$ and we find the usual Thomas--Fermi potential with a minus sign. 
\end{remark}

%%%%%%%%%%%%%%%%%%%%%%%%%%%%%%%%%%%%%%%%
\subsection{Kohn--Sham theory}
The idea of Kohn--Sham theory \cite{KohSha-PR-65} is to replace the true (interacting) $N$ electrons by $N$ fictitious non-interacting particles reproducing the exact time-dependent density of the true system, thanks to some nonlinear terms modeling the missing interaction. This way one ends up with Hartree--Fock-type equations that are much easier to solve on a computer than the high-dimensional Schrödinger equation.
The theory relies on the non-interacting functional obtained by applying our inverse result in Theorem~\ref{thm.main.V} in the case $V_{ee}\equiv0$.

\begin{definition}[Universal non-interacting functional]\label{def:V0_functional}
When $V_{ee}\equiv0$ we denote by $V_{\mathrm{s}}[\Psi_0,\rho]$ the functional introduced in  Definition~\ref{def:V_functional} and by $T^{\rm max}_{\mathrm{s}}[\Psi_0,\rho]$ its maximal time of existence.\footnote{The subscript `$\mathrm{s}$' refers to `Slater'. For non-interacting systems it is natural to assume the initial condition $\Psi_0$ to be a Slater determinant, in which case $\Psi(t)$ remains one.}
\end{definition}

Next, let us give ourselves a (known) real-valued external potential $V_{\rm ext}(t,x)$ in $C^0([0,T];\cA_\sigma)$ for some $T,\sigma>0$ and an initial wavefunction $\Psi_0\in\cB_\sigma$ with $\rho_{\Psi_0}>0$. Let $\Psi^{\rm S}(t)$ denote the corresponding many-body Schrödinger solution to~\eqref{eq:Schrodinger}. By Theorem~\ref{thm.density.is.analytic} we know that $\Psi^{\rm S}\in C^0([0,T];\cB_\sigma)$ and $\rho^{\rm S}:=\rho_{\Psi^{\rm S}}\in C^2([0,T];\cA_\sigma)$, after possibly decreasing $\sigma$. In particular $\rho^{\rm S}(t)$ remains bounded away from 0 for some time and up to decreasing $T$ we can assume that $\rho^{\rm S}(t,x)>0$ for all $x\in\Lambda$ and $t\leq T$.

Let us now consider any initial Slater determinant $\Phi_0= \frac{1}{\sqrt{N!}} \varphi_{0,1}\wedge \cdots \wedge\varphi_{0,N}$ with $\varphi_{0,j}\in\cA_\sigma$, $\rho_{\Phi_0}=\rho_{\Psi_0}$,  and $\nabla \cdot j_{\Phi_0} = \nabla \cdot j_{\Psi_0}$.\footnote{Such states always exist for a density in $\cA_\sigma$ using for instance the Harriman--Lieb construction in~\cite{Harriman-81} and~\cite[Thm.~1.2]{Lieb-83b} and adding a phase to reproduce the current.}
Since $\rho^{\rm S}\in C^2([0,T];\cA_\sigma)$ we can apply Theorem~\ref{thm.main.V} with $V_{ee}\equiv0$ and conclude that the solution $\Phi^{\rm V}(t)\in C^0([0,T'];\cB_{\sigma'})$ to the non-interacting equation
\begin{equation}
    i \partial_t \Phi^{\rm V}(t)
    = \left(\sum_{j=1}^N (-\Delta_{x_j} + V_{\mathrm{s}}[\Phi_0,\rho^{\rm S}](t,x_j))\right) \Phi^{\rm V}(t),
    \qquad 
    \Phi^{\rm V}(0) = \Phi_0,
\label{eq:KS_Phi}
\end{equation}
has the desired density $\rho_{\Phi^{\rm V}(t)}=\rho^{\rm S}(t)$, for all $0<T'<T^{\rm max}_{\mathrm{s}}[\Phi_0,\rho^{\rm S}]\leq T$ and some $0<\sigma'\leq \sigma$ depending on $T'$. In other words we can \textbf{represent the density of any interacting system by that of a non-interacting system}, at the expense of changing the external potential, for data fulfilling the assumptions of our theorems. 

Although this proves the feasibility of using non-interacting particles to reproduce an interacting density, this does not lead to a predictive scheme. This is because everything depends here on the exact density $\rho^{\rm S}$, which is not known. 

Let us now explain how this difficulty is handled in TDDFT, by introducing \textbf{nonlinear Kohn--Sham equations} that do not explicitly rely on $\rho^{\rm S}$ in their definition. We follow the standard textbook description that can be read for instance in \cite{Ullrichs-11,MarUllNogRubBurGro-06,MarMaiNogGroRub-12} and introduce the following

\begin{definition}[Universal Hartree-exchange-correlation functional]
\label{def.Hxc.potential}
Let $\Psi_0$ and $\rho$ satisfy the assumptions of Definition~\ref{def:V_functional}. Let $\Phi_0=\frac{1}{\sqrt{N!}} \varphi_{0,1}\wedge \cdots \wedge\varphi_{0,N}$ be a Slater determinant with $\varphi_{0,j}\in\cA_\sigma$, $\rho_{\Phi_0}=\rho_{\Psi_0}$, and $\nabla\cdot j_{\Phi_0}=\nabla\cdot j_{\Psi_0}$. We define the \emph{Hartree-exchange-correlation (Hxc) functional} by
\begin{equation}
\boxed{V_{\rm Hxc}[\Psi_0,\Phi_0,\rho]:=V_{\mathrm{s}}[\Phi_0,\rho]-V[\Psi_0,\rho].}
 \label{eqn.def.Hxc}
\end{equation}
This potential is well-defined over some maximal time of existence
$$
T^{\rm max}_{\rm Hxc}[\Psi_0,\Phi_0,\rho]
:=\min\left(T^{\rm max}[\Psi_0,\rho]\,,\, T^{\rm max}_{\mathrm{s}}[\Phi_0,\rho]\right)
$$
and belongs to $C^0([0,T'],\cA_{\sigma'})$ for all $0<T'<T^{\rm max}_{\rm Hxc}[\Psi_0,\Phi_0,\rho]$ and some $0<\sigma'\leq \sigma$ depending on $T'$.
\end{definition}

We can now introduce the exact Kohn--Sham equations of which $\rho^{\rm S}$ is the unique solution. The following is, to our knowledge, the first rigorous justification of time-dependent Kohn--Sham theory.

\begin{corollary}[Exact Kohn--Sham nonlinear equations]\label{cor:KS}
Let $\Psi_0\in\cB_\sigma$ be so that $\rho_{\Psi_0}>0$ on $\Lambda$. Consider a real-valued external potential $V_{\rm ext}\in C^0([0,T];\cA_\sigma)$ normalized in the manner $\int_\Lambda V_{\rm ext}=0$ and denote by $\Psi^{\rm S}(t)$ the solution to the many-body Schrödinger equation~\eqref{eq:Schrodinger}. Upon decreasing $T$ and $\sigma$, we can assume that its density $\rho^{\rm S}:=\rho_{\Psi^{\rm S}}$ is in $C^2([0,T];\cA_\sigma)$ and is strictly positive. Let finally $\Phi_0=\frac{1}{\sqrt{N!}}\varphi_{0,1}\wedge \cdots \wedge\varphi_{0,N}$ be a Slater determinant with $\varphi_{0,j}\in\cA_\sigma$, $\rho_{\Phi_0}=\rho_{\Psi_0}$, and $\nabla\cdot j_{\Phi_0}=\nabla\cdot j_{\Psi_0}$.

Then, the Slater $\Phi^{\rm V}(t)$ from~\eqref{eq:KS_Phi} solves the \emph{nonlinear Kohn--Sham equations}
\begin{equation}
\left\{
\begin{aligned}
\dps
    i \partial_t \Phi^{\rm KS}(t)
    & = 
    \left(\sum_{j=1}^N 
        \left(
            -\Delta_{x_j} 
            + V_{\rm ext}(t,x_j)
            + V_{\rm Hxc}\big[\Psi_0,\Phi_0,\rho_{\Phi^{\rm KS}}\big](t,x_j)
        \right)
    \right) \Phi^{\rm KS}(t),
    \\
    \Phi^{\rm KS}(0) 
    & = \Phi_0.
\end{aligned}
\right.
\label{eq:TDKS_Phi_V}
\end{equation}
It is the unique solution to this nonlinear equation among all $\Phi\in C^0([0,T''],\cB_{\sigma''})$ satisfying $\rho_{\Phi}\in C^2([0,T''],\cA_{\sigma''})$ for some $0<T''\leq T'$ and $0<\sigma''\leq \sigma'$, with $\Phi(0)=\Phi_0$, $\rho_\Phi>0$ and $\partial_t\rho_{\Phi}(0)=-\nabla\cdot j_{\Phi_0}$. In particular, this unique solution reproduces the exact Schrödinger density: $\rho_{\Phi^{\rm KS}(t)}=\rho^{\rm S}(t)$.
\end{corollary}

\begin{proof}
The solution $\Phi^{\rm V}(t)$ to~\eqref{eq:KS_Phi} has density $\rho^{\rm S}$ by definition, hence $V[\Psi_0,\rho_{\Phi^{\rm V}}]=V[\Psi_0,\rho^{\rm S}]=V_{\rm ext}$ by Corollary~\ref{cor:orbital_free}. This implies from the definition of $V_{\rm Hxc}$ in~\eqref{eqn.def.Hxc} that
$$V_{\rm ext}+V_{\rm Hxc}[\Psi_0,\Phi_0,\rho_{\Phi^{\rm V}}]=V_{\mathrm{s}}[\Phi_0,\rho^{\rm S}].$$ 
Hence~\eqref{eq:TDKS_Phi_V} reduces to~\eqref{eq:KS_Phi}. 

Next, we prove uniqueness. Consider an arbitrary trajectory $\Phi\in C^0([0,T''],\cB_{\sigma''})$ so that $\rho:=\rho_{\Phi}\in C^2([0,T''],\cA_{\sigma''})$ satisfies the mentioned constraints at $t=0$. These conditions allow us to give a meaning to $V[\Psi_0,\rho]$ and $V_{\mathrm{s}}[\Phi_0,\rho]$, and hence to $V_{\rm Hxc}[\Psi_0,\Phi_0,\rho]$ until some maximal time $T^{\rm max}_{\rm Hxc}[\Psi_0,\Phi_0,\rho]$. 
Thus, $V_{\rm Hxc}[\Psi_0,\Phi_0,\rho]\in C^0([0,T'''],\cA_{\sigma'''})$ for all  $0<T'''<T^{\rm max}_{\rm Hxc}[\Psi_0,\Phi_0,\rho]$ and some $0<\sigma'''\leq\sigma''$. We assume that $\Phi$ solves the nonlinear equation~\eqref{eq:TDKS_Phi_V} in those spaces for $t<T^{\rm max}_{\rm Hxc}[\Psi_0,\Phi_0,\rho]$. Now, the unique potential that can give rise to the density $\rho$ in the non-interacting equation is $V_{\mathrm{s}}[\Phi_0,\rho]$, from the uniqueness in Theorem~\ref{thm.main.V} with $V_{ee}\equiv0$. From the definition of $V_{\rm Hxc}$ in~\eqref{eqn.def.Hxc} we conclude that $V_{\rm ext}(t)=V[\Psi_0,\rho](t)$ for any $t<T^{\rm max}_{\rm Hxc}[\Psi_0,\Phi_0,\rho]$. By Corollary~\ref{cor:orbital_free} this proves that $\rho=\rho^{\rm S}$ on that time interval. But then $V_{\mathrm{s}}[\Phi_0,\rho]=V_{\mathrm{s}}[\Phi_0,\rho^{\rm S}]$, and thus $\Phi=\Phi^{\rm V}$. Iterating the argument to later times, we conclude that $T^{\rm max}_{\rm Hxc}[\Psi_0,\Phi_0,\rho]=T''$ and $\Phi=\Phi^{\rm V}$. This concludes the proof. 
\end{proof}

\begin{remark}[Formulation in terms of orbitals]
\Cref{eq:TDKS_Phi_V,eq:KS_Phi} are often written in terms of orbitals instead, giving rise to $N$ coupled equations, and hence the name Kohn--Sham \emph{equations}. 
Concretely, denoting by  
$\varphi_1^{\mathrm{KS}},\ldots,\varphi_N^{\mathrm{KS}}$ the (orthonormal) orbitals of the Slater determinant $\Phi^{\mathrm{KS}} = \frac{1}{\sqrt{N!}} \varphi_1^{\mathrm{KS}} \wedge \cdots \wedge \varphi_N^{\mathrm{KS}}$, \eqref{eq:TDKS_Phi_V} reads
\begin{equation*}
\left\{
\begin{aligned}
        i \partial_t \varphi_j^{\mathrm{KS}}(t) & = \left(-\Delta + V_{\mathrm{ext}}(t) + V_{\mathrm{Hxc}}\bigl[\Psi_0, \Phi_0, {\textstyle \sum_{k=1}^N} |\varphi_k^{\mathrm{KS}}|^2 \bigr](t)\right) \varphi_j^{\mathrm{KS}} , 
        \\
        \varphi_j^{\mathrm{KS}}(0) & = \varphi_{0,j} , 
\end{aligned}
\right.
\quad j=1,\ldots,N
.
\end{equation*}
\end{remark}

\begin{remark}[Universality]\label{rmk.universality}
The functionals $V_{\mathrm{s}}, V, V_{\mathrm{Hxc}}$ are ``universal'' in the sense that they do not depend on the external potential $V_{\mathrm{ext}}$, that is, on the physical situation. They of course still depend on the kinetic energy operator $\cT= \sum_{j=1}^N -\Delta_{x_j}$ and interaction $V_{ee}$. In quantum chemistry applications these are, however, considered fixed. 
What one does vary, however, is the externally applied field $V_{\mathrm{ext}}$, and the spirit of DFT is to provide approximations of the objects which do not depend on this externally given field. 
\end{remark}

\begin{remark}[Non-locality and causality]
The functionals $V[\Psi_0,\rho]$, $V_{\mathrm{s}}[\Phi_0,\rho]$, and subsequently $V_{\mathrm{Hxc}}[\Psi_0,\Phi_0,\rho]$ are non-local in both space and time in the sense that $V[\Psi_0,\rho](t,x)$ (analogously for $V_{\mathrm{s}}$ and $V_{\mathrm{Hxc}}$) 
depends on $\Psi_0(x_1',...,x_N')$ and $\rho(t',x')$ for all $x',x_1',...,x_N'$, and $t'$. 
They are, however, causal in the sense that they depend only on $\rho(t')$ for $t'\leq t$. 

The causality follows from (the precise versions of) the formulas for $V$ in \eqref{eq:V_Psi_vague} and \eqref{eq:V_Psi_vague_1derivative}. This causality is important both as a physical property and in terms of solvability of the nonlinear equation \eqref{eq:TDKS_Phi_V}. 
\end{remark}

\begin{remark}[Different interactions]
In \eqref{eq:KS_Phi} and \eqref{eq:TDKS_Phi_V} we reproduce the density $\rho^{\mathrm{S}}$ of an interacting system ($V_{ee}\ne 0$) using non-interacting systems ($V_{ee}=0$) with appropriately chosen external potentials. 
We can also reproduce $\rho^{\mathrm{S}}$ in an interacting system, but with a different interaction $0\ne V_{ee}' \ne V_{ee}$ similarly as in \cite{vanLeeuwen-99}. 
Indeed, simply replace $V_{\mathrm{s}}$ in \cref{def.Hxc.potential} by the universal functional of \cref{def:V_functional} with interaction $V_{ee}'$. We obtain a suitably translated version of \cref{cor:KS}. 
\end{remark}

\begin{remark}[Approximations]\label{rmk.approx}
Solving the nonlinear \eqref{eq:TDKS_Phi_V} exactly is only feasible for (toy) systems with few particles. It more or less amounts to solving the full interacting Schrödinger equation \eqref{eq:Schrodinger}. 
In practice, one thus approximates the Hxc potential and solves then \eqref{eq:TDKS_Phi_V} with $V_{\mathrm{Hxc}}$ replaced by one such approximation. The perhaps most famous and used approximation is the \emph{adiabatic local density approximation (ALDA)}, but many more approximations exist, see \cite[Ch.~4]{Ullrichs-11} and \cite[Sec.~4.7]{MarMaiNogGroRub-12}.  
Many approximations only deal with the dependence on the density $\rho$, see \cite{MaiBur-01} for a discussion on the dependence on the initial states $\Psi_0,\Phi_0$.

The time-dependent Kohn--Sham equations with an approximate Hxc potential have been studied in the mathematical literature both for adiabatic \cite{PusSig-21,BreFanFau-26,DupLetLev-25,SprCiaBor-17} and non-adiabatic \cite{Jerome-2015} approximations. A reference work in the stationary case is~\cite{AnaCan-09}.
\end{remark}

%%%%%%%%%%%%%%%%%%%%%%%%%%%%%%%%%%%%%%%%%%%%%%%%%%%%%%%%%%%%%%%%%%%%%%%%%%%%%%%%%%
%%%%%%%%%%%%%%%%%%%%%%%%%%%%%%%%%%%%%%%%%%%%%%%%%%%%%%%%%%%%%%%%%%%%%%%%%%%%%%%%%%
\section{The density of an interacting system is real-analytic in space (Proof of \texorpdfstring{\cref{thm.density.is.analytic}}{Theorem~\ref*{thm.density.is.analytic}})}
%%%%%%%%%%%%%%%%%%%%%%%%%%%%%%%%%%%%%%%%%%%%%%%%%%%%%%%%%%%%%%%%%%%%%%%%%%%%%%%%%%
%%%%%%%%%%%%%%%%%%%%%%%%%%%%%%%%%%%%%%%%%%%%%%%%%%%%%%%%%%%%%%%%%%%%%%%%%%%%%%%%%%
In this section we give the proof of \cref{thm.density.is.analytic}, that the density of an interacting system is in fact real-analytic in space. We start with some preliminary properties of the norms $\norm{\cdot}_\sigma$ defined in \cref{def.norms.sigma}. 

\subsection{Properties of norms}
We collect here some properties of the norms $\norm{\cdot}_\sigma$.

\begin{lemma}[Gradient bound]
\label{lem.bdd.gradient}
	We have for $0 \leq \sigma' < \sigma $ and $f\in \cA_\sigma, \Psi\in \cB_\sigma$
	\begin{equation*}
  \norm{\abs{\cP} \Psi}_{\sigma'} \leq \frac{1}{\sigma-\sigma'} \norm{\Psi}_{\sigma}, 
  \qquad 
  \norm{\nabla f}_{\sigma'} \leq \frac{1}{\sigma - \sigma'} \norm{f}_\sigma
  .
\end{equation*}
\end{lemma}

\begin{proof}
    This follows immediately from the definition as $\abs{\cP} e^{\sigma'\abs{\cP}} \leq e^{\sigma\abs{\cP}}/(\sigma-\sigma')$.
\end{proof}

The following says that $(\cA_\sigma\otimes 1)\cdot \cB_\sigma\subset\cB_\sigma$ (where $\cdot$ is to be understood as multiplication of elements) and implies in particular that $\cA_\sigma$ is an algebra.

\begin{lemma}[Product rule]
\label{lemma:multiplication}
Let $\Psi \in \cB_\sigma$ and $f\in \cA_\sigma$. 
For any $j\in\{1,...,N\}$ and $s\geq 0$ we have 
\begin{equation*}
    \norm{\pscal{\cP}^s (f(x_j)\Psi (x_1,\dots,x_N))}_\sigma
    \leq C \norm{\pscal{\cP}^s \Psi}_\sigma \norm{f}_\sigma
    +C \norm{ \Psi}_\sigma \norm{(1-\Delta)^{s/2} f}_\sigma
    ,
\end{equation*}
where $\pscal{x} = \sqrt{1+|x|^2}$ denotes the Japanese bracket for $x\in \R^n$ for any $n\in\N$. 

For $N=1$ and $s=0$ we simply get $\norm{fg}_\sigma \leq C\norm{f}_\sigma\norm{g}_\sigma$.
\end{lemma}

To prove \cref{lemma:multiplication}, we recall the following property  of the Fourier transform (recall that we consider functions on the torus $\Lambda = [0,1]^D$ with periodic boundary conditions): 
\begin{equation*}
    \hat \Psi_1 * \hat \Psi_2(\vec k) 
    = 
    \hat{\Psi_1 \Psi_2}(\vec k),
\end{equation*}
using the notation $\vec k = (k_1,\ldots,k_N)$. 

\begin{proof}[Proof of \cref{lemma:multiplication}]
By symmetry we can take $j=1$. For simplicity of notation we just write $f$ for the function $(x_1,...,x_N)\mapsto f(x_1)$. Then we have
$$|\widehat{f\Psi}(\vec p)|\leq \sum_{k}|\widehat{f}(k)|\;|\widehat{\Psi}(\vec p - k)|,$$
where used the short-hand notation $\vec{p}-k=(p_1-k,p_2,...,p_N)$. 
Next, we use that $e^{\sigma|P|}\leq e^{\sigma|k|}e^{\sigma|P-k|}$ from the triangle inequality, with $P=\sum_{j=1}^Np_j$, as well as the inequality
\begin{equation}
\pscal{x}^2\leq 2\pscal{y}^2+2\pscal{x-y}^2
\label{eq:japanese}
\end{equation}
for $x,y\in \R^n$. 
We use this for $x=\vec{p}$ and $y=(k,0,...,0)\in \R^{DN}$
and obtain
\begin{multline}
\pscal{P}^s\pscal{\vec p}^2e^{\sigma|P|}|\widehat{f\Psi}(\vec p)|\\*
\leq C\sum_{k} \bigl(\pscal{P-k}^s+\pscal{k}^s\bigr)\bigl(\pscal{k}^2+\pscal{\vec p-k}^2\bigr)e^{\sigma|k|}|\widehat{f}(k)|\;e^{\sigma|P-k|}|\widehat{\Psi}(\vec p - k)|.    
\label{eq:estim_fPsi}
\end{multline}
Next we consider the term containing $\pscal{P-k}^s$ and take the square and sum over $p_1$. To handle the two terms in~\eqref{eq:estim_fPsi} ($\pscal{k}^2$ and $\pscal{\vec p - k}^2$), we use the discrete Young and Cauchy--Schwarz inequalities in the form (for the first term)
\begin{align*}
    & \sum_{p_1}\left(\sum_k \pscal{\vec{p}-k}^2 \pscal{P-k}^s e^{\sigma|P-k|} a_k b_{\vec{p}-k}\right)^2
    \nn\\
    &\quad 
    \leq \left(\sum_k a_k\right)^2\left(\sum_{p_1} \pscal{\vec{p}}^4 \pscal{P}^{2s} e^{2\sigma|P|} b_{\vec p}^2\right)
    \nn\\
    &\quad 
    \leq \left(\sum_k \pscal{k}^{-4}\right)\left(\sum_k \pscal{k}^4a_k^2\right)\left(\sum_{p_1} \pscal{P}^{2s} e^{2\sigma|P|} \pscal{\vec p}^4b_{\vec p}^2\right)
\end{align*}
for $a_k,b_{\vec p}\geq0$. It is here essential that $D\in\{1,2,3\}$ to ensure the convergence of the first sum. We arrive at
\begin{multline*}
    \sum_{p_1} \left(\sum_k \pscal{P-k}^s \pscal{\vec{p}-k}^2 e^{\sigma|k|} |\hat f(k)| e^{\sigma|P-k|} |\hat\Psi(\vec p-k)|\right)^2
    \\
    \leq 
    C\norm{f}_{\sigma}^2\sum_{p_1}\pscal{P}^{2s}\pscal{\vec{p}}^4e^{2\sigma|P|}|\widehat{\Psi}(\vec p)|^2.    
\end{multline*}
Summing over $p_2,...,p_N$ yields $\norm{f}_\sigma^2 \norm{\pscal{\cP}^s\Psi}_\sigma^2$. 
The term with $\pscal{P-k}^s$ and $\pscal{k}^2$ in \eqref{eq:estim_fPsi} is treated similarly yielding the same bound. 
Finally, the terms with $\pscal{k}^s$ in \eqref{eq:estim_fPsi} can again be treated analogously yielding the bound $\norm{(1-\Delta)^{s/2} f}^2_\sigma\norm{\Psi}_\sigma^2$. We conclude the desired. 
\end{proof}

\subsection{Analyticity of eigenstates}\label{sec.eigenstate.analytic}
We first consider the eigenstates of $\cH_V$ (recall the definition of $\cH_V$ in \eqref{eq:H_V}) and prove the first part of \cref{thm.density.is.analytic}. Concretely, we prove the following. 

\begin{lemma}[Analyticity of eigenstates]\label{lem.eigenstate.analytic}
Let  $V\in \cA_\sigma$ and let $\Psi_0\in D(\cH_V)=H^2 \cap L^2_a(\Lambda^N)$ denote an eigenstate of $\cH_V$, i.e., such that $\cH_V\Psi_0 = E \Psi_0$ for some $E\in \R$. 
Then, $\Psi_0\in \cB_\sigma$.
\end{lemma}

\begin{proof}
First, we note that $V,V_{ee}\in L^2$ and so, by the Kato--Rellich theorem~\cite[Thm.~6.2]{Lewin-Spectral}, the operator $\cH_V$ is self-adjoint on $D(\cT)=H^2 \cap L^2_a(\Lambda^N)$ 
(recall that $\cT = \sum_{j=1}^N -\Delta_{x_j}$). 
Any eigenstate must thus belong to this space. Moreover, 
for some constants $c,C>0$ depending only on $V_{ee},N,D$, we have for any $\Psi$ 
\begin{equation}
c \norm{(C+\cH_0)\Psi}_{L^2} 
    \leq \norm{\Psi}_{H^2}
    = \norm{\left(1 + \cT\right) \Psi}_{L^2}
    \leq \frac{1}{c} \norm{\left(C + \cH_0\right) \Psi}_{L^2},
    \label{eq:compare_cH0_cT}
\end{equation}
where we recall the notation $\cH_0 = \cH_{V=0}$. 
By Cauchy--Schwarz, we have $\cP^2\leq N \cT$ as operators. In particular, we may bound (rather crudely) $|\cP|\leq \sqrt{N}(1+\cT)$. 

We now prove the desired claim using a Grönwall argument. More precisely, we replace $|\cP|$ by the bounded operator $\cP_M = \min\{|\cP|, M\}$ for some $M>0$ and take the limit $M\to \infty$ at the end. Recall that $\cH_0$ commutes with $\cP$ (in the sense of unbounded operators, see \cite[Thm.~4.41]{Lewin-Spectral}), hence with any bounded function of $\cP$, such as $\cP_M$. Using that $\cP_M$ and $\cH_0$ commute, we find
\begin{align*}
    \frac{\rd}{\rd \tau } \norm{(C+\cH_0)e^{\tau\cP_M}\Psi_0}^2_{L^2}
    &
    = 2\pscal{(C+\cH_0)e^{\tau\cP_M}\Psi_0, \cP_M  e^{\tau\cP_M}(C+\cH_0)\Psi_0}
\nonumber
\\ 
&    
    \leq 2\norm{(C+\cH_0)e^{\tau\cP_M}\Psi_0}_{L^2} \norm{\cP_Me^{\tau\cP_M}(C+\cH_0)\Psi_0}_{L^2}.
\end{align*}
Using that $(\cH_V-E)\Psi_0=(\cH_0+\cV-E)\Psi_0=0$ with $\cV:=\sum_{j=1}^N V(x_j)$ denoting the external potential in the $N$-particle space, we find for the last term
\begin{align*}
    &\norm{e^{\tau \cP_M}\cP_M(C+\cH_0)\Psi_0}_{L^2}\\*
    &\qquad\qquad \leq |C+E|\norm{\cP_M e^{\tau \cP_M} \Psi_0}_{L^2}+\norm{\cP_M e^{\tau \cP_M}\cV\Psi_0}_{L^2}\\ 
    &\qquad\qquad  \leq |C+E|\sqrt{N}\norm{(1+\cT) e^{\tau \cP_M} \Psi_0}_{L^2}+\sqrt{N}\norm{(1+\cT) e^{\tau \cP_M}\cV\Psi_0}_{L^2}. 
\end{align*}
In the second inequality we used that $\cP_M\leq |\cP|\leq \sqrt{N}(1+\cT)$. From the proof of \cref{lemma:multiplication} (we need to use that $p\mapsto\min(|p|,M)$ satisfies the triangle inequality), the last term can be bounded by 
\begin{align}
\norm{(1+\cT) e^{\tau \cP_M}\cV\Psi_0}_{L^2}&\leq CN\|V\|_\tau\norm{(1+\cT)e^{\tau\cP_M}\Psi_0}_{L^2}\nn\\
&\leq \frac{C}cN\|V\|_\tau\norm{(C+\cH_0)e^{\tau\cP_M}\Psi_0}_{L^2}.
\label{eq:bound_V_right_cP_M}
\end{align}
For $\tau\leq \sigma$, we have thus proved that
\begin{equation*}
    \frac{\rd}{\rd \tau} \norm{(C+\cH_0)e^{\tau\cP_M}\Psi_0}_{L^2}^2 
    \leq C'\big(1+|E|+\|V\|_\sigma\big) \norm{(C+\cH_0)e^{\tau\cP_M}\Psi_0}_{L^2}^2, 
\end{equation*}
for some constant $C'$ depending only on $V_{ee},N,D$. 
By Grönwall's lemma and~\eqref{eq:compare_cH0_cT} we conclude that 
$$\norm{(1+\cT)e^{\sigma\cP_M}\Psi_0}^2_{L^2} \leq c^{-4}e^{C'\sigma (1+|E|+\|V\|_\sigma)}\norm{(1+\cT)\Psi_0}^2_{L^2}.$$ 
Recall that $\Psi_0\in D(\cH_V)=D(\cT)$. Taking the limit $M\to \infty$ and by the monotone convergence theorem we get that the norms $\norm{(1+\cT)e^{\sigma\cP_M}\Psi_0}_{L^2}$ converge. Further, clearly $(1+\cT)e^{\sigma\cP_M}\Psi_0$ converges weakly in $L^2$. Thus, $(1+\cT)e^{\sigma\cP_M}\Psi_0$ converges in norm and we obtain $\Psi_0\in \cB_\sigma$ as desired. 
\end{proof}

\subsection{Propagation of analyticity for the time-dependent equation}\label{sec.SE.solvable}
Let us now move to the second part of Theorem~\ref{thm.density.is.analytic} concerning the time-dependent Schrödinger equation \eqref{eq:Schrodinger} 
$i \partial_t \Psi(t) = \cH_{V(t)}\Psi(t)$ with $\Psi(0) = \Psi_0$. Under our assumption that $V_{ee}\in L^2(\Lambda)$, we recall that $\Psi\in C^0([0,T];H^2(\Lambda^N))\cap C^1([0,T];L^2(\Lambda^N))$, whenever $\Psi_0\in H^2(\Lambda^N)$ and $V\in C^0([0,T];C^2(\Lambda))$. The proof of this standard fact is very similar to that of Lemma~\ref{lem.SE.solvable} below (with $\sigma=0$). Our goal is to show that the analyticity in the center of mass is preserved for analytic $V$ and $\Psi_0$. The precise statement is the following.

\begin{lemma}[Propagation of analyticity]\label{lem.SE.solvable}
Let $\Psi_0\in \cB_\sigma$ and $V\in C^0([0,T];\cA_\sigma)$. Then, the unique solution to Schrödinger's equation~\eqref{eq:Schrodinger} satisfies 
$\Psi\in C^0([0,T];\cB_\sigma)$.
\end{lemma}

\begin{proof}[Proof of \cref{lem.SE.solvable}]
The argument is similar to that of \cref{lem.eigenstate.analytic}. We compute
\begin{align*}
\frac{\rd}{\dt}\norm{(C+\cH_0)e^{\sigma\cP_M}\Psi}^2_{L^2}&=\pscal{\Psi,i\left[\cH_V,(C+\cH_0)^2e^{2\sigma\cP_M}\right]\Psi}\\
& =\pscal{\Psi,i\left[\cV,(C+\cH_0)^2e^{2\sigma\cP_M}\right]\Psi}\\
& =2\Re\pscal{(C+\cH_0)e^{\sigma\cP_M}\Psi,i\left[\cV(t)\,,\,(C+\cH_0)e^{\sigma\cP_M}\right]\Psi}\\
&\leq2\norm{(C+\cH_0)e^{\sigma\cP_M}\Psi}_{L^2}\Big( N\norm{V(t)}_{L^\ii}\norm{(C+\cH_0)e^{\sigma\cP_M}\Psi}_{L^2}\\
&\qquad\qquad +\norm{(C+\cH_0)e^{\sigma\cP_M}\cV(t)\Psi}_{L^2}\Big),
\end{align*}
where we recall that $\cV(t)=\sum_{j=1}^NV(t,x_j)$ denotes the external potential. To give a precise meaning to the calculation leading to the last equality, we can replace $C+\cH_0$ everywhere by $\frac{C+\cH_0}{1+\eps(C+\cH_0)}$ and take $\eps\to0$ at the end of the computation. Using~\eqref{eq:bound_V_right_cP_M} for the last term, together with~\eqref{eq:compare_cH0_cT} to replace $C+\cH_0$ by $1+\cT$, we find
\begin{align}
\frac{\rd}{\dt}\norm{(C+\cH_0)e^{\sigma\cP_M}\Psi}^2_{L^2}
&=2\Re\pscal{(C+\cH_0)e^{\sigma\cP_M}\Psi,i\left[\cV(t)\,,\,(C+\cH_0)e^{\sigma\cP_M}\right]\Psi}\label{eq:diff_t_B_sigma_M}\\
&\leq C'\norm{V(t)}_\sigma\norm{(C+\cH_0)e^{\sigma\cP_M}\Psi}_{L^2}^2,\nn
\end{align}
where we also used that $\|V\|_{L^\ii}\leq C\|V\|_{H^2}\leq C \|V\|_\sigma$. By Grönwall's lemma we conclude that 
\begin{equation*}
\norm{(C+\cH_0)e^{\sigma\cP_M}\Psi(t)}^2_{L^2}\leq e^{C'T\sup_{0\leq t\leq T}\norm{V(t)}_\sigma}\norm{(C+\cH_0)e^{\sigma\cP_M}\Psi_0}_{L^2}^2,\qquad \forall 0\leq t\leq T.
\end{equation*}
Taking now $M\to\ii$ and using that $\Psi_0\in\cB_\sigma$, we obtain $\Psi\in L^\ii([0,T];\cB_\sigma)$. This can be upgraded to $C^0([0,T];\cB_\sigma)$ as follows. First we pass to the limit $M\to\ii$ in~\eqref{eq:diff_t_B_sigma_M} and obtain
$$\frac{\rd}{\dt}
\norm{(C+\cH_0)e^{\sigma|\cP|}\Psi(t)}^2_{L^2}\in L^\ii([0,T],\R).$$
This proves that $t\mapsto\|(C+\cH_0)e^{\sigma|\cP|}\Psi(t)\|_{L^2}$ is a continuous function. Since $\cB_\sigma$ is a Hilbert space (it is a weighted $L^2$-space in Fourier variables), convergence of norms together with weak convergence imply norm convergence. Since we already know that $\Psi\in C^0([0,T],L^2)$, weak continuity with respect to $t$ holds in $\cB_\sigma$. The continuity of the norm in $t$ implies the claim that $\Psi\in C^0([0,T];\cB_\sigma)$. 
\end{proof}

\subsection{Analyticity of the density}
Next, we prove that a wavefunction $\Psi\in\cB_\sigma$ has its one-particle density $\rho_\Psi\in \cA_\sigma$. 
Recall \eqref{eqn.def.density} 
that the one-particle density of $\Psi$ is given by $\rho_\Psi(x) = N \int_{\Lambda^{N-1}} \rd X \, \overline{\Psi(x,X)}\Psi(x,X)$. 
For later purposes, it will be useful to have a slightly more general statement. To this end, we denote the analogous bilinear integral by
\begin{equation*}
    \rho_{\Psi_1,\Psi_2}(x) = N \int_{\Lambda^{N-1}} \rd X \, \overline{\Psi_1(x,X)}\Psi_2(x,X).
\end{equation*}
Then, we have the following. 

\begin{lemma}[Analyticity of the density]
\label{lem:Psi_to_density_B_sigma}
    For $\Psi_1, \Psi_2 \in \cB_\sigma$ and $s\geq 0$ we have 
    \begin{equation}\label{eq:rho_psi_chi}
    \norm{(1-\Delta)^{s/2}\rho_{\Psi_1,\Psi_2}}_\sigma 
    \leq  
    C \norm{\frac{\pscal{\cP}^{2+s}}{1+\cT}\Psi_1}_\sigma \norm{\frac{\pscal{\cP}^2}{1+\cT}\Psi_2}_\sigma
    +C \norm{\frac{\pscal{\cP}^2}{1+\cT}\Psi_1}_\sigma \norm{\frac{\pscal{\cP}^{2+s}}{1+\cT}\Psi_2}_\sigma
    .
 \end{equation}
    In particular, we have
    $\rho_{\Psi_1}\in \cA_\sigma$ with
    $$
    \norm{\rho_{\Psi_1}}_\sigma \leq  C \norm{\Psi_1}_\sigma^2.
    $$
\end{lemma}

\begin{proof}
Recall the notation $\pscal{x}:=\sqrt{1+|x|^2}$ for $x\in \R^n$ for any integer $n$. 
We denote again $\vec{p}=(p_1,...,p_N)$ and $p+\vec{p}=(p+p_1,p_2,...,p_N)\in 2\pi \Z^{DN}$.
We need to bound
$$\norm{(1-\Delta)^{s/2}\rho_{\Psi_1,\Psi_2}}_\sigma^2 = \sum_{p\in 2\pi\Z^D} e^{2\sigma\abs{p}} \pscal{p}^{4+2s}\abs{\hat\rho_{\Psi_1,\Psi_2}(p)}^2.$$
First we recall that
\begin{equation*}
\hat\rho_{\Psi_1,\Psi_2}(p)=N\sum_{p_i\in 2\pi\Z^D}
\overline{\hat\Psi_1(\vec p)}\; \hat \Psi_2(p+\vec p),
\end{equation*}
so that
    \begin{equation*}
  \abs{\hat\rho_{\Psi_1,\Psi_2}(p)}
\leq N \sum_{p_i\in 2\pi\Z^D}
\abs{\hat\Psi_1(\vec p)} \abs{\hat \Psi_2(p+\vec p)}.
\end{equation*}
From \eqref{eq:japanese} we use $\pscal{p}^{2+s}\leq C \pscal{p+P}^{2+s}+ C \pscal{P}^{2+s} $ with $P=\sum_{j=1}^Np_j$.
Next, we observe that $e^{\sigma|p|}\leq e^{\sigma|p+P|}e^{\sigma|P|}$ from the triangle inequality. Thus,
    \begin{equation*}
\begin{split}
  e^{\sigma\abs{p}} \pscal{p}^{2+s}\abs{\hat\rho_{\Psi_1,\Psi_2}(p)}
&\leq C \sum_{p_j\in 2\pi\Z^D}
(\pscal{P}^{2+s} + \pscal{p+P}^{2+s})
e^{\sigma|P|}\abs{\hat\Psi_1(\vec p)}  e^{\sigma|p+P|}\abs{\hat \Psi_2(p+\vec p)}.
\end{split}
\end{equation*}
We bound the contribution of the first term (with $\pscal{P}^{2+s}$) to $\norm{(1-\Delta)^{s/2}\rho_{\Psi_1,\Psi_2}}_\sigma$ using Cauchy--Schwarz as 
\begin{align*}
& \sum_{p\in 2\pi\Z^D} \left(\sum_{p_j\in 2\pi\Z^D} \frac{\pscal{P}^{2+s}}{\pscal{p+P}^2}  e^{\sigma|P|}\abs{\hat \Psi_1(\vec p)}\pscal{p+P}^2 e^{\sigma|p+P|}\abs{\hat\Psi_2(p+\vec p)} \right)^2
\\ & \quad  \leq 
\sum_{p,p_j,k_j\in 2\pi\Z^D} \frac{\pscal{K}^{4+2s}}{\pscal{p+K}^4}  e^{2\sigma|K|}\abs{\hat \Psi_1(\vec k)}^2 \pscal{p+P}^4
e^{2\sigma|p+P|}\abs{\hat\Psi_2(p+\vec p)}^2
 \\
&  \quad 
\leq C \sum_{p\in 2\pi\Z^D} \pscal{p}^{-4} \norm{\frac{\pscal{\cP}^{2+s}}{1+\cT}\Psi_1}_\sigma^2 \norm{\frac{\pscal{\cP}^{2}}{1+\cT}\Psi_2}_\sigma^2
= C'  \norm{\frac{\pscal{\cP}^{2+s}}{1+\cT}\Psi_1}_\sigma^2 \norm{\frac{\pscal{\cP}^{2}}{1+\cT}\Psi_2}_\sigma^2.
\end{align*}
Here we inserted an additional $1=\pscal{p+P}^2/\pscal{p+P}^2$ that turned into the fraction $\frac{\pscal{p+P}^4}{\pscal{p+K}^4}$ by Cauchy--Schwarz and allows summability in $p$. To arrive at the last line, one computes the sums in the order $\vec p$, then $p$, and finally $\vec k$. 

The term with $\pscal{p+P}^{2+s}$ instead of $\pscal{P}^{2+s}$ is bounded analogously. 
We conclude the desired bound in \eqref{eq:rho_psi_chi}. 

Finally, using that $\pscal{\cP}^2/(1+\cT)$ is a bounded Fourier multiplier we obtain the second inequality with $\Psi_2=\Psi_1$ and $s=0$. 
\qedhere
\end{proof}

\begin{lemma}[The density is $C^2$ in time]\label{lem:Psi_to_density_C^2}
Let $\Psi_0\in \cB_\sigma$ and $V\in C^0([0,T];\cA_\sigma)$. Then, for the unique solution $\Psi$ to Schrödinger's equation~\eqref{eq:Schrodinger}, we have $\rho_{\Psi}\in C^2([0,T];\cA_{\sigma'})$ for any fixed $0<\sigma' < \sigma$ with
$$
\norm{\partial_t \rho_\Psi}_{\sigma'} \leq  \frac{C}{(\sigma-\sigma')^2} \norm{\Psi}_\sigma^2
\qquad
\textnormal{and}
\qquad
\norm{\partial_t^2 \rho_\Psi}_{\sigma'} \leq  \frac{C}{(\sigma-\sigma')^4} \norm{\Psi}_\sigma^2.
$$
\end{lemma}

The proof of \cref{lem:Psi_to_density_C^2} is much more involved than those of the previous results of this section. It relies on a precise computation of the second derivative $\partial_t^2\rho_\Psi$. For this reason we give the proof of~\cref{lem:Psi_to_density_C^2} later in \cref{sec.proof.rho.C2} below, after we have made the appropriate calculations and estimates.

\subsection{Concluding the proof of \texorpdfstring{\cref{thm.density.is.analytic}}{Theorem~\ref*{thm.density.is.analytic}}}
Combining the above arguments we conclude the 
\begin{proof}[Proof of \cref{thm.density.is.analytic}]
Part $(i)$ is proved in \cref{lem.eigenstate.analytic}. 

For part $(ii)$, the regularity of solutions to the Schrödinger equation is proved in \cref{lem.SE.solvable} and the stated properties of $\rho_\Psi$ are proved in \cref{lem:Psi_to_density_B_sigma,lem:Psi_to_density_C^2}. 
It remains to give the proof of \cref{lem:Psi_to_density_C^2}. 
This is given in \cref{sec.proof.rho.C2} below. 
\end{proof}

\section{The force-balance equation}
We next study the \emph{force-balance equation}, also called the \emph{van Leeuwen equation}~\cite{vanLeeuwen-99}. This is a central equation both for proving that $\rho_\Psi$ is twice continuously differentiable in time for the true Schrödinger evolution (as in \cref{thm.density.is.analytic}), and for the inverse problem of finding a $V$ that reproduces a given density (as in \cref{thm.main.V}).
It is the precise formulation of the equation alluded to in \eqref{eq:V_Psi_vague} and \eqref{eq:V_Psi_vague_1derivative}.

 \begin{proposition}[Force-balance equation]
\label{prop.force.balance}
Let $V\in C^0([0,T];\cA_\sigma)$ and $\Psi_0\in \cB_\sigma$. Then for the solution $\Psi$ to the Schrödinger equation  \eqref{eq:Schrodinger}, $i\partial_t \Psi = \cH_V \Psi$, we have
the continuity equation 
\begin{equation*}
    \partial_t \rho_{\Psi(t)} + \nabla \cdot j_{\Psi(t)} = 0,
\end{equation*}
and, in the sense of distributions (on $\Lambda$), 
the \emph{force-balance equation}
\begin{equation}\label{eq:partialt^2rho}
    \partial_t^2 \rho_{\Psi(t)} 
		= - \Delta^2\rho_{\Psi(t)} - 2 K_{\rho_{\Psi(t)}} V(t)
        + Q_{\Psi(t)},
\end{equation}
where 
\begin{equation}
K_\rho (V)=-\nabla\cdot(\rho\nabla V)
\label{eqn.def.Krho}
\end{equation}
and (with $\partial^i = \partial_{x^i}$ for $x=(x^1,\ldots,x^D)\in \Lambda$)
\begin{multline}
    Q_\Psi(x) 
    = 4N\sum_{i,j=1}^D \partial^i\partial^j\int \rd X \,\overline{\partial^i \Psi(x,X)} \partial^j\Psi(x,X)
	+ 2\Delta_x  \int \rd y \,   \rho^{(2)}_\Psi(x,y)   V_{ee}(x-y)
    \\*
    -2 \nabla_x  \cdot \int \rd y  \left(\nabla_x \rho^{(2)}_\Psi(x,y) V_{ee}(x-y) \right),
    \label{eqn.def.cT.Psi}
\end{multline}
where we used the two-particle density $\rho^{(2)}_\Psi(x,y) = N(N-1)\int_{\Lambda^{N-2}} \rd Y \, \abs{\Psi(x,y,Y)}^2$ with $Y = (x_3,\ldots,x_N)\in \Lambda^{N-2}$. 
\end{proposition}

In the TDDFT literature, the term \emph{force-balance equation} is sometimes also known as the (divergence of the) \emph{local force equation} and either refers sometimes to an equation satisfied by $\partial_t j_\Psi$ instead~\cite{TchPenTheRugRub-19,TanPenLaeCsiRugRub-24,TanPenRugRub-25}. 
We will refer to \eqref{eq:partialt^2rho} as the \emph{force-balance equation}.

This equation is the main equation used both to prove the regularity of $\rho$ in \cref{lem:Psi_to_density_C^2} and to give an explicit expression for the potential $V$ to use when studying the inverse problem in \cref{thm.main.V}. 
What van Leeuwen realized in \cite{vanLeeuwen-99} is that \eqref{eq:partialt^2rho} yields an equation for $V$ that reproduces the prescribed $\rho$ if we equate the left-hand-side with $\partial_t^2 \rho$ and replace $K_{\rho_{\Psi(t)}}$ by $K_{\rho(t)}$; see also \cite{MaiTodWooBur-10}. We use a slightly different formulation than \cite{vanLeeuwen-99} and also replace $\Delta^2\rho_{\Psi(t)}$ by $\Delta^2\rho (t)$, which at consistency of course yields the same. 
This is a crucial replacement, as it allows us to avoid handling fourth order derivatives of $\Psi$.

\subsection{Derivation (Proof of \texorpdfstring{\cref{prop.force.balance}}{Proposition~\ref*{prop.force.balance}})}
We give now the 
\begin{proof}[Proof of \cref{prop.force.balance}]
Recall that $\Psi\in C^1([0,T], L^2)$ whence $\rho_\Psi\in C^1([0,T], L^1)$ with
\begin{equation*}
\begin{split}
    \partial_t\rho_\Psi
  &= 2 N \Re\int \rd X\, \overline{\Psi(x,X)}\partial_t \Psi(x,X)
= 2 N \Re\int \rd X\, \overline{\Psi(x,X)} i\Delta_x\Psi(x,X)\\
&= - 2 N \nabla\cdot \Im\int \rd X\, \overline{\Psi(x,X)} \nabla_x\Psi(x,X)
=-\nabla \cdot j_{\Psi(t)}\\
\end{split}
\end{equation*}
for $j_{\Psi(t)} = 2 N \Im\int \rd X\, \overline{\Psi(x,X)} \nabla_x\Psi(x,X) = -2 N \Im\int \rd X\, \overline{\nabla_x\Psi(x,X)} \Psi(x,X)$.
As the second time-derivative formally contains higher order derivatives of $\Psi$, we cannot differentiate classically again. However, after integrating $\partial_t\rho_\Psi$ against a smooth test function $h \in C_c^\infty(\Lambda,\R)$, we have
\begin{align*}
\int \rd x \, \partial_t\rho_\Psi(x) h(x)
& = -\int \rd x \, \nabla \cdot j_{\Psi(t)} (x)h(x)
\\ & 
= -2 N \Im\int_{\Lambda^N} 
\rd x \, \rd X \, 
\overline{\nabla_x\Psi(x,X)}\Psi(x,X)\nabla h(x).
\end{align*}
This is now a differentiable function in time as can be seen by viewing it as an $L^2$-pairing between $\nabla_x\Psi$ and $\Psi(x,X)\nabla h(x)$. Indeed, recall from \cref{lem.SE.solvable} that $\Psi\in C^0([0,T];H^2) \cap C^1([0,T]; L^2)$ and so the second function $\Psi(x,X)\nabla h(x)$  has the same regularity. The first satisfies  $\nabla_x\Psi\in C^0([0,T];H^1) \cap C^1([0,T]; H^{-1})$, but as its time-derivative is integrated against an $H^2$ function, the pairing is well-defined. Thus,
\begin{align*}
  & \int \rd x  \, \partial_t^2\rho_\Psi(x) h(x)
  \\*
   & \quad 
   =-2 N \Im\int\overline{\nabla_x\partial_t\Psi(x,X)}\Psi(x,X)\nabla h(x)
   -2 N \Im\int\overline{\nabla_x\Psi(x,X)}\partial_t\Psi(x,X)\nabla h(x)\\
    & \quad = 2 N \Im\int\overline{\partial_t\Psi(x,X)}\nabla_x\cdot(\Psi(x,X)\nabla h(x))
      -2 N \Im\int\overline{\nabla_x\Psi(x,X)}\partial_t\Psi(x,X)\nabla h(x)\\
  & \quad = 2N \Im\int 
  \overline{\partial_t\Psi(x,X)}\Psi(x,X) 
  \Delta h(x) 
  +4 N \Im\int
  \overline{\partial_t\Psi(x,X)}\nabla_x \Psi(x,X) 
  \nabla h(x)\\
    & \quad = 2N \Re\int 
  \overline{ \cH_V\Psi(x,X)}\Psi(x,X) 
  \Delta h(x) 
  +4 N \Re\int
  \overline{ \cH_V\Psi(x,X)}\nabla_x \Psi(x,X) 
  \nabla h(x).
\end{align*}

Note that the parts of $V$ and $V_{ee}$ not involving $x$ as well as the Laplacian with respect to all other coordinates cancel out; take for instance 
\begin{multline*}
  2N \Re\int 
  \overline{ V(x_2)\Psi(x,X)}\Psi(x,X) 
  \Delta h(x) 
  +4 N \Re\int
  \overline{ V(x_2)\Psi(x,X)}\nabla_x \Psi(x,X) 
  \nabla h(x)
  \\
    =  2N \Re\int   V(x_2)
  \abs{\Psi(x,X)}^2 
  \Delta h(x) 
  +2 N \Re\int  V(x_2)
   \nabla_x\abs{\Psi(x,X)}^2 
  \nabla h(x) =0.
\end{multline*}
Hence,
\begin{equation*}
\begin{split}
  & \int \rd x \,  \partial_t^2\rho_\Psi(x) h(x)\\
    & \quad = 2N \Re\int 
  \overline{ -\Delta_x\Psi(x,X)}\Psi(x,X) 
  \Delta h(x) 
  +4 N \Re\int
  \overline{ -\Delta_x\Psi(x,X)}\nabla_x \Psi(x,X) 
  \nabla h(x)\\
  &\qquad +2N \Re\int 
  V(x) |\Psi(x,X)|^2
  \Delta h(x) 
  +4 N \Re\int
  V(x) \overline{ \Psi(x,X)}
  \nabla_x \Psi(x,X) 
  \nabla h(x)\\
   &\qquad +2N(N-1) 
   \int
   %\Re\int 
  V_{ee}(x-y)
  |\Psi(x,y,Y)|^2
  \Delta h(x) \\
  &\qquad +4 N (N-1)\Re\int
  V_{ee}(x-y)\overline{ \Psi(x,y,Y)}\nabla_x \Psi(x,y,Y) 
  \nabla h(x) . 
\end{split}
\end{equation*}
First, we rewrite the term in the first line by integration by parts.
It equals
\begin{align*}
  &-2N \Re \sum_{j,k=1}^D\int \left(
  \partial^j\partial^j\overline{ \Psi(x,X)}\Psi(x,X) 
  \partial^k\partial^k h(x) 
  +2 \partial^j\partial^j
  \overline{ \Psi(x,X)} \partial^k \Psi(x,X) 
  \partial^k h(x)
  \right)\\
    &\quad =-2N \Re \sum_{j,k=1}^D\int \left(
  \frac{1}{2}
  \abs{ \Psi(x,X)}^2 
  \partial^j\partial^j\partial^k\partial^k h(x) 
  -\abs{\partial^j \Psi(x,X)}^2 
  \partial^k\partial^k h(x) 
  \right)\\*
    &\qquad +2N \Re \sum_{j,k=1}^D\int \left(
 2 \partial^j
  \overline{ \Psi(x,X)} \partial^k \Psi(x,X) 
  \partial^j \partial^k h(x)
  +
  \partial^k\abs{\partial^j \Psi(x,X)}^2 
   \partial^k h(x)
  \right)\\
  &\quad = -\int \rd x \,  \rho_\Psi(x) \Delta^2 h(x) 
  +4N  \sum_{j,k=1}^D\int 
 \partial^j
  \overline{ \Psi(x,X)} \partial^k \Psi(x,X) 
  \partial^j \partial^k h(x).
\end{align*}
Next, we note that the term involving $V$ is
\begin{multline*}
  2N \int V(x)
  \abs{ \Psi(x,X)}^2 
  \Delta h(x) 
  +2 N \int
  V(x)\nabla_x \abs{ \Psi(x,X)}^2 
  \nabla h(x)\\
  =
  -2 \int 
  \nabla V(x) \rho_\Psi(x) 
  \nabla h(x)
  =-2\int K_{\rho_\Psi} V(x) h(x), 
\end{multline*}
where $K_{\rho_\Psi} f 
    = -\nabla\cdot (\rho_\Psi \nabla f) = -\rho_\Psi\Delta f - \nabla \rho_\Psi \cdot \nabla f$ is as defined in \eqref{eqn.def.Krho}.
Similarly, the term involving $V_{ee}$ is
\begin{multline*}
    2 N(N-1) \int V_{ee}(x-y) \left( |\Psi(x,y,Y)|^2 \Delta h(x) + \nabla_x |\Psi(x,y,Y)|^2 \nabla h(x)\right)
    \\
    = 2 \iint \rd x \, \rd y \, V_{ee}(x-y) \left( \rho^{(2)}_\Psi(x,y) \Delta h(x) +  \nabla \rho^{(2)}_\Psi(x,y) \nabla h(x)\right). 
\end{multline*}
This concludes the proof.
\end{proof}

\subsection{Bounding the nonlinear terms}
For the analysis of the force-balance equation \eqref{eq:partialt^2rho}, we first give a bound on the term $Q_\Psi$, showing that this distribution is in fact an analytic function in our setting. It will be convenient to replace one of the two occurrences of $\Psi$ by a possibly different function (recall $\rho_\Psi^{(2)} = N(N-1) \int_{\Lambda^{N-2}} \rd Y \, |\Psi(x,y,Y)|^2$). 
More concretely, we prove

\begin{lemma}[{The bilinear $Q_\Psi$-term}]
\label{lemma:boundspreliminaryTi}
Let $\Psi_1,\Psi_2\in \cB_\sigma$ and define, using $ \partial^i =  \partial_{x^i}$, for $i\in \{1,\ldots,D\}$, 
\begin{align}
    Q_{\Psi_1,\Psi_2}(x) 
    & = 4N\sum_{i,j=1}^D \partial^i\partial^j\int \rd X \,\overline{\partial^i \Psi_1(x,X)} \partial^j\Psi_2(x,X)
    \nn \\* & \quad 
	+ 2N(N-1)\Delta_x  \iint \rd y \, \rd Y \, \overline{\Psi_1(x,y,Y)}\Psi_2(x,y,Y) V_{ee}(x-y)
    \nn \\* & \quad 
    -2 N(N-1) \sum_{j=1}^D \partial^j \iint \rd y \, \rd Y \,  \partial^j \left(\overline{\Psi_1(x,y,Y)}  \Psi_2(x,y,Y)\right)    V_{ee}(x-y).
    \label{eq:T_Psi_chi}
\end{align}
(Note that $Q_{\Psi,\Psi} = Q_\Psi$.) 
Then, for any $s\geq 0$, 
\begin{equation*}
  \norm{(1-\Delta)^{s/2-1} Q_{\Psi_1,\Psi_2}}_{\sigma} 
  \leq C \norm{\pscal{\cP}^{s+1}\Psi_1}_{\sigma}\norm{\pscal{\cP}\Psi_2}_{\sigma}
  +C \norm{\pscal{\cP}\Psi_1}_{\sigma}\norm{\pscal{\cP}^{s+1}\Psi_2}_{\sigma}
  .
\end{equation*}
\end{lemma}

\begin{proof}
Define 
\begin{align*}
    q_1^{ij}(\Psi_1, \Psi_2) 
    & = 4N \int \rd X \, \partial^i \overline{\Psi_1(x,X)} \partial^j \Psi_2(x,X), 
    \\
    q_2(\Psi_1, \Psi_2) & = {2N(N-1)}\iint \rd y \,\rd Y  \,  \overline{\Psi_1(x,y,Y)} \Psi_2(x,y,Y) V_{ee}(x-y) , 
    \\
    q_3^i(\Psi_1, \Psi_2) & = \frac{{2 N(N-1)}}{\sqrt{1-\Delta_x}} \iint \rd y \, \rd Y \, \overline{\Psi_1(x,y,Y)} \partial^i \Psi_2(x,y,Y)    V_{ee}(x-y), 
\end{align*}
and note that 
\begin{multline*}
    Q_{\Psi_1,\Psi_2}
    = \sum_{i,j=1}^D \partial^i \partial^j q_1^{ij}(\Psi_1,\Psi_2) +   \Delta_x q_2(\Psi_1,\Psi_2)
    \\*
    -  \sum_{i=1}^D (1-\Delta_x)^{1/2} \partial^i \left( q_3^i(\Psi_1,\Psi_2) + \overline{q_3^i(\Psi_2,\Psi_1)}\right). 
\end{multline*}
We prove that, for $s\geq 0$ and $k=1,2,3$, 
\begin{equation*}
  \norm{(1-\Delta)^{s/2} q_k (\Psi_1, \Psi_2)}_{\sigma} 
  \leq C \norm{\pscal{\cP}^{s+1}\Psi_1}_{\sigma}\norm{\pscal{\cP}\Psi_2}_{\sigma}
  +C \norm{\pscal{\cP}\Psi_1}_{\sigma}\norm{\pscal{\cP}^{s+1}\Psi_2}_{\sigma}
  ,
\end{equation*}
where $q_1$ means any of the $q_1^{ij}$ and analogously $q_3$ refers to any $q_3^i$. 
From this the desired immediately follows.

For $q_1^{ij}$ we apply \cref{lem:Psi_to_density_B_sigma} noting that
$q_1^{ij}= 4\rho_{\partial^i\Psi_1, \partial^j \Psi_2}$ (clearly, \cref{lem:Psi_to_density_B_sigma} applies also to this setting)
so that
\begin{multline*}
\norm{(1-\Delta)^{s/2} q_1^{ij}}_\sigma
\leq
C \norm{\frac{\pscal{\cP}^{2+s}}{1+\cT}\partial^i\Psi_1}_\sigma \norm{\frac{\pscal{\cP}^2}{1+\cT}\partial^j \Psi_2}_\sigma
    \\* 
    +C \norm{\frac{\pscal{\cP}^2}{1+\cT}\partial^i\Psi_1}_\sigma \norm{\frac{\pscal{\cP}^{2+s}}{1+\cT}\partial^j \Psi_2}_\sigma.
\end{multline*}
The claim now follows immediately as $\frac{\pscal{\cP}}{1+\cT}\partial^i$ is a bounded Fourier multiplier.

For $q_2$ we split the potential into positive and negative parts $V_{ee} = V_{ee}^+ - V_{ee}^-$ and correspondingly $q_2 = q_2^+ - q_2^-$. We only consider $q_2^+$, as $q_2^-$ is handled analogously. Note that $q_2^+ = 2(N-1) \rho_{ \sqrt{V_{ee}^+}\Psi_1, \sqrt{V_{ee}^+}\Psi_2}$ so that, by \cref{lem:Psi_to_density_B_sigma}, 
\begin{multline*}
  \norm{(1-\Delta)^{s/2} q_2^+}_\sigma
\leq
C \norm{\frac{\pscal{\cP}^{2+s}}{1+\cT}\sqrt{V_{ee}^+}\Psi_1}_\sigma \norm{\frac{\pscal{\cP}^2}{1+\cT}\sqrt{V_{ee}^+}\Psi_2}_\sigma
   \\* 
   +C \norm{\frac{\pscal{\cP}^2}{1+\cT}\sqrt{V_{ee}^+}\Psi_1}_\sigma \norm{\frac{\pscal{\cP}^{2+s}}{1+\cT}\sqrt{V_{ee}^+}\Psi_2}_\sigma
.
\end{multline*}
We observe that the first term (the other one is the same up to swapping $\Psi_1$ and $\Psi_2$) satisfies
\begin{align*}
& \norm{\frac{\pscal{\cP}^{2+s}}{1+\cT}\sqrt{V_{ee}^+}\Psi_1}_\sigma
\norm{\frac{\pscal{\cP}^2}{1+\cT}\sqrt{V_{ee}^+}\Psi_2}_\sigma
   \\ & \quad 
   =
 \norm{e^{\sigma\abs{\cP}}\pscal{\cP}^{2+s}\sqrt{V_{ee}^+}\Psi_1}_{L^2} 
 \norm{e^{\sigma\abs{\cP}}\pscal{\cP}^{2}\sqrt{V_{ee}^+}\Psi_2}_{L^2}
 \\
     & \quad 
     =
 \norm{\sqrt{V_{ee}^+}e^{\sigma\abs{\cP}}\pscal{\cP}^{2+s}\Psi_1}_{L^2} 
 \norm{\sqrt{V_{ee}^+}e^{\sigma\abs{\cP}}\pscal{\cP}^{2}\Psi_2}_{L^2}\\
      &\quad 
      \leq C 
 \norm{\sqrt{1+\cT}e^{\sigma\abs{\cP}}\pscal{\cP}^{2+s}\Psi_1}_{L^2} 
 \norm{\sqrt{1+\cT}e^{\sigma\abs{\cP}}\pscal{\cP}^{2}\Psi_2}_{L^2}\\
  &\quad 
  \leq C 
 \norm{\pscal{\cP}^{1+s}\Psi_1}_{\sigma}\norm{\pscal{\cP}\Psi_2}_{\sigma} ,
\end{align*}
where we used the fact that the translation-invariant multiplication operator $\sqrt{V_{ee}^+}$ commutes with $\cP$ and then used the form bound $V_{ee}^+ \le C 
 \norm{V_{ee}}_{L^2} 
(1-\Delta)$
which follows from 
$
\pscal{f, V_{ee}^+f} \leq \norm{V_{ee}^+}_{L^2} \norm{f}_{L^4}^{2}  \leq C \norm{V_{ee}}_{L^2} \norm{ f}_{H^1}^2
$
by the Sobolev inequality \cite[Thms~8.3 and 8.5]{LieLos-01}.

For $q_3$ we again write $V_{ee} = V_{ee}^+ - V_{ee}^-$ and correspondingly $q_3 = q_3^+ - q_3^-$. 
Let $A = \sqrt{V_{ee}^+}\Psi_1$ and $B = \sqrt{V_{ee}^+}\partial^i \Psi_2$ so that $q_3^+ = 2(N-1) (1-\Delta)^{-1/2}\rho_{A,B}$
and 
\begin{equation*}
\norm{(1-\Delta)^{s/2} q_3^+}_\sigma^2 
= \sum_{p\in 2\pi\Z^D} \left(e^{\sigma\abs{p}} \pscal{p}^{s+2}\abs{\hat q_3^+ (p)}\right)^2 
= \sum_{p\in 2\pi\Z^D} \left(e^{\sigma\abs{p}} \pscal{p}^{s+1}\abs{\hat \rho_{A,B}(p)}\right)^2.
\end{equation*}
The proof now mirrors that of \cref{lem:Psi_to_density_B_sigma} with minor changes, since we wish to compensate for the presence of the additional derivative in the expression $B$. 
First, we have the bound 
\begin{multline*}
   e^{\sigma\abs{p}} \pscal{p}^{s+1}\abs{\hat \rho_{A,B}(p)} 
   \leq C\sum_{p_i\in 2\pi\Z^D: \pscal{P} \leq \pscal{p+P}}  
   e^{\sigma\abs{P}} \abs{\hat A(\vec p)}  
\pscal{p+P}^{s+1} e^{\sigma\abs{p+P}} \abs{\hat B(p+\vec p)}
\\*
+C\sum_{p_i\in 2\pi\Z^D: \pscal{P} > \pscal{p+P}}  
  e^{\sigma\abs{P}} 
\abs{\hat A(\vec p)} \pscal{P}^{s+2} \pscal{p+P}^{-1} e^{\sigma\abs{p+P}}\abs{\hat B(p+\vec p)}.
\end{multline*}
The first term corresponds to the second one in \cref{lem:Psi_to_density_B_sigma} up to replacing $s+2$ by $s+1$. 
For the second term we insert $\frac{\pscal{p+P}^2}{\pscal{p+P}^2}$ in order to obtain summability in the one free variable. Here the weight $\pscal{p+P}^{-1}$ obtained in the splitting is crucial as it turns this into a first power.
More precisely, with Cauchy--Schwarz the contribution from the first term is bounded by
\begin{equation*}
\begin{split}
  \sum_{p,p_i,k_i\in 2\pi\Z^D} 
  e^{2\sigma\abs{K}}
\abs{\hat A(\vec k)}^2
\frac{\pscal{K}^4}{\pscal{P}^4}\pscal{p+P}^{2s+2}e^{2\sigma\abs{p+P}}  \abs{\hat B(p+\vec p)}^2\\
\leq C  \norm{\frac{\pscal{\cP}^2}{1+\cT}A}_\sigma^2 \norm{\frac{\pscal{\cP}^{s+1}}{1+\cT}B}_\sigma^2
\end{split}
\end{equation*}
whereas the one from the second term is
\begin{equation*}
\begin{split}
  \sum_{p,p_i,k_i\in 2\pi\Z^D} 
  e^{2\sigma\abs{K}} \abs{\hat A(\vec k)}^2 
\pscal{K}^{2s+4}\frac{\pscal{p+P}^{2}}{\pscal{p+K}^{4}} e^{2\sigma\abs{p+P}}\abs{\hat B(p+\vec p)}^2\\
\leq C  \norm{\frac{\pscal{\cP}^{s+2}}{1+\cT}A}_\sigma^2 \norm{\frac{\pscal{\cP}}{1+\cT}B}_\sigma^2.
\end{split}
\end{equation*}
Thus, we arrive at
\begin{align*}
 \norm{(1-\Delta)^{s/2} q_3^+}_\sigma
  & \leq C \norm{\frac{\pscal{\cP}^2}{1+\cT}\sqrt{V_{ee}^+}\Psi_1}_\sigma \norm{\frac{\pscal{\cP}^{s+1}}{1+\cT} \sqrt{V_{ee}^+}\partial^i\Psi_2}_\sigma
  \\* & \quad 
  +  C\norm{\frac{\pscal{\cP}^{s+2}}{1+\cT}\sqrt{V_{ee}^+}\Psi_1}_\sigma \norm{\frac{\pscal{\cP}}{1+\cT} \sqrt{V_{ee}^+}\partial^i\Psi_2}_\sigma
  \\
   & 
   \leq C  \norm{\sqrt{V_{ee}^+}e^{\sigma\abs{\cP}}\pscal{\cP}^{2}\Psi_1}_{L^2}
 \norm{\sqrt{V_{ee}^+}e^{\sigma\abs{\cP}}\pscal{\cP}^{s+1}\partial^i\Psi_2}_{L^2}
 \\* & \quad 
  +  C \norm{\sqrt{V_{ee}^+}e^{\sigma\abs{\cP}}\pscal{\cP}^{s+2}\Psi_1}_{L^2}
 \norm{\sqrt{V_{ee}^+}e^{\sigma\abs{\cP}}\pscal{\cP}\partial^i\Psi_2}_{L^2}
 \\
 &
 \leq C \norm{V_{ee}}_{L^2}
 \norm{\pscal{\cP}\Psi_1}_{\sigma}\norm{\pscal{\cP}^{s+1}\Psi_2}_{\sigma}
\\* & \qquad 
 + C \norm{V_{ee}}_{L^2}
 \norm{\pscal{\cP}^{s+1}\Psi_1}_{\sigma}\norm{\pscal{\cP}\Psi_2}_{\sigma}
\end{align*}
similarly to the $q_2$ case, and where we used that $\frac{\pscal{\cP}}{\sqrt{1+\cT}}$ and $\frac{\partial^i}{\sqrt{1+\cT}}$ are bounded Fourier multipliers. The term $q_3^-$ is handled in the same way. 
We conclude the desired.
\end{proof}

\subsection{Time-regularity of \texorpdfstring{$\rho$}{rho} (Proof of \texorpdfstring{\cref{lem:Psi_to_density_C^2}}{Lemma~\ref*{lem:Psi_to_density_C^2}})}
\label{sec.proof.rho.C2}
With \eqref{eq:partialt^2rho} and the bounds in \cref{lemma:boundspreliminaryTi} we can now give the 

\begin{proof}[Proof of \cref{lem:Psi_to_density_C^2}]
Recall that we want to show that $\rho_{\Psi(t)}\in C^2([0,T];\cA_{\sigma'})$ for fixed $0<\sigma' < \sigma$.
We begin with
\begin{align*}
    \partial_t \rho_{\Psi(t)} 
    =- \nabla \cdot j_{\Psi(t)} 
   & = - 2 N \sum_{k=1}^D \partial^k \Im\int \rd X\, \overline{\Psi(t,x,X)} \partial^k\Psi(t,x,X)\\
    &= - 2 \Im \sum_{k=1}^D \partial^k \rho_{\Psi(t),\partial^k\Psi(t)}, 
\end{align*}
where we used the continuity equation from \cref{prop.force.balance} and recall $\partial^k = \partial_{x^k}$. With \cref{lem:Psi_to_density_B_sigma} and using $\tilde \sigma =\frac{\sigma+\sigma'}{2}$ it follows immediately
\begin{align*}
\norm{\partial_t \rho_{\Psi(t)}}_{\sigma'}
&\leq  \frac{C}{\sigma-\tilde\sigma} \sum_{k=1}^D\norm{\rho_{\Psi(t),\partial^k\Psi(t)}}_{\tilde\sigma}
\leq  \frac{C}{\sigma-\tilde\sigma} \sum_{k=1}^D\norm{\frac{\pscal{\cP}^{2}}{1+\cT}\Psi(t)}_{\tilde\sigma} \norm{\frac{\pscal{\cP}^2}{1+\cT}\partial^k\Psi(t)}_{\tilde\sigma}\\
&\leq  \frac{C}{(\sigma-\sigma')^2} \norm{\Psi(t)}_\sigma^2, 
\end{align*}
where we also used \cref{lem.bdd.gradient} in the last inequality.

Next, we consider the second derivative.
Recalling the formula in \eqref{eq:partialt^2rho} for $\partial_t^2\rho_\Psi$ we bound the terms on the right-hand-side in $\cA_{\sigma'}$ as follows.
We have
$$
\norm{\Delta^2\rho_\Psi}_{\sigma'} \leq \frac{C}{(\sigma-\sigma')^4} \norm{\rho_\Psi}_{\sigma}
\leq \frac{C}{(\sigma-\sigma')^4} \norm{\Psi}_{\sigma}^2, 
$$
where we used \cref{lem.bdd.gradient} and \cref{lem:Psi_to_density_B_sigma}.
Next, using also the product rule (\cref{lemma:multiplication})
and using again $\tilde\sigma = (\sigma' + \sigma)/2$, 
\begin{multline*} 
\norm{K_{\rho_\Psi} V}_{\sigma'}
\leq \frac{C}{\sigma-\tilde \sigma} \norm{\rho_\Psi}_{\tilde \sigma} \norm{ \nabla V}_{\tilde \sigma}
\leq \frac{C}{(\sigma-\tilde \sigma)(\tilde \sigma-\sigma')} \norm{\rho_\Psi}_{\tilde \sigma} \norm{V}_{\sigma}
\\
\leq \frac{C}{(\sigma-\sigma')^2} \norm{V}_{\sigma} \norm{\Psi}_{\sigma}^2.
\end{multline*}
We bound $Q_\Psi$ using \cref{lemma:boundspreliminaryTi} with $\Psi_1=\Psi_2=\Psi$. 
This gives (using again also \cref{lem.bdd.gradient})
\begin{equation*}
  \norm{Q_{\Psi}}_{\sigma'} 
  \leq C \norm{\pscal{\cP}^{3}\Psi}_{\sigma'}\norm{\pscal{\cP}\Psi}_{\sigma'} 
  \leq\frac{C}{(\sigma-\sigma')^4} \norm{\Psi}_\sigma^2. 
\end{equation*}
Putting it all together, we find
$$
\norm{\partial_t^2 \rho_\Psi}_{\sigma'} \leq \frac{C}{(\sigma-\sigma')^{4}} (\norm{V}_{\sigma}+1)
    \norm{\Psi}_{\sigma}^2
    < \infty. 
$$
To see the continuity in time we consider $\partial_t^2\rho_\Psi(t) - \partial_t^2\rho_\Psi(t')$. Writing this using \eqref{eq:partialt^2rho} we arrive at 
\begin{align*}
    \partial_t^2\rho_\Psi(t) - \partial_t^2\rho_\Psi(t')
    & = -\Delta^2(\rho_\Psi(t) - \rho_\Psi(t')) - 2K_{\rho_\Psi(t)}V(t) + 2 K_{\rho_\Psi(t')}V(t')
\\ & \quad 
    + \Re Q_{\Psi(t) + \Psi(t'), \Psi(t) - \Psi(t')}
\end{align*}
where $Q_{\Psi_1,\Psi_2}$ is as in \eqref{eq:T_Psi_chi}. Here we used that $Q_{\Psi,\Psi}=Q_\Psi$ and $\overline{Q_{\Psi_1,\Psi_2}} = Q_{\Psi_2,\Psi_1}$. 
From \cref{lemma:boundspreliminaryTi} the continuity in time follows by a similar computation as the one above using that $\rho, V$, and $\Psi$ are all continuous in time (with values in $\cA_\sigma, \cA_\sigma$, and $\cB_\sigma$, respectively).
\end{proof}

\section{The inverse problem (Proof of \texorpdfstring{\cref{thm.main.V}}{Theorem~\ref*{thm.main.V}})}
In this section we give the proof of the inverse problem, \cref{thm.main.V}. 
We first reinterpret \eqref{eq:partialt^2rho} as an equation for $V$ in terms of $\rho$ and $\Psi$. 
As we are looking for the potential $V$ such that the evolution reproduces a given density $\rho$, we can replace all occurrences of $\rho_\Psi$ with $\rho$ and obtain
\begin{equation}\label{eq:KrhoV}
    K_\rho V = \frac{1}{2}(-\partial_t^2 \rho - \Delta^2 \rho + Q_\Psi),
\end{equation}
where we recall that $K_\rho f = -\nabla\cdot(\rho\nabla f)$ and $Q_\Psi = Q_{\Psi,\Psi}$ is given in \eqref{eqn.def.cT.Psi} (and \eqref{eq:T_Psi_chi}). 
Assuming that $K_\rho$ is invertible (we prove this in \cref{lem.Krho} below;  note that the right-hand-side of \eqref{eq:KrhoV} has zero integral) we thus define the nonlinear potential as 
\begin{equation}\label{eq:VPsi}
    V_\Psi(t) = \frac{1}{2} K_{\rho(t)}^{-1}\left(-\partial_t^2\rho(t) - \Delta^2\rho(t) + Q_{\Psi}\right). 
\end{equation}
The nonlinear Schrödinger equation of interest then reads 
\begin{equation}\label{eq:SEnonlinearV}
    i\partial_t \Psi (t) 
    = \cH_{V_{\Psi(t)}(t)}\Psi(t), 
    \quad \Psi(0) = \Psi_0 . 
\end{equation}
The key step in proving \cref{thm.main.V} is proving that the nonlinear Schrödinger equation \eqref{eq:SEnonlinearV} is solvable. We prove

\begin{proposition}[The nonlinear Schrödinger equation]
\label{prop.solve.NLS}
Let $\rho\in C^2([0,T];\cA_{\sigma})$  with $\int \rho =N$ and $\rho(t)\geq c > 0$ and $\Psi_0\in\cB_\sigma$. 
Then, there exists $0 < T'\leq T$ and $0 < \sigma'\leq \sigma$ 
such that the Schrödinger equation \eqref{eq:SEnonlinearV} admits a unique solution $\Psi\in C^0([0,T'];\cB_{\sigma'})$. 
\end{proposition}

The main idea in the proof of \cref{prop.solve.NLS} is that, after factoring out the evolution given by the linear (translation-invariant) term $\cH_0$, the resulting equation morally has a loss of one derivative and can be solved using a standard nonlinear Cauchy--Kovalevskaya theorem \cite{Nishida-77,Nirenberg-72,BaoGou-77,BaoGou-77b}.
We explain the details in due time.

The proof of \cref{prop.solve.NLS} will occupy most of the rest of the paper. At the end (in \cref{sec.final.steps.proof.thm.main}), we explain how to prove \cref{thm.main.V} using \cref{prop.solve.NLS}. 
First, however, we will discuss our assumptions on $\rho$.

\subsection{The reciprocal of the density is analytic}
Let $\rho \in C^2([0,T];\cA_\sigma)$ be non-vanishing, i.e., as in the statement of \cref{thm.main.V} or \cref{prop.solve.NLS}. 
We claim that, up to reducing the value of $\sigma$, 
we can additionally assume that 
\begin{equation}
    \sup_{0\leq t \leq T} 
    \left\{\norm{\sqrt{\rho(t)}^{-1}}_{\sigma}
    +\norm{\Delta^2\rho(t)}_{\sigma}
    \right\}
    < \infty. \label{eq.ass.rho} 
\end{equation}
Proving this and related bounds is the content   of this subsection. 
The bound on $\Delta^2\rho$ follows from \cref{lem.bdd.gradient}. 
The bound on $1/\sqrt{\rho}$ is the content of the following lemma.

\begin{lemma}[Analyticity of $1/\sqrt\rho$]
\label{lemma.1/rho.analytic}
Let $\rho\in C^0([0,T]; \cA_\sigma)$ satisfy $\rho(t) \geq c > 0$. 
Then, there exists some $\sigma'>0$ such that $1/\sqrt{\rho} \in C^0([0,T];\cA_{\sigma'})$.  
\end{lemma}

The main ingredient is the time-independent statement.

\begin{lemma}[Analyticity of the reciprocal]\label{lemma.1/f.analytic}
Let $f\in \cA_\sigma$ be real-valued with $f(x) > 0$ for all $x\in \Lambda$. 
Then, there exists $\sigma' > 0$ such that $1/\sqrt{f} \in \cA_{\sigma'}$. Concretely, for $\sigma'>0$ sufficiently small, we have the bound
\begin{equation*}
    \norm{\frac{1}{\sqrt f}}_{\sigma'} \leq C (\sigma')^{-2-D/2} 
    \left(\norm{\frac{1}{f}}_{L^\infty}^{-1} - \frac{C\sigma'}{\sigma - 2\sigma'} \norm{f}_\sigma\right)^{-1/2}. 
\end{equation*}
\end{lemma}

\begin{remark}\label{rmk.sqrt.f.analytic}
It follows that $\sqrt{f}, 1/f\in \cA_{\sigma'}$. Indeed, by \cref{lemma:multiplication}, we have $\norm{\sqrt{f}}_{\sigma'} \leq C \norm{f}_{\sigma'} \norm{1/\sqrt{f}}_{\sigma'}$ and $\norm{1/f}_{\sigma'}\leq C \norm{1/\sqrt{f}}_{\sigma'}^2$. 
\end{remark}

\begin{remark}[Wiener's $1/f$ theorem]
\cref{lemma.1/f.analytic} is reminiscent of a weighted version of Wiener's $1/f$ theorem (also known as the Wiener--Lévy theorem or Wiener's Tauberian theorem) that for a non-vanishing function $f:[0,1]\to \C$ whose Fourier series is absolutely convergent, the reciprocal function has an absolutely convergent Fourier series \cite{Wiener-32}, \cite[Thm.~18.21]{Rudin-87}. Weighted versions are discussed in \cite{BhaDed-03,BhaDabDed-09,Domar-1956}. 
The weighted version (convergence in $\ell^1(\omega)$ with a weight $\omega : 2\pi\Z\to [0,\infty)$) holds if and only if the weight $\omega$ satisfies the Beurling condition $\sum_{n\in 2\pi\Z} \frac{\log \omega(n)}{1+n^2} < \infty$ \cite[Thm.~2.11]{Domar-1956}, \cite{BhaDabDed-09}. 
The space $\cA_\sigma$ is a weighted $\ell^2$-space (in Fourier space) with weight $\omega(n) = e^{2\sigma |n|}(1+n^2)^2$. 
This grows exponentially, and thus does not satisfy the Beurling condition. We thus cannot hope in general to have $1/f\in \cA_\sigma$. As \cref{lemma.1/f.analytic} shows, however, we do in fact have $1/f\in \cA_{\sigma'}$ for a $\sigma' < \sigma$.

More abstract versions of Wiener's $1/f$ theorem in terms of ideals of Banach algebras can be found in \cite[Sec.~VIII.6]{Katznelson-04} and \cite[Ch.~6]{ReiSte-00}. 
(Note that the Banach space $\cA_\sigma$ is in fact a Banach algebra (it is closed under products) by \cref{lemma:multiplication}.)
\end{remark}

From \eqref{eq.ass.rho} we further have the following bounds. 

\begin{lemma}[Further bounds]\label{lemma:more_bounds_rho}
Let $\rho\in C^2([0,T];\cA_\sigma)$ satisfy \eqref{eq.ass.rho}. 
Then, 
\begin{equation*}
    \sup_{0\leq t \leq T}
    \left\{
    \norm{\frac{1}{\rho(t)}}_{L^\ii} 
    +\norm{\Delta\sqrt{\rho(t)}^{-1}}_{\sigma}
    +\norm{\sqrt{\rho(t)}}_{\sigma}
    +\norm{\Delta\sqrt{\rho(t)}}_{\sigma}
    +\norm{\nabla\rho(t)}_{\sigma}
    \right\}
    < \infty. 
\end{equation*}
\end{lemma}

We now give the proofs of the above claims.

\begin{proof}[Proof of \cref{lemma.1/rho.analytic}]
Let $\sigma' = \min\{\frac{\sigma}{4}, \frac{c\sigma}{4C \sup_{t\leq T}\norm{\rho(t)}_\sigma}\} $.    
Then for any $0\leq t\leq T$ we have
\begin{align*}
    \norm{\frac{1}{\rho(t)}}_{L^\infty}^{-1} - \frac{C \sigma'}{\sigma - 2\sigma'} \norm{\rho(t)}_\sigma
    & \geq c - \frac{2C}{\sigma} \norm{\rho(t)}_{\sigma} \frac{c\sigma}{4C \sup_{s \leq T}\norm{\rho(s)}_\sigma} \geq \frac{c}{2}. 
\end{align*}
By \cref{lemma.1/f.analytic} we thus have $1/\sqrt \rho \in L^\infty([0,T]; \cA_{\sigma'})$ as desired. The continuity in time follows by interpolation of the norms. More precisely, for any $f$ we have 
\begin{equation*}
    \norm{f}_{\sigma'/2}^2
    = \norm{e^{|\nabla|\sigma'/2}f}_{H^2}^2
    \leq 
    \norm{e^{|\nabla|\sigma'}f}_{H^2}
    \norm{f}_{H^2}. 
\end{equation*}
Taking $f=1/\sqrt{\rho(t)} - 1/\sqrt{\rho(t')}$ and noting that clearly $1/\sqrt{\rho}$ is continuous in time with values in $H^2$ (simply compute the derivatives explicitly), we see that $1/\sqrt{\rho}$ is continuous in time with values in $\cA_{\sigma'/2}$. Relabeling $\sigma'$ we conclude the desired. 
\end{proof}

\begin{proof}[Proof of \cref{lemma.1/f.analytic}]
Recall that $f\in \cA_\sigma$ is analytic. Hence, we can extend its domain of definition slightly into the complex plane by analytic continuation. More concretely, we define by Fourier inversion
\begin{equation*}
    \tilde f(z) = \sum_{k\in  2\pi\Z^D} e^{ikz} \hat f(k)
\end{equation*}
for any $z \in \Lambda_\sigma = \Lambda + i \{y\in \R^D : |y| \leq \sigma\}$. 
The sum converges absolutely for such $z$. Indeed, 
\begin{align*}
    \sum_{k\in  2\pi\Z^D} |e^{ikz} \hat f(k)| 
& 
    \leq  \sum_{k\in  2\pi\Z^D} e^{\sigma |k|} |\hat f(k)| \frac{1+k^2}{1+k^2}
    \leq C \norm{f}_\sigma \sqrt{ \sum_{k\in  2\pi\Z^D} \frac{1}{(1+k^2)^2}}
    < \infty. 
\end{align*}
This calculation further shows that $\Vert \tilde f \Vert_{L^\infty(\Lambda_\sigma)} \leq C \norm{f}_\sigma$.

We next claim that 
\begin{equation}\label{eqn.norm.extend.complex}
    \norm{f}_{\sigma'} \leq \frac{C}{(\sigma-\sigma')^{2+D/2}} \Vert{\tilde f}\Vert_{L^\infty(\Lambda_\sigma)}
\end{equation}
for any $0 < \sigma' < \sigma$.

To prove this, consider the formula for the Fourier coefficients and shift the integral slightly into the complex plane by adding an imaginary value $+i\eta$ for some $|\eta|\leq \sigma$. Noting that the segments joining the real domain $\Lambda$ and the domain $\Lambda + i\eta$ are in fact the same by periodicity and thus cancel, we see that 
\begin{equation*}
\hat f(k) = \int_\Lambda e^{-ikx} f(x) \, \rd x = \int_\Lambda \tilde f(x+i\eta) e^{-ik(x+i\eta)} \, \rd x. 
\end{equation*}
Choosing $\eta = -\sigma \frac{k}{|k|}$ we get 
\begin{equation*}
\abs{\hat f(k)} 
\leq 
|\Lambda| \Vert{\tilde f}\Vert_{L^\infty(\Lambda_\sigma)} e^{-\sigma|k|} 
= 
\Vert{\tilde f}\Vert_{L^\infty(\Lambda_\sigma)} e^{-\sigma|k|} . 
\end{equation*}
Thus, for any $0 < \sigma' < \sigma$, we have
\begin{equation*}
\norm{f}_{\sigma'}^2 
    \leq  \sum_{k\in  2\pi\Z^D} e^{-2(\sigma-\sigma')|k|} (1+k^2)^2 \Vert{\tilde f}\Vert_{L^\infty(\Lambda_\sigma)}^2
    \leq \frac{C}{(\sigma-\sigma')^{D+4}} \Vert{\tilde f}\Vert_{L^\infty(\Lambda_\sigma)}^2. 
\end{equation*}
This proves \eqref{eqn.norm.extend.complex}. 

Next, we show that the analytically extended $\sqrt{f}$ is nonvanishing.
By uniqueness of analytic continuation we have that 
$\nabla \tilde f = \tilde{\nabla f}$ and that $1/{\sqrt{\tilde f}} = \tilde{1/\sqrt{f}}$. For the complex square root, we can put the branch-cut on the negative real line. Since $f$ is bounded away from $0$, the complex square root is single-valued, at least for sufficiently small $\sigma'$.

By \cref{lem.bdd.gradient} and the bound above, we have, for any $0< \sigma' < \sigma$,
\begin{equation*}
    \bigl\Vert {\nabla \tilde f} \bigr\Vert_{L^\infty(\Lambda_{\sigma'})} 
    \leq C \norm{\nabla f}_{\sigma'}
    \leq \frac{C}{\sigma-\sigma'} \norm{f}_\sigma. 
\end{equation*}
Thus, for any $\tau < \sigma$ we have for any $z\in \Lambda_\tau$
\begin{equation*}
    \abs{\tilde f(z)} \geq \inf_{x\in \Lambda} |f(x)| - \frac{C\tau}{\sigma-\tau} \norm{f}_\sigma. 
\end{equation*}
For $\tau$ sufficiently small we thus have $\Bigl\vert{\tilde{\sqrt{f(z)}}}\Bigr\vert =  \Bigl\vert{\tilde f(z)}\Bigr\vert^{1/2} \geq \sqrt{\frac{1}{2} \inf_{x\in \Lambda} |f(x)|} > 0$ for all $z\in \Lambda_\tau$. In particular $\bigl\Vert {\tilde{1/\sqrt f}}\bigr\Vert_{L^\infty(\Lambda_\tau)} < \infty$. 
We conclude from \eqref{eqn.norm.extend.complex} that for any $\sigma' < \tau$ we have $1/\sqrt f\in \cA_{\sigma'}$ as claimed.
Choosing in particular $\sigma' = \tau/2$ we have the bound stated in the lemma.
\end{proof}

\begin{proof}[Proof of \cref{lemma:more_bounds_rho}]
Note that by Sobolev's inequality \cite[Thms.~8.3 and 8.5]{LieLos-01} and \cref{lemma:multiplication} we have
$$
\norm{\frac{1}{\rho(t)}}_{L^\ii}  \leq C \norm{\frac{1}{\rho(t)}}_{H^2} 
\leq C \norm{\frac{1}{\sqrt{\rho(t)}}}_{\sigma}^2 .
$$
The other bounds follow essentially using the Cauchy--Schwarz inequality and the product rule \cref{lemma:multiplication}.
For instance, for $\norm{\nabla\rho(t)}_{\sigma}$ we write
\begin{align*}
\norm{\nabla\rho(t)}_{\sigma}^2
&=
\sum_{p\in 2\pi\Z^D} p^2 e^{2\sigma|p|} \left(1+ p^2\right)^2 \big|\hat{\rho}(p)\big|^2\\
&\leq C
\sum_{p\in 2\pi\Z^D} (1+ \abs{p}^8) e^{2\sigma|p|} \left(1+ p^2\right)^2 \big|\hat{\rho}(p)\big|^2
\leq C \norm{\rho(t)}_{\sigma}^2+ C \norm{\Delta^2\rho(t)}_{\sigma}^2
\end{align*}
and, by explicitly computing the Laplacian,
\begin{align*}
\norm{\Delta\sqrt{\rho(t)}}_{\sigma}
&\leq C \norm{\Delta\rho(t)}_{\sigma}\norm{\sqrt{\rho(t)}^{-1}}_{\sigma}
+ C \norm{\nabla\rho(t)}_{\sigma}^2\norm{\sqrt{\rho(t)}^{-1}}_{\sigma}^3.
\end{align*}
The others follow analogously.
\end{proof}

\subsection{Properties of the operator \texorpdfstring{$K_\rho$}{Krho}}
To solve \eqref{eq:KrhoV} for $V$, we need to invert the operator $K_\rho$ given by (recall \eqref{eqn.def.Krho}) $K_\rho f = - \nabla\cdot(\rho\nabla f)$. 
We collect here some properties of the operator $K_\rho$.

First, note that, for any function $f$ we have 
$\int K_\rho f \, \rd x = -\int \nabla\cdot(\rho\nabla f) \, \rd x = 0$ by the divergence theorem. Additionally, $K_\rho(f+c) = K_\rho f$ for any constant $c$. 
Define $P^\perp = \mathbbm{1} - \ket{1}\bra{1}$ as the projection onto the space of functions orthogonal to constants. 
The above proves that $K_\rho$ commutes with $P^\perp$, i.e., it preserves the space $\cA_\sigma^\perp = P^\perp \cA_\sigma$.

We prove the following bound on the norm of the inverse essentially stating that $K_\rho^{-1}$ gains two derivatives.

\begin{lemma}[{Invertibility of $K_\rho$}]
\label{lem.Krho}
Let $\rho\in C^2([0,T];\cA_{\sigma})$ satisfy \eqref{eq.ass.rho}.
Then, $K_\rho$ is invertible on the space $\cA_\sigma^\perp$ and
\begin{equation}\label{eq:resolventKDelta}
\sup_{t\leq T}\norm{K_{\rho(t)}^{-1}(1-\Delta)}
_{\cA_\sigma^\perp \to \cA_\sigma^\perp} < \infty.
\end{equation}
As a result, 
\begin{equation*}
\sup_{t\leq T}
\norm{K_{\rho(t)}^{-1}m(-i\nabla)}_{\cA_\sigma^\perp \to \cA_\sigma^\perp} < \infty
\end{equation*}
for any polynomial $m$ of degree at most $2$.
Furthermore, we have the following continuity result on the resolvent:
\begin{multline}
\label{eq:K_resolvent_continuous}
    \norm{(K_{\rho_1}^{-1} -K_{\rho_2}^{-1})(1-\Delta)}_{\cA^\perp_\sigma\to \cA^\perp_\sigma}
    \leq  C \norm{K_{\rho_1}^{-1}(1-\Delta)}_{\cA^\perp_\sigma\to \cA^\perp_\sigma}
     \norm{K_{\rho_2}^{-1}(1-\Delta)}_{\cA^\perp_\sigma\to \cA^\perp_\sigma}
     \\
     \times\norm{(1-\Delta)^{1/2}(\rho_1-\rho_2)}_{\sigma}
    .
\end{multline}
\end{lemma}

\begin{proof}
For the sake of simplicity of notation, we suppress from the notation the time-dependence. 

Observe that $K_\rho$ is self-adjoint on $L^2$ and gapped as (recall that on the torus $\Lambda=[0,1]^D$, the Laplacian has a spectral gap $4\pi^2$)
\begin{equation}\label{eqn.Krho.gapped}
K_\rho \geq \left(\min_{x\in \Lambda} \rho(x)\right) (-\Delta) \geq 4\pi^2 \left(\min_{x\in \Lambda}  \rho(x)\right)   P^\perp.
\end{equation}
The existence of the inverse as a bounded operator on $P^\perp L^2$ is thus immediate.

Note that we are working on the one-particle space here for which the norm in fact satisfies the product rule (cf.~\cref{lemma:multiplication}). 
We first show \eqref{eq:resolventKDelta}.
Note that
\begin{equation*}
  K_\rho 
  = (-\sqrt{\rho}\Delta + (\Delta \sqrt{\rho}))\sqrt{\rho}
  = \sqrt{\rho}(1-\Delta) \sqrt{\rho} 
     + ((\Delta\sqrt{\rho}) - \sqrt\rho) \sqrt{\rho} 
     .
\end{equation*}
Suppose $K_\rho f= (1-\Delta) g$ for $f$ with $\int f =0 = \int g$.
Then
\begin{equation*}
\begin{split}
\sqrt\rho f = (1-\Delta)^{-1}\frac{(1-\Delta) g}{\sqrt\rho} - (1-\Delta)^{-1} (\Delta \sqrt{\rho}-\sqrt{\rho}) f, 
\end{split}
\end{equation*}
such that
\begin{equation}\label{eqn.bdd.sqrt.rho.f.Kinv}
  \norm{\sqrt\rho f}_{\sigma}
  \leq
  \norm{(1-\Delta)^{-1}\frac{(1-\Delta) g}{\sqrt\rho}}_{\sigma} 
  	+\norm{(1-\Delta)^{-1} (\Delta \sqrt{\rho}-\sqrt{\rho}) f}_{\sigma}.
\end{equation}
The first term in \eqref{eqn.bdd.sqrt.rho.f.Kinv} is bounded as
\begin{equation}
\begin{split}
  \norm{(1-\Delta)^{-1}\frac{(1-\Delta) g}{\sqrt\rho}}_{\sigma} 
  &\leq
  \norm{(1-\Delta)^{-1}\frac{1}{\sqrt\rho}(1-\Delta) }_{\sigma\to \sigma} \norm{g}_{\sigma}\\
    &\leq C
  \left(\norm{\frac{1}{\sqrt\rho}}_{\sigma} + \norm{\nabla\frac{1}{\sqrt\rho}}_{\sigma}+\norm{\Delta\frac{1}{\sqrt\rho}}_{\sigma}
  \right)\norm{g}_{\sigma} ,
\end{split}
\label{eqn.bdd.Kinv.first.term}
\end{equation}
where we used \cref{lemma:multiplication} in the first step and the formula
$$(1-\Delta)^{-1}\frac1{\sqrt\rho}(1-\Delta)=\frac1{\sqrt\rho}+\frac1{1-\Delta} \left(-\Delta \frac1{\sqrt\rho}\right)+\frac{2\nabla}{1-\Delta}\cdot\left(\nabla \frac1{\sqrt\rho}\right).$$
For the second term in \eqref{eqn.bdd.sqrt.rho.f.Kinv} we introduce a splitting into high and low momenta.
\begin{align*}
 &\norm{(1-\Delta)^{-1} (\Delta \sqrt{\rho}-\sqrt{\rho}) f}_{\sigma}
 \\
  &\quad \leq 
  	\norm{\frac{\mathbbm{1}(-\Delta\leq R^2)}{1-\Delta} (\Delta \sqrt{\rho}-\sqrt{\rho}) f}_{\sigma}
  	+\norm{\frac{\mathbbm{1}(-\Delta> R^2)}{1-\Delta} (\Delta \sqrt{\rho}-\sqrt{\rho}) f}_{\sigma} \\
   &\quad \leq  e^{\sigma R} \norm{(\Delta \sqrt{\rho}-\sqrt{\rho}) f}_{L^2}
  	+(R^2+1)^{-1}\norm{\frac{\Delta \sqrt{\rho}-\sqrt{\rho}}{\sqrt{\rho}} \sqrt\rho f}_{\sigma}\\
  &\quad 
  \leq e^{\sigma R}
  	\norm{\Delta \sqrt{\rho}-\sqrt{\rho}}_{L^\infty}\norm{ f}_{L^2}
  	+\tilde C(R^2+1)^{-1}\norm{\frac{\Delta \sqrt{\rho}-\sqrt{\rho}}{\sqrt{\rho}} }_{\sigma} \norm{\sqrt\rho f}_{\sigma}, 
\end{align*}
where $\tilde C$ is the constant from \cref{lemma:multiplication}. 
Here we used that the Fourier multipliers act explicitly.
Choosing $R^2 = 2\tilde C \norm{\frac{\Delta \sqrt{\rho}-\sqrt{\rho}}{\sqrt{\rho}}}_{\sigma}$
we conclude the bound 
\begin{equation}\label{eqn.bdd.Kinv.second.term}
    \norm{(1-\Delta)^{-1} (\Delta \sqrt{\rho}-\sqrt{\rho}) f}_{\sigma}
    \leq C \norm{f}_{L^2} + \frac{1}{2} \norm{\sqrt{\rho}f}_\sigma. 
\end{equation}
Note that $\frac{1}{2} \norm{\sqrt{\rho}f}_\sigma$ on the right hand side has to be regularized in order to be absorbed in the left hand side; this can be done via a momentum cutoff which then has to be removed similarly to earlier justifications of such computations. We omit the details.
Using the Poincaré inequality~\cite[Thm.~8.11]{LieLos-01}, \eqref{eqn.Krho.gapped}, and $K_\rho f = (1-\Delta)g$ we have 
\begin{multline*}
    \norm{f}_{L^2}^2 
    \leq C \norm{\nabla f}_{L^2}^2 
    \leq \frac{C}{\min_x \rho(x)} \pscal{f, K_\rho f}
    \leq C \norm{f}_{L^2} \norm{(1-\Delta)g}_{L^2}
    \\ 
    \leq \frac{1}{2}\norm{f}_{L^2}^2 + C \norm{(1-\Delta)g}_{L^2}^2
    \leq \frac{1}{2}\norm{f}_{L^2}^2 + C \norm{g}_{\sigma}^2.
\end{multline*}
Combining with \eqref{eqn.bdd.sqrt.rho.f.Kinv}, \eqref{eqn.bdd.Kinv.first.term}, and \eqref{eqn.bdd.Kinv.second.term} we find 
\begin{equation*}
    \norm{\sqrt\rho f}_\sigma \leq C \norm{g}_\sigma. 
\end{equation*}
Thus, by \cref{lemma:multiplication}, we have 
$\norm{f}_\sigma \leq C \norm{\sqrt{\rho}^{-1}}_\sigma\norm{\sqrt\rho f}_\sigma < \infty$ as desired.

Next, the second bound is an immediate consequence of \eqref{eq:resolventKDelta} noting that for any such polynomial $m$ we have $\norm{(1-\Delta)^{-1}m(-i\nabla)}_{\sigma\to \sigma}<\infty$.

Lastly, we consider the effect of perturbing the density and compute
\begin{align*}
   & 
    \norm{(K_{\rho_1}^{-1} -K_{\rho_2}^{-1})(1-\Delta)}_{\cA^\perp_\sigma\to \cA^\perp_\sigma}
    =
    \norm{K_{\rho_1}^{-1} (K_{\rho_1} -K_{\rho_2})K_{\rho_2}^{-1}(1-\Delta)}_{\cA^\perp_\sigma\to \cA^\perp_\sigma}\\
        & \quad  \leq 
    \norm{K_{\rho_1}^{-1}\nabla \cdot (\rho_1-\rho_2)\nabla}_{\cA^\perp_\sigma \to \cA^\perp_\sigma}
    \norm{K_{\rho_2}^{-1}(1-\Delta)}_{\cA^\perp_\sigma\to \cA^\perp_\sigma}\\
& \quad \leq C
     \norm{K_{\rho_1}^{-1}(1-\Delta)}_{\cA^\perp_\sigma\to \cA^\perp_\sigma}
     \left(\norm{\rho_1-\rho_2}_{\sigma} +\norm{\nabla(\rho_1-\rho_2)}_{\sigma}\right)
    \norm{K_{\rho_2}^{-1}(1-\Delta)}_{\cA^\perp_\sigma\to \cA^\perp_\sigma}
    .
\end{align*}
Here we used that by \cref{lemma:multiplication}, the multiplication with $\rho_1-\rho_2$ is a bounded operator (on $\cA_\sigma$) with norm at most the norm of $\rho_1-\rho_2$.
\end{proof}

\subsection{Duhamel form and (space-time) norms}
To study the equation \eqref{eq:SEnonlinearV} we write it instead in Duhamel form. 
We look for solutions $\Psi\in C^0([0,T'];\cB_{\sigma'})$. 
From \cref{lem.VPsi.bdds} below, it follows that, in particular, 
$V_\Psi\in C^0([0,T'],C^2(\Lambda))$, and thus that $\Psi$ solves the nonlinear equation~\eqref{eq:SEnonlinearV} if and only if it solves the Duhamel formulation
\begin{equation}
\label{eqn.NLS.V.Duhamel}
    \Psi(t) = e^{-it\cH_0} \Psi_0 -i \int_0^t e^{-i(t-s)\cH_0} \sum_{j=1}^N V_{\Psi(s)}(s,x_j) \Psi(s) \, \rd s. 
\end{equation}
Factoring out explicitly the evolution governed by $\cH_0$, we rewrite it in terms of $\tilde\Psi=e^{it\cH_0}\Psi$ and obtain 
\begin{equation}\label{eqn.NLS.V.Duhamel.tildePsi}
    \tilde\Psi(t) = \Psi_0 -i \int_0^t e^{is\cH_0} \sum_{j=1}^N V_{e^{-is\cH_0}\tilde\Psi(s)}(s,x_j) e^{-is\cH_0}\tilde \Psi(s) \, \rd s. 
\end{equation}
This is the equation we will study by a fixed-point method. 
We follow the approach of \cite{BaoGou-77,BaoGou-77b} and show that the integral term in \eqref{eqn.NLS.V.Duhamel.tildePsi} is a contraction in an appropriate norm.
We next define the relevant norms.

As mentioned above, the key step in the proof of \cref{prop.solve.NLS} is the fact the nonlinear term morally loses one derivative. To state this we define the norm 
\begin{definition}
For any $\sigma\geq 0$ and any $\Psi\in L^2_a(\Lambda^N)$  we define $\nnorm{\Psi}_\sigma = \norm{\pscal{\cP}\Psi}_\sigma$. 
\end{definition}

Clearly, as in \cref{lem.bdd.gradient}, this norm is analytic with respect to the total momentum,
	\begin{equation*}
  \nnorm{\abs{\cP} \Psi}_{\sigma'} \leq \frac{1}{\sigma-\sigma'} \nnorm{\Psi}_{\sigma}.
\end{equation*}

\begin{remark}[``Linearizing'' the nonlinearity]
The reason for using this norm is that while the terms in \cref{lemma:boundspreliminaryTi} are bounded in terms of a first power of the total momentum acting on the wavefunction (in the case $s=0$), two such terms appear on the right-hand-side, thus yielding a loss $1/(\sigma-\sigma')^2$ on the scale of spaces $\cB_\sigma$. 
Adding the additional weight $\pscal{\cP}$ turns this into a linear loss; this can be seen  more precisely in \cref{lem.VPsi.bdds} below.
This modification is essentially the same as doubling the number of variables by turning an equation containing $\abs{\nabla f}^2$ into an equation for the pair $g=(f, \nabla f)$ where the gradient $\nabla g$ now appears linearly. This is a standard technique in the study of nonlinear PDEs~\cite[Sec.~3]{Nirenberg-72}.
\end{remark}

We then define the norm on trajectories as follows. 

\begin{definition}[{Baouendi--Goulaouic space~\cite{BaoGou-77,BaoGou-77b}}]
\label{def.BG.space}
For any $0 < T' \leq T$ and $\sigma > 0$ we define the norm
\begin{equation*}
\nnorm{\chi}_{T',\sigma}:=\sup_{\substack{0\leq\sigma'<\sigma\\ 0\leq t\leq T'\frac{\sigma-\sigma'}{\sigma}}}\left\{\nnorm{\partial_t \chi}_{\sigma'}\frac{\sigma-\sigma'}{\sigma}\sqrt{1-\frac{t\sigma}{T'(\sigma-\sigma')}}\right\}
\end{equation*}
and the associated Banach space 
\begin{equation*}
    \cB_{T',\sigma} = \{\chi \in C^1([0,T']; H^2) : \chi(0) = 0, \nnorm{\chi}_{T',\sigma} < \infty\}. 
\end{equation*}
\end{definition}

This norm dominates the $\nnorm{\cdot}_{\sigma'}$-norms for appropriate $\sigma'$ as follows.

\begin{lemma}[{cf.~\cite[Lems.~1 and 2]{BaoGou-77}}]\label{lemma:BGnormbounds}
Let $\chi \in C^0([0,T'];H^2)$ with $\chi(0)=0$. Then,  
\begin{align}
\sup_{0\leq t \leq T'}\nnorm{\chi(t)}_{\sigma_{T'}(t)} 
& \leq CT'\nnorm{\chi}_{T',\sigma}, 
\label{eq:BG_norm1}
\\\label{eq:BGnormPchi}
\sup_{\substack{0\leq\sigma'<\sigma\\ 0\leq t\leq T'\frac{\sigma-\sigma'}{\sigma}}}
\frac{\sigma-\sigma'}{\sigma}\sqrt{1-\frac{t\sigma}{T'(\sigma-\sigma')}}
\nnorm{\pscal{\cP}\chi(t)}_{\sigma'}
& \leq CT'\nnorm{\chi}_{T',\sigma} 
, 
\end{align}
where $\sigma_{T'}(t) = \sigma(1-t/T')$. 
\end{lemma}

Note that the norm on the left-hand-side of \eqref{eq:BG_norm1} can be expressed as
\begin{equation}\label{eqn.sigma.t.sup.norm.order}
\sup_{0\leq t \leq T'}\nnorm{\chi(t)}_{\sigma_{T'}(t)} 
=
\sup_{0 \leq t \leq T'} \sup_{0 < \sigma' \leq \sigma(1-t/T')} \nnorm{\chi(t)}_{\sigma'} 
=
\sup_{0 < \sigma' < \sigma} \sup_{0 \leq t \leq T' \frac{\sigma - \sigma'}{\sigma}} \nnorm{\chi(t)}_{\sigma'}
.
\end{equation}

\begin{proof}
\Cref{eq:BG_norm1} is exactly \cite[Lem.~1]{BaoGou-77}. We give a proof here for convenience of the reader. 

For a $\chi$ with $\chi(0)=0$ we have $\chi(t) = \int_0^t \partial_s \chi(s) \, \rd s $. Then, with $\sigma_{T'}(t) = \sigma(1-t/T')$ we have 
\begin{align*}
\nnorm{\chi(t)}_{\sigma_{T'}(t)}
& \leq \int_0^t \nnorm{\partial_s\chi(s)}_{\sigma_{T'}(t)} \rd s\\
& \leq \nnorm{\chi}_{T',\sigma} \frac{\sigma}{\sigma-\sigma_{T'}(t)}
\int_0^t 
\frac1{\sqrt{1-\frac{s\sigma}{T'(\sigma-\sigma_{T'}(t))}}}\rd s
\\ & 
= \nnorm{\chi}_{T',\sigma} \frac{T'}{t}\int_0^t\frac1{\sqrt{1-\frac{s}{t}}}\rd s
=T'\nnorm{\chi}_{T',\sigma} \int_0^1\frac1{\sqrt{1-s}}\,\rd s
=2T'\nnorm{\chi}_{T',\sigma},
\end{align*}
since $\frac{\sigma-\sigma_{T'}(t)}{\sigma}=\frac{t}{T'}$.
Recalling \eqref{eqn.sigma.t.sup.norm.order}, we conclude the desired. 

Next, to prove \eqref{eq:BGnormPchi}, we compute, for $0 \leq \sigma' < \sigma(1-t/T')$ with a $\tau(s) > \sigma'$ to be chosen, using \cref{lem.bdd.gradient} 
\begin{equation}\label{eq:BGbound.P.chi}
\begin{split}
\nnorm{\pscal{\cP}\chi(t)}_{\sigma'}
& \leq \int_0^t \nnorm{\pscal{\cP}\partial_s\chi(s)}_{\sigma'} \rd s
    \leq \int_0^t \frac{\nnorm{\partial_s\chi(s)}_{\tau(s)}}{\tau(s)-{\sigma'}} \rd s\\
& \leq \int_0^t \frac{\sigma \, \rd s}{(\tau(s) - \sigma')(\sigma - \tau(s))\sqrt{1 - \frac{s\sigma}{ T'(\sigma - \tau(s))}}}
    \nnorm{\chi}_{T',\sigma}.
    \end{split}
\end{equation}
We choose 
$$\tau(s)=\frac{\sigma(1-s/T')+\sigma'}{2}=\frac{\sigma_{T'}(s)+\sigma'}{2}$$
Note that $\tau(t) = \frac{\sigma(1-t/T') + \sigma'}{2} > \sigma'$, since we consider only $\sigma' < \sigma(1-t/T')$. 
Hence, 
\begin{align*}
T'(\sigma-\tau(s)) - s\sigma & = \frac{T'}{2}(\sigma_{T'}(s)-\sigma'),
&
\tau(s)-\sigma' & =\frac{\sigma_{T'}(s)-\sigma'}{2},
\\
\sigma-\tau(s) & =\frac{\sigma(1+s/T')-\sigma'}{2}.    
\end{align*}

Then, we may bound the integral as 
\begin{align*}
    & \int_0^t \frac{\sigma \, \rd s}{(\tau(s) - \sigma')(\sigma - \tau(s))\sqrt{1 - \frac{s\sigma}{T'(\sigma - \tau(s))}}}
    \\ & \quad 
    = 
    \sqrt{T'}\int_0^t \frac{\sigma \, \rd s}{(\tau(s) - \sigma')\sqrt{\sigma - \tau(s)}\sqrt{T'(\sigma - \tau(s)) - s\sigma}}
    \\ & \quad 
    = 4 \int_0^t\frac{\sigma\,\rd s}{(\sigma(1-s/T')-\sigma')^{3/2}\sqrt{\sigma(1+s/T')-\sigma'}}
    \\ & \quad 
    = 4\int_0^t\frac{\sigma\,\rd s}{(\sigma-\sigma')^2\left(1-\frac{s\sigma}{T'(\sigma-\sigma')}\right)^{3/2}\left(1+\frac{s\sigma}{T'(\sigma-\sigma')}\right)^{1/2}}
    \\ & \quad 
    \leq 
    \frac{4T'}{\sigma-\sigma'} 
    \int_0^{\frac{t\sigma}{T'(\sigma-\sigma')}} \frac{\rd s}{(1-s)^{3/2}}
    \\ & \quad 
    \leq \frac{8T'}{(\sigma-\sigma')\left(1-\frac{t\sigma}{T'(\sigma-\sigma')}\right)^{1/2}},
\end{align*}
where we changed variables $s \to \frac{T'(\sigma-\sigma')}{\sigma}s$. We also used that $\sqrt{1+s}\geq 1$ and
$$\int_0^t \frac{\rd s}{(1-s)^{3/2}} = \frac{2 - 2(1-t)^{1/2}}{(1-t)^{1/2}} \leq \frac{2}{(1-t)^{1/2}}.$$ 
Inserting in \eqref{eq:BGbound.P.chi}, multiplying with  $\frac{\sigma-\sigma'}{\sigma}\sqrt{1-\frac{t\sigma}{T'(\sigma-\sigma')}}$, which precisely cancels the singularity, and taking the supremum concludes the proof of \eqref{eq:BGnormPchi}.
\end{proof}

\subsection{Bounds on the nonlinearity}
Next, we bound the nonlinear term in \eqref{eqn.NLS.V.Duhamel.tildePsi}. Concretely, we bound for any $T'$ the map $\cG :\cB_{T',\sigma}\to \cB_{T',\sigma}$ defined by 
\begin{equation}\label{eq:defG}
    \cG[\chi](t) 
    = 
     -i \int_0^t e^{is\cH_0} \sum_{j=1}^N V_{e^{-is\cH_0}(\Psi_0+\chi(s))}(s,x_j) e^{-is\cH_0}(\Psi_0+\chi(s)) \, \rd s, 
\end{equation}
where we recall the definition of $V_\Psi$ in \eqref{eq:VPsi}. 
We start with the bound on the unitary group $e^{it\cH_0}$.

\begin{lemma}[Unitary group]
\label{lemma:Vunitary.norm.sigma} 
We have 
 \begin{equation*}
  \norm{e^{-it\cH_0}\Psi}_{\sigma}
  \leq
  C \norm{\Psi}_{\sigma}
  ,
\end{equation*}
where the constant $C$ is independent of $\sigma$. 
\end{lemma}

\begin{proof}
This follows from the fact that $\cH_0$ is a self-adjoint operator with domain $H^2$, and that $\cP$ commutes with $\cH_0$. We give some details, for the convenience of the reader. Recall that two unbounded operators $A,B$ commute when $f(A)g(B)=g(B)f(A)$ for any bounded functions $f,g$. Then $g(B)$ preserves the domain of $f(A)$ for any locally bounded function $f$ and globally bounded function $g$, with the same relation $f(A)g(B)=g(B)f(A)$ on the domain of $f(A)$, by~\cite[Thm.~4.41]{Lewin-Spectral}.

In our case, we have $e^{i a \cdot \cP} \cH_0 e^{-i a \cdot \cP} = \cH_0$ for any $a\in \Lambda$ by translation-invariance, on the domain $D(\cH_0)=H^2$. By \cite[Thm.~4.41]{Lewin-Spectral}, this means that $\cP$ and $\cH_0$ commute. We then obtain $e^{-it\cH_0}e^{\sigma|\cP|} = e^{\sigma|\cP|}e^{-it\cH_0}$ on $D(e^{\sigma|\cP|})$ and $(C+\cH_0)e^{\sigma|\cP|}e^{-it\cH_0} = e^{-it\cH_0}(C+\cH_0)e^{\sigma|\cP|}$ on $D((C+\cH_0)e^{\sigma|\cP|})=\cB_\sigma$. Thus, from $\cH_0$ having domain $H^2$ we have the bounds
\begin{multline*}
\norm{e^{-it\cH_0}\Psi}_\sigma 
=
\norm{(1+\cT) e^{\sigma|\cP|} e^{-it\cH_0} \Psi }_{L^2} 
\leq C \norm{(C+\cH_0)  e^{\sigma|\cP|} e^{-it\cH_0} \Psi}_{L^2}
\\
= C \norm{ e^{-it\cH_0} (C+\cH_0)  e^{\sigma|\cP|} \Psi}_{L^2}
\leq C \norm{(1+\cT)  e^{\sigma|\cP|} \Psi}_{L^2}
\leq C \norm{\Psi}_\sigma
\end{multline*}
as desired. 
\end{proof}

\begin{lemma}[The nonlinear potential]
\label{lem.VPsi.bdds}
Let $\Psi_1, \Psi_2 \in \cB_\sigma$, let $\rho\in C^2([0,T];\cA_{\sigma})$ with $\int\rho=N$ satisfy \eqref{eq.ass.rho}, and let $V_{\Psi_1}(t)$ be as defined in \eqref{eq:VPsi}. 
Then, for any $0 \leq \sigma' < \sigma$, 
\begin{multline}    
\label{eqn.VPsi.nonlinear.Lipschitz}
    \nnorm{V_{\Psi_1} (t,x_1) \Psi_1 - V_{\Psi_2} (t,x_1)\Psi_2}_{\sigma'} 
    \\
    \leq 
    C \Bigl(1 + 
        \left(\nnorm{\Psi_1}_{\sigma'} + \nnorm{\Psi_2}_{\sigma'}\right)
        \left(\nnorm{\pscal{\cP}\Psi_1}_{\sigma'} + \nnorm{\pscal{\cP}\Psi_2}_{\sigma'}\right) 
    \Bigr)
    \nnorm{\Psi_1 - \Psi_2}_{\sigma'}
    \\
    + C \left(\nnorm{\Psi_1}_{\sigma'} + \nnorm{\Psi_2}_{\sigma'}\right)^2 \nnorm{\pscal{\cP}(\Psi_1 - \Psi_2)}_{\sigma'}.
\end{multline}
Moreover, we have the continuity bound on $V_\Psi$ 
(which we state in the $\norm{\cdot}_{\sigma'}$-norms, as this will only be used in the proof of \cref{thm.main.V} in \cref{sec.final.steps.proof.thm.main} below)
\begin{align}
    \norm{V_{\Psi_1}(t)}_{\sigma'}
    & \leq C + C \nnorm{\Psi_1}_{\sigma'}^2
    ,
    \label{eqn.VPsi.bdd}
    \\
    \norm{V_{\Psi_1}(t) - V_{\Psi_2}(t')}_{\sigma'}
    & \leq C \left(\nnorm{\Psi_1}_{\sigma'} + \nnorm{\Psi_2}_{\sigma'}\right) \nnorm{\Psi_1 - \Psi_2}_{\sigma'}
    \nn \\* &
    \quad + C \bigl(1 + \nnorm{\Psi_1}_{\sigma'}^2 + \nnorm{\Psi_2}_{\sigma'}^2\bigr)
    \nn \\* & \qquad 
    \times\left(\norm{(1+\Delta^2)(\rho(t) - \rho(t'))}_{\sigma'} +\norm{ \partial_t^2(\rho(t) - \rho(t'))}_{\sigma'} \right)
    .
    \label{eqn.VPsi.cts.in.time.bound}
\end{align}
\end{lemma}
\begin{proof}
As a helpful intermediate bound, we first consider $K_\rho^{-1} Q_{\Psi_1,\Psi_2}$, where $Q_{\Psi_1,\Psi_2}$ is as given in \eqref{eq:T_Psi_chi}; note that $Q_{\Psi_1,\Psi_2}\in \cA_{\sigma'}^\perp$ by \cref{lemma:boundspreliminaryTi}. 
First, we establish a bound on $K_\rho^{-1}(1-\Delta)$ as an operator in the $\nnorm{\cdot}_{\sigma'}$-norm. 
For brevity of notation we write $\norm{\cdot}_{\sigma'\to\sigma'}$ also for the operator norm on $\cA_{\sigma'}^\perp$. 
Note that for any $f\in \cA_{\sigma'}^\perp$
\begin{align*}
    \nnorm{K_\rho^{-1} (1-\Delta)f}_{\sigma'} 
    & = \norm{(1-\Delta)^{1/2}K_\rho^{-1}(1-\Delta)^{1/2}(1-\Delta)^{1/2}f}_{\sigma'} 
    \\ & \leq \norm{(1-\Delta)^{1/2}K_\rho^{-1}(1-\Delta)^{1/2}}_{\sigma' \to \sigma'} \nnorm{f}_{\sigma'} 
\end{align*}
That is, it suffices to bound the operator $(1-\Delta)^{1/2}K_\rho^{-1}(1-\Delta)^{1/2}$ as an operator on $\cA_{\sigma'}^\perp$. 
For this we note that $\norm{(1-\Delta)^{1/2}\cdot}_{\sigma'} $ and $\norm{\cdot}_{\sigma'} +\sum_{j=1}^D \norm{\partial_j\cdot}_{\sigma'} $ 
are equivalent norms. Thus,
\begin{align*}
&\norm{(1-\Delta)^{1/2}K_\rho^{-1}(1-\Delta)^{1/2}}_{\sigma' \to \sigma'}\\
&\quad 
\leq C \norm{K_\rho^{-1}(1-\Delta)^{1/2}}_{\sigma' \to \sigma'}
+C\sum_j \norm{\partial_j K_\rho^{-1}(1-\Delta)^{1/2}}_{\sigma' \to \sigma'}\\
&\quad \leq C \norm{K_\rho^{-1}(1-\Delta)^{1/2}}_{\sigma' \to \sigma'}
+C \sum_j\norm{K_\rho^{-1}\partial_j (1-\Delta)^{1/2}}_{\sigma' \to \sigma'}
\\* & \qquad 
+C \sum_j\norm{[K_\rho^{-1},\partial_j] (1-\Delta)^{1/2}}_{\sigma' \to \sigma'}\\
&\quad \leq C \norm{K_\rho^{-1} (1-\Delta)}_{\sigma' \to \sigma'}
+C \sum_j\norm{K_\rho^{-1}\nabla\cdot [\rho,\partial_j] \nabla K_\rho^{-1}(1-\Delta)^{1/2}}_{\sigma' \to \sigma'}
.
\end{align*}
The first term is bounded due to \cref{lem.Krho}. After moving the gradient to the left in the second term we obtain that it is bounded by
\begin{align*}
&\left(
\norm{K_\rho^{-1}\Delta}_{\sigma' \to \sigma'}
\sum_j\norm{ \partial_j \rho}_{\sigma'}
+\sum_{j,k}\norm{K_\rho^{-1}\partial_k}_{\sigma' \to \sigma'}
\norm{\partial_j \partial_k \rho}_{\sigma'}
\right)
\norm{K_\rho^{-1}(1-\Delta)^{1/2}}_{\sigma' \to \sigma'}. 
\end{align*}
Here we expanded the commutators to see that they are in fact multiplication operators. Thus, by \cref{lemma:multiplication}, they are bounded by $\norm{(1-\Delta)\rho}_{\sigma'}$.
We conclude that   
\begin{equation*}
     \nnorm{K_\rho^{-1} (1-\Delta)f}_{\sigma'} \leq C \nnorm{f}_{\sigma'}. 
\end{equation*}
With this bound and \cref{lemma:boundspreliminaryTi} we conclude
\begin{align}
    \nnorm{K_\rho^{-1} Q_{\Psi_1,\Psi_2}}_{\sigma'} 
    & \leq C \nnorm{(1-\Delta)^{-1}Q_{\Psi_1,\Psi_2}}_{\sigma'}
    \nn \\ & 
    \leq C \left(\nnorm{\pscal{\cP}\Psi_1}_{\sigma'} \nnorm{\Psi_2}_{\sigma'} + \nnorm{\Psi_1}_{\sigma'} \nnorm{\pscal{\cP}\Psi_2}_{\sigma'}\right). 
\label{eqn.bdd.Krho.inv.Q}
\end{align}
Then, we write 
\begin{equation*}
    V_{\Psi_1} (t)\Psi_1 - V_{\Psi_2}(t)\Psi_2
    =
   ( V_{\Psi_1}(t)- V_{\Psi_2}(t)) \Psi_1 + V_{\Psi_2}(t) (\Psi_1 - \Psi_2).
\end{equation*}
Recalling the definition of $V_\Psi(t)$ in \eqref{eq:VPsi} we have 
\begin{equation*}
    V_{\Psi_1}(t) - V_{\Psi_2}(t) = \frac{1}{2} K_{\rho(t)}^{-1}(Q_{\Psi_1} - Q_{\Psi_2})
    = \frac{1}{2} \Re K_{\rho(t)}^{-1} Q_{\Psi_1 + \Psi_2, \Psi_1 - \Psi_2}. 
\end{equation*}
Using \eqref{eqn.bdd.Krho.inv.Q} we thus conclude the bound 
\begin{multline*}
    \nnorm{V_{\Psi_1} (t)- V_{\Psi_2}(t)}_{\sigma'} 
    \leq C \nnorm{\Psi_1 + \Psi_2}_{\sigma'} \nnorm{\pscal{\cP}(\Psi_1 - \Psi_2)}_{\sigma'}
    \\
    + C \nnorm{\pscal{\cP}(\Psi_1 + \Psi_2)}_{\sigma'} \nnorm{\Psi_1 - \Psi_2}_{\sigma'}. 
\end{multline*}
Similarly, recalling \eqref{eq.ass.rho} and again using \eqref{eqn.bdd.Krho.inv.Q}, we have 
\begin{align}
    \nnorm{V_{\Psi_2}}_{\sigma'} 
    & \leq
    \frac{1}{2}\nnorm{K_{\rho}^{-1}(\partial_t^2\rho + \Delta^2\rho)}_{\sigma'} + \frac{1}{2}\nnorm{K_\rho^{-1}Q_{\Psi_2,\Psi_2}}_{\sigma'}
    \leq 
    C + C \nnorm{\Psi_2}_{\sigma'} \nnorm{\pscal{\cP}\Psi_2}_{\sigma'}. 
    \label{eqn.bdd.VPsi.|||.norm}
\end{align}
Finally, by \cref{lemma:multiplication} (which immediately implies a product rule for $\nnorm{\cdot}_{\sigma'}$) we thus conclude the proof of \eqref{eqn.VPsi.nonlinear.Lipschitz}.

Next, to prove \eqref{eqn.VPsi.bdd} and \eqref{eqn.VPsi.cts.in.time.bound} we note, using \cref{lemma:boundspreliminaryTi}, similarly as \eqref{eqn.bdd.Krho.inv.Q}, that
\begin{align}
    \norm{K_\rho^{-1}Q_{\Psi_1,\Psi_2}}_{\sigma'}
    & \leq \norm{K_\rho^{-1}(1-\Delta)}_{\sigma'\to\sigma'} \norm{(1-\Delta)^{-1}Q_{\Psi_1,\Psi_2}}_{\sigma'}
    \leq C \nnorm{\Psi_1}_{\sigma'} \nnorm{\Psi_2}_{\sigma'}. 
    \label{eqn.bdd.Krho.||-norm}
\end{align}
With this, \eqref{eqn.VPsi.bdd} follows similarly as \eqref{eqn.bdd.VPsi.|||.norm} as 
\begin{equation*}
    \norm{V_{\Psi_1}(t)}_{\sigma'} 
    \leq
    \frac{1}{2}\norm{K_{\rho}^{-1}(\partial_t^2\rho + \Delta^2\rho)}_{\sigma'} + \frac{1}{2}\norm{K_\rho^{-1}Q_{\Psi_1,\Psi_1}}_{\sigma'}
    \leq 
    C + C \nnorm{\Psi_1}_{\sigma'}^2. 
\end{equation*}

Finally, to prove \eqref{eqn.VPsi.cts.in.time.bound} we write 
\begin{equation}\label{eqn.split.VPsi.cts}
    \norm{V_{\Psi_1}(t)-V_{\Psi_2}(t')}_{\sigma'}
    \leq 
\norm{V_{\Psi_1}(t)-V_{\Psi_2}(t)}_{\sigma'}
+\norm{V_{\Psi_2}(t)-V_{\Psi_2}(t')}_{\sigma'}.
\end{equation}
Using \eqref{eqn.bdd.Krho.||-norm}, the first term in \eqref{eqn.split.VPsi.cts} is bounded as 
\begin{equation*}
    \norm{V_{\Psi_1}(t)-V_{\Psi_2}(t)}_{\sigma'}
    \leq C \nnorm{(\Psi_1 + \Psi_2)}_{\sigma'} \nnorm{\Psi_1- \Psi_2}_{\sigma'}.
\end{equation*}
For the second term in \eqref{eqn.split.VPsi.cts} we write 
\begin{align*}
    V_{\Psi_2}(t)-V_{\Psi_2}(t')
    & = \frac{1}{2} K_{\rho(t)}^{-1}\left(-\partial_t^2\rho(t) - \Delta^2\rho(t) + \partial_t^2\rho(t') + \Delta^2\rho(t')\right)
    \\ &  
    \quad 
     + \frac{1}{2} \left(K_{\rho(t)}^{-1} - K_{\rho(t')}^{-1}\right)\left(-\partial_t^2\rho(t') - \Delta^2\rho(t') + Q_{\Psi_2}\right)
\end{align*}
By \cref{lem.Krho} we bound the first term here by 
\begin{multline*}
    \norm{K_{\rho(t)}^{-1}\left(\partial_t^2\rho(t) + \Delta^2\rho(t) - \partial_t^2\rho(t') - \Delta^2\rho(t')\right)}_{\sigma'}
    \\ 
    \leq C \norm{\partial_t^2\rho(t) + \Delta^2\rho(t) - \partial_t^2\rho(t') - \Delta^2\rho(t')}_{\sigma'}. 
\end{multline*}
To bound the second term, we use the continuity bound in \eqref{eq:K_resolvent_continuous} of \cref{lem.Krho}. 
By bounding the $Q_{\Psi_2}$-term as above using \cref{lemma:boundspreliminaryTi} we get 
\begin{align*}
    & \norm{\left(K_{\rho(t)}^{-1} - K_{\rho(t')}^{-1}\right)\left(-\partial_t^2\rho(t') - \Delta^2\rho(t') + Q_{\Psi_2}\right)}_{\sigma'}
    \\ & \quad 
    \leq C \norm{(1-\Delta)^{1/2}(\rho(t) - \rho(t'))}_{\sigma'} 
     \norm{ \partial_t^2 \rho(t') + \Delta^2 \rho(t')}_{\sigma'}
     \\ & \qquad 
     + C \norm{(1-\Delta)^{1/2}(\rho(t) - \rho(t'))}_{\sigma'} \norm{(1-\Delta)^{-1}Q_{\Psi_2}}_{\sigma'}
     \\ & \quad 
     \leq C \left(1 + \nnorm{\Psi_2}_{\sigma'}^2\right) \norm{(1-\Delta)^{1/2}(\rho(t) - \rho(t'))}_{\sigma'}. 
\end{align*}
Combining this with the above bounds we conclude the proof of \eqref{eqn.VPsi.cts.in.time.bound}. 
\end{proof}

\begin{lemma}[Contraction estimate]
\label{lem.VPsi.nonlinear.Lipschitz}
Let $\pscal{\cP}\Psi_0\in \cB_\sigma$, let $\rho\in C^2([0,T];\cA_\sigma)$ with $\int \rho = N$ satisfy \eqref{eq.ass.rho}, and let $\cG$ be as defined in \eqref{eq:defG}.
Then, there exists a constant $C>0$ such that for all $0 < T' < T$
we have
\begin{equation*}
    \nnorm{\cG[\chi_1] - \cG[\chi_2]}_{T',\sigma}
    \leq 
    C T' \left(1+\nnorm{\chi_1}_{T',\sigma}  T' + \nnorm{\chi_2}_{T',\sigma} T'\right)^2
    \nnorm{\chi_1 - \chi_2}_{T',\sigma}        
    .
\end{equation*}
\end{lemma}

\begin{proof}
Recall \cref{def.BG.space}; thus, we need to bound
$$
\nnorm{\cG[\chi_1] - \cG[\chi_2]}_{T',\sigma}=
\sup_{\substack{0\leq\sigma'<\sigma\\ 0\leq t\leq T'\frac{\sigma-\sigma'}{\sigma}}}
\nnorm{\partial_t\left(\cG[\chi_1] - \cG[\chi_2]\right)}_{\sigma'}\frac{\sigma-\sigma'}{\sigma}\sqrt{1-\frac{t\sigma}{T'(\sigma-\sigma')}}.$$
By the definition of $\cG$ in \eqref{eq:defG}, this follows from controlling
\begin{equation*}
    \begin{split}
\sup_{\substack{0\leq\sigma'<\sigma\\ 0\leq t\leq T'\frac{\sigma-\sigma'}{\sigma}}}
& \nnorm{V_{\Psi_1(t)} (t) \Psi_1(t) - V_{\Psi_2(t)} (t)\Psi_2(t)}_{\sigma'} 
\frac{\sigma-\sigma'}{\sigma}\sqrt{1-\frac{t\sigma}{T'(\sigma-\sigma')}}
,
    \end{split}
\end{equation*}
where we denote $\Psi_k(t)=e^{-it\cH_0}(\Psi_0+\chi_k(t))$ for brevity and
where we drop the $x_j$ in $V$ for simplicity, this just amounts to a factor $N$. Note that we also used \cref{lemma:Vunitary.norm.sigma} to bound the first appearance of the unitary group.
From \eqref{eqn.VPsi.nonlinear.Lipschitz} we know
\begin{align*}    
    & \nnorm{V_{\Psi_1(t)} (t) \Psi_1(t) - V_{\Psi_2(t)} (t)\Psi_2(t)}_{\sigma'} 
    \\ & \quad 
    \leq 
    C \Bigl(1 + 
        \left(\nnorm{\Psi_1(t)}_{\sigma'} + \nnorm{\Psi_2(t)}_{\sigma'}\right)
        \left(\nnorm{\pscal{\cP}\Psi_1(t)}_{\sigma'} + \nnorm{\pscal{\cP}\Psi_2(t)}_{\sigma'}\right) 
    \Bigr)
    \\ & \qquad \quad \times 
    \nnorm{\Psi_1 (t)- \Psi_2(t)}_{\sigma'}
    \\ & \qquad 
    + C \left(\nnorm{\Psi_1(t)}_{\sigma'} + \nnorm{\Psi_2(t)}_{\sigma'}\right)^2 \nnorm{\pscal{\cP}(\Psi_1(t) - \Psi_2(t))}_{\sigma'}.
\end{align*}
Note that by \cref{lemma:Vunitary.norm.sigma} and 
\eqref{eq:BG_norm1},
$$
\nnorm{\Psi_k(t)}_{\sigma'} 
\leq C \nnorm{\Psi_0+\chi_k(t) }_{\sigma'} 
\leq C (\nnorm{\Psi_0}_{\sigma}  +T'\nnorm{\chi_k}_{T',\sigma} ).
$$
Similarly,
$$
\nnorm{\Psi_1(t)-\Psi_2(t)}_{\sigma'} 
\leq C T'\nnorm{\chi_1-\chi_2}_{T',\sigma} .
$$
Combining the above, we obtain
\begin{align*}
& \nnorm{\cG[\chi_1] - \cG[\chi_2]}_{T',\sigma} 
\\ & \quad 
\leq CN T' \nnorm{\chi_1-\chi_2}_{T',\sigma}
\sup_{\substack{0\leq\sigma'<\sigma\\ 0\leq t\leq T'\frac{\sigma-\sigma'}{\sigma}}}
\frac{\sigma-\sigma'}{\sigma}\sqrt{1-\frac{t\sigma}{T'(\sigma-\sigma')}}
\Biggl[ 
    1 + 
        \Bigl(\nnorm{\Psi_0}_{\sigma}  
        \\ & \qquad \quad 
        + T'\nnorm{\chi_1}_{T',\sigma}  
        + T'\nnorm{\chi_2}_{T',\sigma} 
\Bigr)
         \Bigl(\nnorm{\pscal{\cP}\Psi_0}_{\sigma'}+\nnorm{\pscal{\cP} \chi_1(t)}_{\sigma'} + \nnorm{\pscal{\cP}\chi_2(t)}_{\sigma'}
         \Bigr) 
\Biggr]
   \\
    & \qquad +CN 
 \left(\nnorm{\Psi_0}_\sigma + 
 \nnorm{\chi_1}_{T',\sigma} T'  + \nnorm{\chi_2}_{T',\sigma} T'
\right)^2
\\ & \qquad \quad \times 
\sup_{\substack{0\leq\sigma'<\sigma\\ 0\leq t\leq T'\frac{\sigma-\sigma'}{\sigma}}} 
\frac{\sigma-\sigma'}{\sigma}\sqrt{1-\frac{t\sigma}{T'(\sigma-\sigma')}}
\nnorm{\pscal{\cP}(\chi_1(t)-\chi_2(t))}_{\sigma'}. 
\end{align*}
Now, using \eqref{eq:BGnormPchi} yields the desired bound.
\end{proof}

\subsection{The Banach fixed-point argument (Proof of \texorpdfstring{\cref{prop.solve.NLS}}{Proposition~\ref*{prop.solve.NLS}})}
\label{sec.Banach.fixed.point}

We first establish existence and uniqueness of solutions in the (affine) Baouendi--Goulaouic space $\Psi_0 + \cB_{T',\sigma}$ (recall \cref{def.BG.space})
for $T'$ small enough. Subsequently, we lift the existence and uniqueness to the space $C^0([0,T''];\cB_{\sigma'})$ and hence prove \cref{prop.solve.NLS}. 
The existence and uniqueness in the (affine) Baouendi--Goulaouic space is stated as follows.

\begin{lemma}[Solutions in the Baouendi--Goulaouic space]
\label{lem.existence.uniqueness.BG.space}
Let $\rho\in C^2([0,T];\cA_{\sigma})$ with $\int \rho = N$ satisfy \eqref{eq.ass.rho} and let $\pscal{\cP}\Psi_0\in\cB_\sigma$. 
Then, for any $R>0$ sufficiently large and any $T'>0$ sufficiently small (depending on $R$), there exists a unique solution $\tilde \Psi \in \Psi_0 + \cB_{T',\sigma}$ with 
$\snnorm{\tilde \Psi - \Psi_0}_{T',\sigma} \leq R$
to the equation \eqref{eqn.NLS.V.Duhamel.tildePsi}. 
\end{lemma}

\begin{proof}
We use a fixed-point method using the map $\cG$ defined in \eqref{eq:defG}. We show that, for any $R>0$ large enough and $T'>0$ small enough (dependent on $R$), the map $\cG$ is a contraction on the ball of radius $R$ in $\cB_{T',\sigma}$. 
By the Banach fixed-point theorem this proves the desired.

Thus, we need to prove that, for appropriate $R,T'$, (1) that the map $\cG$ preserves the ball of radius $R$ and (2) that $\cG$ is a contraction. To prove that $\cG$ preserves the ball of radius $R$ we first bound $\nnorm{\cG[0]}_{T',\sigma}$. 
Recalling \cref{def.BG.space} of the $\nnorm{\cdot}_{T',\sigma}$-norm, we thus consider (using \cref{lemma:Vunitary.norm.sigma} to bound the unitary group) $\nnorm{V_{e^{-it\cH_0}\Psi_0}e^{-it\cH_0}\Psi_0}_{\sigma'}$. From \cref{lem.VPsi.bdds,lemma:Vunitary.norm.sigma,lem.bdd.gradient} we have
\begin{equation*}
    \nnorm{V_{e^{-it\cH_0}\Psi_0}e^{-it\cH_0}\Psi_0}_{\sigma'}
    \leq C \nnorm{\Psi_0}_{\sigma'} + C \nnorm{\Psi_0}_{\sigma'}^2 \nnorm{\pscal{\cP}\Psi_0}_{\sigma'}
    \leq \frac{C}{\sigma - \sigma'} \left(1 + \nnorm{\Psi_0}_{\sigma}^3\right). 
\end{equation*}
Thus, 
\begin{equation*}
    M := \nnorm{\cG[0]}_{T',\sigma} \leq \frac{C}{\sigma} \left(1 + \nnorm{\Psi_0}_{\sigma}^3\right) < \infty. 
\end{equation*}

Suppose now that $\nnorm{\chi_1}_{T',\sigma} \leq R$. Then, by \cref{lem.VPsi.nonlinear.Lipschitz} (labeling the constant $C_1$), we have
\begin{align*}
    \nnorm{\cG[\chi_1]}_{T',\sigma} 
    & \leq \nnorm{\cG[0]}_{T',\sigma} + \nnorm{\cG[\chi_1] - \cG[0]}_{T',\sigma} 
    \leq M + C_1 T' \left(1 + T'R\right)^2 R . 
\end{align*}
Next, suppose that $\nnorm{\chi_k}_{T',\sigma} \leq R$ for $k=1,2$. Then, by \cref{lem.VPsi.nonlinear.Lipschitz} we have 
\begin{equation*}
    \nnorm{\cG[\chi_1] - \cG[\chi_2]}_{T',\sigma} 
    \leq C_1 T' (1+2RT')^2 \nnorm{\chi_1 - \chi_2}_{T',\sigma}. 
\end{equation*}
We conclude that $\cG$ preserves the ball of radius $R$ and is a contraction on this set if
\begin{align*}
    M + C_1 T' \left(1 + T'R\right)^2 R \leq R
    \qquad &\textnormal{and} \qquad 
    C_1 T' (1+2RT')^2 < 1. 
\end{align*}
Clearly, for any
\begin{align*}
    R > M \qquad &\textnormal{and} \qquad 
    T' < \min \left\{\frac{1}{R}, \frac{R-M}{4C_1R}, \frac{1}{9C_1}\right\}
\end{align*}
these desiderata are met and we conclude the existence and uniqueness of solutions by the Banach fixed-point theorem. 
\end{proof}

With this lemma we now give the 

\begin{proof}[Proof of \cref{prop.solve.NLS}]
Note first that the conditions $\Psi_0\in \cB_\sigma$, $\rho \in C^2([0,T];\cA_\sigma)$, and $\rho \geq c>0$ imply, using \cref{lemma.1/rho.analytic,lem.bdd.gradient}, that $\pscal{\cP}\Psi_0\in \cB_\sigma$ and that $\rho$ satisfies \eqref{eq.ass.rho}, for some smaller value of $\sigma$. 
We nonetheless continue to denote this reduced analyticity parameter by $\sigma$. 

\smallskip
\paragraph*{\underline{Existence}:}
Applying \cref{lem.existence.uniqueness.BG.space} we conclude the existence of a (unique) solution in the (affine) Baouendi--Goulaouic space $\tilde\Psi\in \Psi_0 + \cB_{T',\sigma}$ to \eqref{eqn.NLS.V.Duhamel.tildePsi}. We claim that $\Psi(t) = e^{-it\cH_0}\tilde\Psi(t)$ is a solution to \eqref{eqn.NLS.V.Duhamel} of the desired regularity. It is clearly a solution. It remains to check the regularity.

First, we note that for $T'' = T'/2$ we have for any $t\leq T''$, by \cref{lemma:Vunitary.norm.sigma} and \eqref{eq:BG_norm1} 
\begin{equation*}
    \norm{\Psi(t)}_{\sigma/2} 
    \leq \nnorm{\Psi(t)}_{\sigma/2}
    \leq C \snnorm{\tilde\Psi(t)}_{\sigma/2} 
    \leq C \nnorm{\Psi_0}_{\sigma/2} + \snnorm{\tilde\Psi - \Psi_0}_{T',\sigma} < \infty. 
\end{equation*}
Then, we claim that $\Psi$ is continuous in time in $\cB_{\sigma/4}$. 
This follows by interpolating the norm as 
$$
\norm{\Psi}_{\sigma/4}^2
= \norm{e^{ \abs{\cP}\sigma/4}\Psi}_{H^2}^2
\leq \norm{e^{ \abs{\cP}\sigma/2}\Psi}_{H^2}\norm{\Psi}_{H^2}
=\norm{\Psi}_{\sigma/2}\norm{\Psi}_{H^2}.
$$
The continuity in $\cB_{\sigma/4}$ then follows from that in $\cB_0 = H^2$ and the boundedness in $\cB_{\sigma/2}$. 
We conclude the existence of solutions in the space $C^0([0,T'/2];\cB_{\sigma/4})$. 
Up to renaming $T'$ and defining $\sigma' = \sigma /4$ this is the existence claimed in the proposition.

\smallskip
\paragraph*{\underline{Uniqueness}:}
Next, we consider uniqueness. 
Suppose $\Psi_1,\Psi_2\in C^0([0,T']; \cB_{\sigma'})$ are two different solutions to \eqref{eqn.NLS.V.Duhamel} with $\Psi_1(0)=\Psi_2(0)=\Psi_0$. Let $t_0 = \inf\{t\geq 0 : \Psi_1(t)\ne \Psi_2(t)\}$ be the first time the two functions differ and assume without loss of generality that $t_0=0$ (otherwise we just start the evolution at $t_0$ and note that $\Psi_1(t_0)=\Psi_2(t_0)$ by continuity from below).
By reducing the value of $\sigma'$ we may assume that $\nnorm{\Psi_k(t)}_{\sigma'}\leq C$ uniformly in $0\leq t \leq T'$ (by \cref{lem.bdd.gradient}).

We bound the difference 
\begin{equation*}
\tilde\Psi_1(t) - \tilde\Psi_2(t) 
    = \left(e^{it\cH_0}\Psi_1(t) - \Psi_0\right) 
    - \left(e^{it\cH_0}\Psi_2(t) - \Psi_0\right)
\end{equation*}
in the $\nnorm{\cdot}_{ T'', \sigma'}$-norm using \cref{lem.VPsi.nonlinear.Lipschitz} for a sufficiently small $T''$. To apply this, we thus need to bound 
$\snnorm{\tilde\Psi_k - \Psi_0}_{T'',\sigma'}$. 
Here we use that $\tilde\Psi_k$ solves \eqref{eqn.NLS.V.Duhamel.tildePsi}. Recalling \cref{lemma:Vunitary.norm.sigma,lem.VPsi.bdds,lem.bdd.gradient} we have 
\begin{align*}
    \snnorm{\partial_t(\tilde\Psi_k - \Psi_0)}_{\sigma''} 
    & = \nnorm{e^{it\cH_0} \sum_{j=1}^N V_{e^{-it\cH_0}\tilde\Psi_k(t)}(t,x_j) e^{-it\cH_0} \tilde\Psi_k(t)}_{\sigma''}
    \\ & 
    \leq \frac{C}{\sigma' - \sigma''} \left(1 + \nnorm{\Psi_k(t)}_{\sigma'}^3\right)
    \leq \frac{C}{\sigma' - \sigma''}. 
\end{align*}
Thus, recalling \cref{def.BG.space}, we have $\snnorm{\tilde\Psi_k - \Psi_0}_{T'',\sigma'} \leq C/\sigma' =: C_1$. 
By \cref{lem.VPsi.nonlinear.Lipschitz} (labeling the constant as $C_2$), we conclude the bound
\begin{equation*}
    \snnorm{\tilde \Psi_1 - \tilde \Psi_2}_{T'',\sigma'} 
    \leq C_2 T'' (1 + 2C_1T'')^2 \snnorm{\tilde \Psi_1 - \tilde \Psi_2}_{T'',\sigma'}. 
\end{equation*}
For some $T''>0$ sufficiently small we have $C_2 T'' (1 + 2C_1T'')^2 < 1$ and thus, for all $t\leq T''$, that $\tilde \Psi_1(t) = \tilde \Psi_2(t)$ contradicting the assumption that $\Psi_1 \ne \Psi_2$. 
We conclude the uniqueness as claimed.
\end{proof}

\subsection{Final steps in the proof of \texorpdfstring{\cref{thm.main.V}}{Theorem~\ref*{thm.main.V}}}
\label{sec.final.steps.proof.thm.main}

Finally, we explain how to prove \cref{thm.main.V} using \cref{prop.solve.NLS}.

\begin{proof}[Proof of \cref{thm.main.V}]

We use both the existence and uniqueness parts of \cref{prop.solve.NLS} for solutions $\Psi$ to the nonlinear equation \eqref{eq:SEnonlinearV} to prove the existence and uniqueness statements of $V$ in \cref{thm.main.V}.

\smallskip
\paragraph*{\underline{Existence}:}
From \cref{prop.solve.NLS} we obtain a (unique) solution $\Psi\in C^0([0,T'];\cB_{\sigma'})$ to the nonlinear equation \eqref{eq:SEnonlinearV} for some time $T'\leq T$. We define then the potential $V(t) = V_{\Psi(t)}(t)$. What is left to show is that $V(t)$ has the desired regularity, i.e., $V\in C^0([0,T''];\cA_{\sigma''})$ for appropriate $T'', \sigma''$, and in fact correctly reproduces $\rho(t)$.

To prove the regularity in time we use 
\eqref{eqn.VPsi.bdd} and \eqref{eqn.VPsi.cts.in.time.bound} of \cref{lem.VPsi.bdds}. 
First, by \eqref{eqn.VPsi.bdd} and \cref{lem.bdd.gradient} we have that $\norm{V(t)}_{\sigma''} \leq C$ for any fixed $\sigma'' < \sigma'$ uniformly in $t\in [0,T']$. 
To obtain the continuity in time we first note that $\Psi$ is in fact continuous in time in the $\nnorm{\cdot}_{\sigma''}$-norm by an interpolation argument as in the proof of \cref{prop.solve.NLS} above. 
The continuity in time of $V$ (in the $\norm{\cdot}_{\sigma''}$-norm) then follows from \eqref{eqn.VPsi.cts.in.time.bound} using the continuity in time of $\rho, \partial_t^2\rho, \Delta^2\rho$, and $\Psi$. 
That is, $V\in C^0([0,T'];\cA_{\sigma''})$ for any $\sigma'' < \sigma'$.

Next, we check that the density is correctly reproduced.
To this end, first note that, since $\Psi$ is the solution to the Schrödinger equation with external potential $V\in C^0([0,T'];\cA_{\sigma''})$, we have the force-balance equation \eqref{eq:partialt^2rho} $\partial_t^2 \rho_\Psi = -\Delta^2\rho_\Psi - 2 K_{\rho_\Psi} V + Q_\Psi$, see \cref{prop.force.balance}.
By \cref{thm.density.is.analytic}
we in fact have $\rho_\Psi\in C^2([0,T'];\cA_{\sigma'''})$ 
for any $\sigma''' < \sigma''$. 
Next, by construction, we have $V=V_\Psi = \frac{1}{2}K_{\rho}^{-1}(-\partial_t^2\rho - \Delta^2\rho + Q_\Psi)$. Solving for $Q_\Psi$, recalling $K_\rho f = -\nabla\cdot(\rho\nabla f)$, and inserting in the force-balance equation above we find 
that the difference $h=\rho_\Psi - \rho \in C^2([0,T'], \cA_{\sigma'''})$ satisfies the equation
\begin{equation}\label{eqn.wave.density.difference}
\partial_t^2 h  = -\Delta^2 h + 2\nabla V\cdot \nabla h + 2h\Delta V , 
\end{equation}
where the initial conditions are $\partial_th(0)=h(0)=0$ by the conditions $\rho(0) = \rho_{\Psi_0}$ and $\partial_t \rho(0) = - \nabla j_{\Psi_0}$. 
This is a wave-like equation. 
Clearly, $h(t)=0$ for all $t\in[0,T']$ is a solution to \eqref{eqn.wave.density.difference}. Our aim is to show that it is the unique solution. This is rather standard. We give a simple proof using an energy method.  

Define the energy $E(t) = \pscal{h,(1+\Delta^2)h} + \|\partial_t h\|^2_{L^2} \geq 0$. Since $h$ is real-analytic in space and twice continuously differentiable in time, clearly $E$ is continuously differentiable in time. Moreover, $E(0)=0$. 
Then, using that $h$ solves \eqref{eqn.wave.density.difference},
\begin{align*}
    \frac{\rd}{\dt}E(t) 
    & = 2 \Re \pscal{\partial_t h, (1+\Delta^2) h} + 2 \Re \pscal{\partial_t h, \partial_t^2 h}
    \\ & 
    = 2 \Re \pscal{\partial_t h, h + 2\nabla V\cdot \nabla h + 2h\Delta V }
    \\ &
    \leq 
     C\norm{\partial_t h}_{L^2}^2 + C\norm{h}_{L^2}^2
    + C \norm{\nabla V}_{L^\infty}^2 \norm{\nabla h}_{L^2}^2 + C \norm{\Delta V}_{L^\infty}^2 \norm{h}_{L^2}^2
    \leq C E(t). 
\end{align*}
By Grönwall's lemma we conclude that $E(t) \leq e^{Ct}E(0) = 0$
and thus that $h(t)=0$, i.e., that $\rho_{\Psi(t)}=\rho(t)$, for all $t\in [0,T']$ as desired.

\smallskip
\paragraph*{\underline{Uniqueness}:}
Lastly, we verify the uniqueness of $V$. Assume by contradiction that there is another $\tilde V\in C^0([0,T''],\cA_{\sigma''})$ for some $0< \sigma''\leq\sigma'$ and $0<T''\leq T'$ satisfying
\eqref{eq:same_integral}, which here reads as $\int \tilde V = 0$,
for all $t\leq T''$, and such that the solution $\tilde\Psi$ to Schrödinger's equation~\eqref{eq:Schrodinger} with this $\tilde V$ satisfies $\rho_{\tilde\Psi(t)}=\rho(t)$ for $t\in[0,T'']$.

Assume without loss of generality that $\tilde V$ and $V$ start to differ at $t=0$ in the sense that
$
\inf \{ t\geq 0: V(t) \ne \tilde V (t)\} = 0
$
(otherwise, we just start the following argument at the corresponding time).
Then, by \cref{thm.density.is.analytic}, $\rho_{\tilde\Psi}\in C^2([0,T''];\cA_{\sigma'''})$ for $0 < \sigma''' < \sigma''$ and by \cref{prop.force.balance} it satisfies the force balance equation~\eqref{eq:partialt^2rho}. Therefore,
\begin{equation*}
    \partial_t^2 \rho (t)  =\partial_t^2 \rho_{\tilde\Psi(t)} 
		= - \Delta^2\rho_{\tilde\Psi(t)} - 2 K_{\rho_{\tilde\Psi(t)}} \tilde V(t)
        + Q_{\tilde\Psi(t)}
        = - \Delta^2\rho(t) - 2 K_{\rho (t)}\tilde V(t)
        + Q_{\tilde\Psi(t)}
\end{equation*}
implying $\tilde V (t) = V_{\tilde\Psi (t)}(t)$ by the definition in \eqref{eq:VPsi}; note that the fact $\int \tilde V  =0$ allows to invert $K_{\rho (t)}$ here. Thus, $\tilde\Psi$ is also a solution to the nonlinear Schrödinger equation~\eqref{eq:SEnonlinearV}. But then by the uniqueness part of \cref{prop.solve.NLS}, $\tilde\Psi (t) = \Psi (t)$ on some possibly smaller time interval $[0,T''']$ and
$$
\tilde V (t) = V_{\tilde\Psi (t)}(t)
= V_{\Psi (t)}(t)
=V (t)
$$
for $t\in [0,T''']$, which is the desired contradiction.
We conclude that $V$ is unique as claimed. 
\end{proof}

\medskip

\paragraph{\textbf{AI statement.}} Large Language Models (LLMs) were used in the final stage of preparing the manuscript for proofreading. All the results of the paper are solely due to the authors.

%%%%%%%%%%%%%%%%%%%%%%%%%%%%%%%%%%%%%%%%%%%%%%%%%%%%%%%%
%%%%%%%%%%%%%%%%%%%%%%%%%%%%%%%%%%%%%%%%%%%%%%%%%%%%%%%%
\appendix
%%%%%%%%%%%%%%%%%%%%%%%%%%%%%%%%%%%%%%%%%%%%%%%%%%%%%%%%
%%%%%%%%%%%%%%%%%%%%%%%%%%%%%%%%%%%%%%%%%%%%%%%%%%%%%%%%

\section{Notation}\label{sec.notation.appendix}
We collect here some notation used in the paper.

\begin{itemize}[leftmargin=1em]
\setlength\itemsep{0.2em}
\item $\Lambda = [0,1]^D$, the $D$-dimensional unit torus, 
\item $L^2_a(\Lambda^N) = \bigwedge^N L^2(\Lambda)$, the anti-symmetric subspace of $L^2(\Lambda^N)$, 
\item $\hat f(k) = \int_\Lambda f(x) e^{-ikx} \, \rd x$ with $k\in  2\pi\Z^D$, the Fourier transform for $f\in L^1(\Lambda)$, 
\item $\hat \Psi(k_1,\ldots,k_N)  =  \idotsint_{\Lambda^N} \Psi(x_1,\ldots,x_N) e^{-i\sum_{j=1}^N k_j x_j} \,\rd x_1 \!\cdots \rd x_N$ 
with $k_1, \ldots , k_N \in  2\pi\Z^D$, the Fourier transform for $\Psi\in L^1(\Lambda^N)$, 
\item $\cH_V = \sum_{j=1}^N -\Delta_{x_j} + \sum_{j=1}^N V(x_j) + \sum_{1\leq j < k \leq N} V_{ee}(x_j-x_k)$, the $N$-body Hamiltonian,
\item $\cH_0 = \sum_{j=1}^N -\Delta_{x_j} + \sum_{1\leq j < k \leq N} V_{ee}(x_j-x_k)$, the $N$-body Hamiltonian without an external potential,
\item $\cT = \sum_{j=1}^N -\Delta_{x_j}$, the kinetic energy operator, 
\item $\cP = \sum_{j=1}^N -i\nabla_{x_j}$, the total momentum operator,
\item $\cV=\sum_{j=1}^N V(x_j)$, the external potential operator,
\item $\pscal{x} = \sqrt{1+|x|^2}$ for $x\in \R^n$, the Japanese bracket.
\end{itemize}

\printbibliography

\end{document}